\documentclass{egpubl}
\usepackage{eurovis2022}

\SpecialIssueSubmission    

\usepackage[T1]{fontenc}
\usepackage{dfadobe}  

\usepackage{cite}  
\BibtexOrBiblatex
\electronicVersion
\PrintedOrElectronic
\ifpdf \usepackage[pdftex]{graphicx} \pdfcompresslevel=9
\else \usepackage[dvips]{graphicx} \fi

\usepackage{egweblnk}   

\usepackage[dvipsnames]{xcolor}
\usepackage{tikz}
\usepackage{amssymb}
\usepackage{amsmath}
\usepackage{mathtools}
\usepackage{todonotes}
\usepackage{subcaption}
\usepackage[linesnumbered,noend,ruled]{algorithm2e}
\usepackage{svg}
\usepackage{enumitem}
\usepackage{appendix}
\usepackage{amssymb}
\usepackage{pifont}

\DeclareMathOperator{\troot}{root}
\DeclareMathOperator{\edges}{edges}
\DeclareMathOperator{\depth}{depth}
\DeclareMathOperator{\cost}{c}

\newtheorem{theorem}{Theorem}
\newtheorem{lemma}{Lemma}
\newtheorem{definition}{Definition}

\newcommand{\cmark}{\ding{51}}%
\title{Branch Decomposition-Independent Edit Distances for Merge Trees}

\author[F. Wetzels \& H. Leitte \& C. Garth]
{\parbox{\textwidth}{\centering Florian Wetzels\orcid{0000-0002-5526-7138},
        Heike Leitte\orcid{0000-0002-7112-2190}, 
        and Christoph Garth\orcid{0000-0003-1669-8549} 
        }
        \\
{\parbox{\textwidth}{\centering Technische Universität Kaiserslautern}
}
}

\begin{document}

\teaser{
\vspace{-5mm}
 \centering
 \includegraphics[width=.9\textwidth]{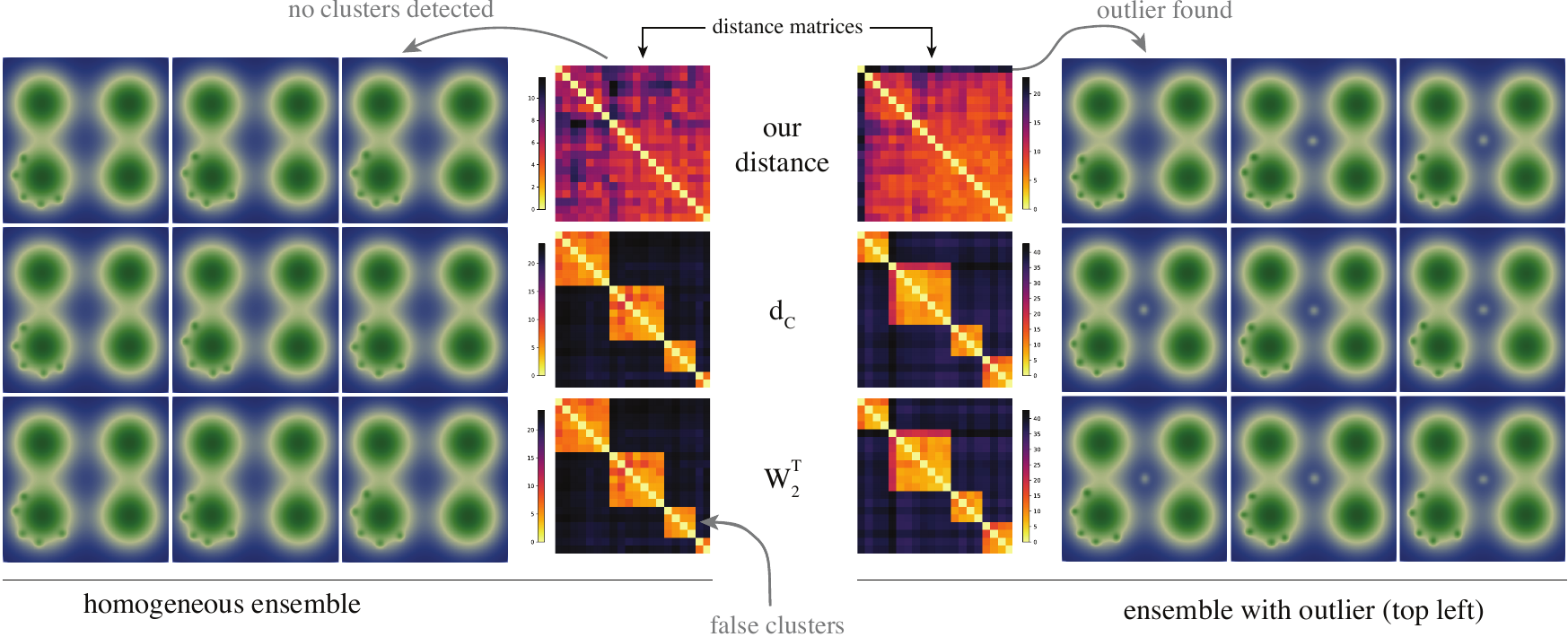}
 \caption{Two scalar field ensembles (illustrated by nine representatives each) with clustermap distance matrices using branch mappings (top) and classic edit mappings (center and bottom). In both ensembles, false clusters are identified by previous approaches, which prevent detection of an outlier in the right ensemble (which is visible but not detectable without previous knowledge of the domain, i.e.\ its existence). Our new distance does not find the false clusters and is therefore able to identify the outlier clearly in the right example.}
 \label{fig:teaser}
}

\maketitle
\begin{abstract}
Edit distances between merge trees of scalar fields have many applications in scientific visualization, such as ensemble analysis, feature tracking or symmetry detection. In this paper, we propose branch mappings, a novel approach to the construction of edit mappings for merge trees. Classic edit mappings match nodes or edges of two trees onto each other, and therefore have to either rely on branch decompositions of both trees or have to use auxiliary node properties to determine a matching. In contrast, branch mappings employ branch properties instead of node similarity information, and are independent of predetermined branch decompositions. Especially for topological features, which are typically based on branch properties, this allows a more intuitive distance measure which is also less susceptible to instabilities from small-scale perturbations. For trees with $\mathcal{O}(n)$ nodes, we describe an $\mathcal{O}(n^4)$ algorithm for computing optimal branch mappings, which is faster than the only other branch decomposition-independent method in the literature by more than a linear factor. Furthermore, we compare the results of our method on synthetic and real-world examples to demonstrate its practicality and utility.


\end{abstract}  
\section{Introduction}

The study or visualization of scalar fields, either acquired through simulation or real-world experiments, is a major avenue of research within scientific visualization.
As computational power increases, so does the size and complexity of scalar fields.
Furthermore, modeling of uncertainty is rapidly becoming increasingly popular as ensemble simulation is becoming more accessible.
Due to this development, efficient measures for the similarity of scalar fields are of high interest, e.g in  clustering tasks or outlier detection in scalar field ensembles, pattern or periodicity recognition in the analysis of time-dependent scalar fields, and symmetry or self-similarity detection within complex scalar fields of high dimension.
Furthermore, the problem of finding mappings between substructures of two given objects is closely related to distance measures, as many algorithms such as edit distances compute both concurrently.
Uses for such algorithms include feature tracking in time-varying scalar fields and identifying common structures or features in ensemble data, which are of importance in many visualization scenarios.

Since working directly on a scalar field becomes computationally infeasible rather quickly, these tasks are often executed on abstract representations.
One prominent example for those is the \emph{contour tree} or a slightly simpler variant of it, the \emph{merge tree}\cite{DBLP:books/daglib/0025666}.
Both concepts originate in the fields of Topological Data Analysis, more specifically Persistent Homology\cite{DBLP:books/daglib/0025666}.
They have a variety of applications (see~\cite{ChristiansSurvey}), either used directly as a (abstract) visualization of the represented scalar field, for generating visualizations, for improving interaction, or simply as a representation amenable to other methods, e.g.\ distance functions.

In this paper, we present a novel distance measure for scalar fields and merge trees that is based on the tree edit distance.
While edit distances and related methods on topological structures have been used before to measure the distance between scalar fields (e.g.~\cite{DBLP:journals/tvcg/SridharamurthyM20,DBLP:journals/corr/abs-2107-07789}), our method generalizes previous approaches to overcome restrictions that prevent these methods from capturing certain semantic properties of the data in the distance.

More precisely, tree edit distances work on mappings between nodes or edges of trees, which are local features, whereas the topological features of interest are global ones such as branches.
Recent methods usually overcome this problem by working on branch decomposition trees or labeling the nodes with properties of their corresponding branch.
Both approaches make the distance depend on a fixed branch decomposition, typically derived by the elder rule\cite{DBLP:books/daglib/0025666}.
These are, however, not stable under small-scale perturbations, leading to semantically imprecise labels of the mapped features and therefore poor mappings.
In this paper, we present a novel method that works on mappings between branches of arbitrary decompositions, making it therefore robust against such instabilities.
We call
them \emph{branch mappings}.
Figure~\ref{fig:teaser} illustrates two example applications where branch mappings show better results than previous approaches.
Furthermore, our method generalizes the notion of tree edit distances as a whole, and we aim to provide an entry point for further research into this field and its application in scientific visualization and computational topology, for which it was developed.
In detail, our contributions are the following: \begin{itemize}
    \item We collect prior work on edit distances in topology-based visualization, and provide a novel categorization based on the desired characteristics discussed above. An overview is given in Table~\ref{tab:categories}.
    \item We describe a \emph{branch decomposition-independent} edit distance that focuses on paths and branches instead of nodes and edges; this stands in contrast to previous methods that are either branch decomposition-dependent or focus on edges.
    \item We illustrate the benefits of this distance in visualization applications by showcasing several synthetic and real-world examples.
\end{itemize}

\noindent
After discussing related work in Section~\ref{section:relatedwork}, we provide basic definitions, an introduction into tree edit distances and the categorization in Section~\ref{section:mergetrees}.
In Section~\ref{section:basics}, we define the concept of \emph{branch mappings}, which are the core of our new distance measure and study it formally in detail.
Application examples and practical comparisons with previous methods can be found in Section~\ref{section:applications}, before we conclude in Section~\ref{section:outlook} and give an outlook on future work.

\section{Related Work}
\label{section:relatedwork}

Similarity or dissimilarity measures based on topological abstractions of scalar fields, or more generally topological descriptors, have been studied in various ways and a variety of techniques have been proposed in increasing frequency.
A recent survey by Yan et al.~\cite{DBLP:journals/cgf/YanMSRNHW21} gives a good overview over these methods.
An introduction into topology based visualization methods in general can be found in the survey by Heine et al.~\cite{ChristiansSurvey}.

Out of these methods, the works of Sridharamurthy et al.~\cite{DBLP:journals/tvcg/SridharamurthyM20,DBLP:journals/corr/abs-2111-04382}, Pont el al.~\cite{DBLP:journals/corr/abs-2107-07789}, Saikia et al.~\cite{DBLP:journals/cgf/SaikiaSW14}, Loh\-fink et al.~\cite{DBLP:journals/cgf/LohfinkWLWG20} and Beketayev et al.~\cite{DBLP:books/daglib/p/BeketayevYMWH14} are most related to ours, as they are all based on edit mappings or similar mappings between contour or merge trees and their branch decomposition trees.
In Section~\ref{section:old_distances}, we discuss these methods in detail in context together with the classification.

Furthermore, our work can be seen as a generalization of tree edit distances without the context of merge trees.
An overview over various versions of tree edit distances and related problems was given by Bille in 2005~\cite{treeEditSurvey}.
Closely related techniques are the \emph{constrained edit distance}~\cite{DBLP:journals/algorithmica/Zhang96} and the \emph{tree alignment distance}~\cite{DBLP:conf/cpm/JiangWZ94} on unordered rooted trees, or the \emph{one-degree edit distance} on ordered trees~\cite{DBLP:journals/ipl/Selkow77}.
We provide a more detailed discussion of these methods in Section~\ref{section:similar_approaches}.

\paragraph*{Topology-based Similarity Measures.}

Apart from edit distances on merge trees or contour trees, other distance measures based in  topological descriptors have been used;
among these, the persistence diagram is frequently used.
Mapping persistence pairs of two scalar fields based on some metric between them is, in principle, largely equivalent to computing edit distances on branch decomposition trees; the only difference is that persistence diagrams do not take into account the nesting of the persistence pairs.
Prominent examples of such distances / mapping methods are the \emph{Wasserstein distance}~\cite{Cohen-Steiner:2010us} and the \emph{bottleneck distance}~\cite{DBLP:journals/dcg/Cohen-SteinerEH07}.
Pont et al.~\cite{DBLP:journals/corr/abs-2107-07789} explicitly compare their edit distance method to these distance measures.
An example for a more advanced method based on persistence diagrams is the work by Rieck et al.~\cite{orderedEditPersistenceHierarchies}.

Other graph-based distance measures work, for example, on Reeb Graphs~\cite{reebgrapheditdistance, reebgraphdistance,localEquivalence,categorifiedreebgraphs}, extremum graphs~\cite{DBLP:conf/apvis/NarayananTN15}, or also on merge and contour trees without specifically using an edit distance~\cite{morozov2013interleaving,DBLP:journals/tvcg/ThomasN11,DBLP:journals/tvcg/YanWMGW20}.

\paragraph*{Topology-based Feature Tracking.}

Edit mappings between topological descriptors can be used to track and visualize features in time-varying scalar fields. Here, they are typically applied to determine correspondence between topologically-characterized features e.g. across time.
The distance measures discussed above can therefore all be used for this task (e.g.~\cite{DBLP:journals/corr/abs-2107-12682} or~\cite{DBLP:journals/corr/abs-2107-07789}).
Furthermore, a further set of tracking techniques is based on topological descriptors directly without relying on edit distance~(e.g.\cite{Oesterling2017,10.1145/997817.997872,Bremer2010}).
Finally, a last class of methods relies on topological descriptors only for feature \emph{identification}, while the actual tracking is then done by other means, e.g. a measure of spatial overlap~\cite{nestedTrackingGraphs,DBLP:journals/cgf/SaikiaW17,DBLP:conf/apvis/ShuGLCLY16,DBLP:journals/tvcg/SchnorrHDKH20}.

\paragraph*{Topology-based Ensemble Visualization.}

Another area where often methods similar to edit distances are used for comparative analysis or similarity assessment is uncertainty visualization via ensembles, specifically the task of finding a representative for the topology of an ensemble of scalar fields. Lohfink et al.~\cite{DBLP:journals/cgf/LohfinkWLWG20} and Pont et al.~\cite{DBLP:journals/corr/abs-2107-07789} derive a representative tree structure for the ensemble from the induced edit mappings. Other examples for contour tree based visualizations of ensemble or uncertain data can be found, for example, in the works of Wu and Zhang~\cite{wu2013contour}, Kraus~\cite{DBLP:conf/imagapp/Kraus10} and Günther et al.~\cite{DBLP:journals/cgf/GuntherST14}. An example that does not work on contour trees but persistence diagrams can be found in~\cite{DBLP:journals/dcg/TurnerMMH14}.

\section{Merge Trees and Tree Edit Distances}
\label{section:basics}

In this section, we start with basic definitions of merge trees and notation for paths and branches in the Subsection~\ref{section:mergetrees}.
Then, we will give an overview over edit distances on general trees that are relevant for edit distances on merge trees in~\ref{section:old_distances}.
After that, we will review previous approaches for mapping based merge tree distances and categorize them in~\ref{section:similar_approaches}. An overview is given in Table~\ref{tab:categories}.

\subsection{Merge Trees}
\label{section:mergetrees}

A detailed introduction into persistent homology and a definition of merge trees can be found in~\cite{DBLP:conf/focs/EdelsbrunnerLZ00} or~\cite{DBLP:conf/ppopp/MorozovW13}.
For this paper, we restrict to an abstract model that captures merge trees of scalar fields of dimension~$>1$.

\paragraph*{Abstract Merge Trees.}
\label{section:abstract_mergetrees}

We now define the abstract model for merge trees, to which we will refer as \emph{abstract} merge trees.
These should just be the class of trees, which can be interpreted as merge trees for some domain of dimension at least~$2$.

\begin{definition}
An unordered, rooted tree $T$ of (in general) arbitrary degree (i.e.\ number of children) with node labels $f:V(T) \rightarrow \mathbb{R}$ is an \emph{Abstract Merge Tree} if the following properties hold:
\begin{itemize}
    \item The root node has degree one, $\deg( \troot (T) ) = 1$
    \item All inner nodes have a degree of at least two,\\ $\deg(v) \neq 1$ for all $v \in V(T)$ with $v \neq \troot (T)$
    \item All nodes have a larger scalar value than their parent node, $f(c) > f(p)$ for all $(c,p) \in E(T)$
\end{itemize}
(Abstract) Merge trees can have arbitrary further edge- or vertex-labels, e.g.\ persistence, volume or the actual segment of an arc.
\end{definition}

\par\medskip
Since the root of an abstract merge tree always has degree one and inner nodes do not, subtrees rooted in an inner node are not abstract merge trees themselves.
Therefore, we identify subtrees by root edges, rather than root nodes:
Formally, for a node $p$ with children $c_1,...,c_k$, the subtree rooted in $(c_1,p)$ is the ``classic'' subtree rooted in $p$ where the subtrees rooted in $c_2,...,c_k$ are removed.
Given an abstract merge tree $T$ with subtree $T'$ rooted in the edge $(c,p)$, we define $T-T'$ to be the tree $T''$, which we obtain by removing all edges and all vertices of $T'$ from $T$ except the root $p$.
If $p$ has degree two in $T$, then we also remove it from $T''$, as otherwise $p$ would be an inner node of degree one in $T''$. With this definition, it holds that
$T'$ and $T''$ are abstract merge trees as well.

From now on, we will often just use the term \emph{merge tree} and it should be clear from the context if this refers to an abstract merge tree or an actual merge tree corresponding to a given scalar field.

\paragraph*{Paths and Branches.}
\label{section:paths}

As for general graphs, a \emph{path} of length $k$ in a merge tree $T$ is a sequence of vertices $p=v_1 ... v_k \in V(T)^k$ with $(v_{i},v_{i-1}) \in E(T)$ for all $2 \leq i \leq k$ (note the strict root-to-leaf direction).
We denote the edges within a path $p$ by $\edges(p) \coloneqq \{(v_{i},v_{i-1}) \bigm| 2 \leq i \leq k\}$.
A \emph{branch} of $T$ is a path that ends in a leaf.
A branch $b=b_1 ... b_k$ is a parent branch of another branch $a=a_1 ... a_\ell$ if $a_1 = b_i$ for some $1 < i < k$.
We also say $a$ is a child branch of $b$.

A set of branches $B=\{B_1,...,B_k\}$ of a merge tree $T$ is called a \emph{Branch Decomposition} of $T$ if $\{\edges(B_1),...,\edges(B_k)\}$ is a partition of $E(T)$.
By $B(T)$ we denote the set of all branch decompositions of $T$.
Every branch decomposition $B \in B(T)$ of an abstract merge tree contains exactly one branch $b \in B$ with $\troot(T) \in b$.
We call this branch the \emph{main branch} of $B$.

The parent-child relations of the branches in a branch decomposition $B \in B(T)$ form a tree structure by themselves.
The tree build from the vertex set $V=B$ and edge set $$E = \{(a,b) \bigm| a,b \in B,\ b \text{ is a parent branch of } a\}$$ is called the \emph{Branch Decomposition Tree} (BDT) of $B$.

Let $T$ be a merge tree with $B \in B(T)$, $T'$ a subtree of $T$ rooted in $(c,p)$ and $T'' = T - T'$.
If there is a branch $b=b_1...b_k$ in $B$ that starts in $(c,p)$, i.e.\ $p=b_1,c=b_2$, then $B$ induces branch decompositions $B' \in B(T')$ and $B'' \in B(T'')$.
$B'$ and $B''$ can be obtained in the obvious way: let $a=a_1...a_\ell$ with $b_1=a_i$ be the parent branch of $b$.
We put $b$ into $B'$ and we put the branch $a_1...a_{i-1}a_{i+1}...a_\ell$ into $B''$ if $p$ has degree two and $a$ otherwise.
All other branches from $B$ are put into $B'$ if they are descendants of $b$ and into $B''$ otherwise.
We denote the branch decompositions $B'$ and $B''$ by $B[T']$ and $B[T'']$.

\subsection{Distance Measures for Trees}
\label{section:old_distances}

Various distance measures for rooted unordered trees have been proposed and investigated in the last decades.
In contrast to the ordered case, they differ significantly in computational complexity, ranging from quadratic algorithms of running time $\mathcal{O}(n \cdot m)$ to NP-hard or even MAX~SNP-hard problems.
We now provide a short overview over often-used methods to provide insight into the framework of tree edit distances in which we aim to place  our new method. Furthermore, this elucidates the relation between the various other edit distances that have been applied to merge trees.

The classic edit distance for rooted, ordered, node-labeled trees was introduced by Tai~\cite{DBLP:journals/jacm/Tai79} and is defined to be the cost-optimal sequence of the following edit operations: 
\begin{itemize}
    \item \textbf{node-deletion}, where all children of the deleted node are appended to its parent,
    \item \textbf{node-insertion}, where a new child is inserted for some node and a subset of its children is made the children of the new node, and
    \item \textbf{node-relabel}, where the label of a node is changed.
\end{itemize}
The cost of the single operations is defined by some cost function $c: ((L \cup \{\bot\}) \times (L \cup \{\bot\})) \rightarrow \mathbb{R}_{\geq 0}$ where~$L$ is the set of node labels and~$\bot$ stands for the null node needed to represent deletions and insertions.
If the basic cost function defines a metric on the set $L \cup \{\bot\}$, then the edit distance between two trees defines a metric on the set of all rooted, unordered, $L$-labeled trees\cite{DBLP:journals/jacm/Tai79}.

Edit sequences induce a mapping between the vertices of the trees, called \emph{edit mappings} or \emph{Tai mappings}~\cite{DBLP:journals/jacm/Tai79}.
These are ancestor-preserving, one-to-one mappings between the two vertex sets.
Their cost is defined as the sum of the relabel costs for all pairs in the relation representing the mapping together with the sum of all insertion/deletion costs for all vertices from both trees that are not present in the relation.
These mappings are induced as all nodes in the first tree that are not changed within an edit sequence as well as those relabeled can be matched to one unique node in the second tree.
Due to this correspondence between edit sequences and mappings, it can be shown that the cost of the optimal sequence and the optimal mapping is equal.
Typical algorithms for edit distances therefore compute the induced mappings and their cost, not the actual edit sequence (which can of course be derived from it).

In contrast to the ordered case, where the general edit distance can be computed in cubic time~\cite{DBLP:conf/icalp/DemaineMRW07,DBLP:conf/icalp/DudekG18}, the problem of computing the edit distance for unordered trees (which we will call $\mathbf{d_e}$) is NP-hard~\cite{DBLP:journals/ipl/ZhangSS92} and even MAX~SNP-hard~\cite{DBLP:journals/ipl/ZhangJ94}, and therefore not arbitrarily close to approximate by a tractable algorithm unless P=NP.
This has led to the development of many simplified versions of the edit distance.
Two very popular approaches are tree alignments~\cite{DBLP:conf/cpm/JiangWZ94} and the constrained edit distance~\cite{DBLP:journals/algorithmica/Zhang96}.

An \emph{alignment} of two
trees is a supertree of both that is obtained by inserting in both trees
until they are
isomorphic.
An alignment induces a mapping and corresponding costs in a similar fashion as before.
The cost of the optimal alignment is equivalent to an edit distance where all insertions have to occur before all deletions.
We denote this distance in the case for general trees as $\mathbf{d_a}$.

An edit mapping between two trees is called
\emph{constrained},
if it also fulfills the property that disjoint subtrees are strictly mapped to disjoint subtrees~\cite{DBLP:journals/algorithmica/Zhang96}.
The \emph{constrained edit distance} is defined to be the cost of an optimal constrained edit mapping between two trees.
We denote this distance in the case for general trees as $\mathbf{d_c}$.

Pont et al.~\cite{DBLP:journals/corr/abs-2107-07789} further restricted the constrained edit mappings specifically for the usecase of BDTs of merge trees. They introduced the constraint that if a node is deleted, i.e.\ mapped to null, the entire subtree rooted in this node has to be deleted, too. This distance is an unordered version of the 1-degree edit distance which has been proposed by Stanley Selkow~\cite{DBLP:journals/ipl/Selkow77} for ordered trees. However, we do not provide a formal proof for this equivalence. We denote the distance on unordered trees by $\mathbf{d_1}$.

All three versions discussed can be computed more efficiently than the general edit distance. The alignment distance is still NP-hard~\cite{DBLP:conf/cpm/JiangWZ94} for trees of arbitrary degree, but can be computed in quadratic time if the trees have bounded degree~\cite{DBLP:conf/cpm/JiangWZ94}, which is a reasonable assumption for merge trees (however, not true in general). Constrained and 1-degree edit distances are computable in polynomial time for both bounded degree trees and arbitrary degree trees, specifically quadratic time for bounded degree~\cite{DBLP:journals/algorithmica/Zhang96,DBLP:journals/corr/abs-2107-07789}, and in time $\mathcal{O}(n_1 \cdot n_2 \cdot (\deg_1+\deg_2) \cdot \log(\deg_1+\deg_2)$ for arbitrary degree. Furthermore, all three distances form a hierarchy in terms of how restrictive they are for the allowed mappings and therefore also their search spaces form a hierarchy of inclusions. For the four distances and arbitrary trees $T_1,T_2$ the following holds if they are applied using the same base metric (derived from a clear hierarchy of the recursions in~\cite{DBLP:conf/cpm/JiangWZ94,DBLP:journals/algorithmica/Zhang96,DBLP:journals/ipl/Selkow77}):
\centerline{$ d_e(T_1,T_2) \leq d_a(T_1,T_2) \leq d_c(T_1,T_2) \leq d_1(T_1,T_2). $}

\subsection{Edit Distances for Contour Trees and Merge Trees}
\label{section:similar_approaches}

Distance metrics between contour trees or merge trees have been used as similarity measures for scalar fields in many varieties, especially through the use of tree edit distances and their corresponding edit mappings.
We will now go through those previous approaches that are most similar to ours, which are those that either explicitly compute structure-preserving mappings between certain parts of the trees or at least internally use them to compute the distance based on a base metric for the mapped objects.

\begin{table}[]
\centering
\scalebox{0.9}{
\begin{tabular}{l|cccccccc}
     & Edit Dist.\    &  BDI   & Global Prop. & Time  \\ \hline
  $d_C$ \cite{DBLP:journals/tvcg/SridharamurthyM20} & \cmark\ ($d_c$) &  -- &  \cmark  &  $\Theta(n^2)$\\
  $W_2^T$ \cite{DBLP:journals/corr/abs-2107-07789} & \cmark\ ($d_1$) &  -- &  \cmark  &  $\Theta(n^2)$\\
  $d_S$ \cite{DBLP:journals/cgf/SaikiaSW14} & \cmark\ ($d_1$) &  -- &  \cmark  &  $\Theta(n^2)$\\
  $d_A$ \cite{DBLP:journals/cgf/LohfinkWLWG20} & \cmark\ ($d_a$) &  \cmark &  --  &  $\Theta(n^2)$\\
  $d_M$ \cite{DBLP:books/daglib/p/BeketayevYMWH14} & -- &  \cmark &  \cmark  &  $\omega(n^5)$\\
  $d_C$ (This paper) & \cmark\ ($d_B$) &  \cmark &  \cmark  &  $\mathcal{O}(n^4)$\\
\end{tabular}
}
\caption{The categorization of mapping based merge tree distances. A distance should be a true edit distance (Edit Distance), should be branch decomposition-independent (BDI) and should use global properties (Global Prop.). The runtimes refer to bounded degree.}
\label{tab:categories}
\end{table}

\paragraph*{Categorization.}

To give a better intuition for the differences between these methods and our method, we categorize them in the following by two further properties (apart from the distinction whether they are explicitly edit distances or not).
The first one is the type of the input objects. These can be actual contour trees or merge trees or branch decompositions of these.
Then, we distinguish the type of labels or properties of these features that are used by the base metric to compute distances of mapped pairs.
Here we distinguish local and global properties as labels of the mapped objects.
Examples could be edge persistence for a local property and branch persistence for a global property.

Sridharamurthy et al.~\cite{DBLP:journals/tvcg/SridharamurthyM20} used the constrained edit distance $\mathbf{d_c}$ to measure the distance between merge trees.
As a base metric for the merge tree nodes, they used $L_\infty$ and overhang costs between the corresponding branches.
In~\cite{DBLP:journals/corr/abs-2111-04382} they adapted this distance to also work on subtrees (similar to eBDGs in~\cite{DBLP:journals/cgf/SaikiaSW14}).
Pont et al.~\cite{DBLP:journals/corr/abs-2107-07789} used a very similar distance to compute geodesics between merge trees and barycenter merge trees as a representative of an ensemble of scalar fields.
They applied the 1-degree edit distance $\mathbf{d_1}$ to \emph{unordered} (see Section~\ref{section:searchspaces} and App.~A in the supp. material) branch decomposition trees using the Wasserstein metric to compare branches.
Another closely related approach is the one by Saikia et al.~\cite{DBLP:journals/cgf/SaikiaSW14}.
They also compute the one-degree edit distance $\mathbf{d_1}$ between branch decomposition graphs of two merge trees to find self-similarities within the scalar fields.
In contrast to Pont et al.\ they considered \emph{ordered} BDTs (see App.~A in the supp. material).
As for the categorization, these three methods fall into the same category: they use fixed branch decompositions as an input (Pont et al.\ and Saikia et al.\ even do so explicitly by working on BDTs) while using global properties (i.e.\ persistence, birth, death, etc.\ of branches) as labels.
We will refer to these distances as $d_C$, $W_2^T$ and $d_S$ and by this mean the mapping method independent on the used base metric (which means we do use the notation $W_2^T$ also for distances that do not use the Wasserstein metric for single branches).
We should note that Pont et al.\ use a normalization step as preprocessing to adapt the edit distance for geodesics and barycenters that modifies the branch labels according to the branch decomposition.
To obtain a meaningful comparison, we only consider their metric without the normalization step, as it then fits better into the hierarchy of different edit distances.

Lohfink et al.~\cite{DBLP:journals/cgf/LohfinkWLWG20} used tree alignments to represent the topology of a scalar field ensemble and to obtain a joint layout of all the contour trees in the ensemble.
The underlying distance is the tree alignment distance $\mathbf{d_a}$.
In contrast to labeling the nodes with their corresponding branches and using a base metric on branches, they labeled the nodes with their unique parent edge and used euclidean distances for different arc properties as the base metric.
This approach gets actual contour trees as an input and uses local properties (edge properties) as labels.
We will refer to this distance as $d_A$, again independent of the used base metric.

A branch decomposition-independent approach that does not fit into the concept of edit distances was introduced by Beketayev et al.~\cite{DBLP:books/daglib/p/BeketayevYMWH14}~($d_M$).
They also compute a cost-optimal mapping between two branch decompositions of two given merge trees.
However, in contrast to the edit distances, the total cost of such a mapping is not the sum of the mapped branches, but the value of the highest cost pair in the mapping.
Furthermore, they also compute the optimal pair of branch decompositions that minimizes the costs.
This falls into the category of methods that get contour trees as inputs (although their core method compares branch decompositions, they actually find the optimal one) and also use global properties (branch persistence etc.) for the base metric.

\label{section:advantages_disadvantes}

\paragraph*{Advantages of using a true edit distance.}

Almost all approaches from the last section use an actual edit distance. A very useful property of edit distances is that they always come with a mapping of different features of the compared objects. In our case they induce a mapping between the nodes, edges or branches of the two input trees. In fact, given a cost function for matching two such features or deleting them, the edit distance value is equal to sum of all mapped or deleted features in the optimal mapping. We usually compute the distance and the mapping simultaneously.

An outlier in the list above is the method of Beketayev et al.~\cite{DBLP:books/daglib/p/BeketayevYMWH14}. They also consider underlying mappings, but since the value of the mapping is only the value of the worst match or largest deletion, the correspondence between mappings and distances is weaker. Moreover, their distance does not correspond to the cost of an edit sequence between the two input trees.

\paragraph*{Using local properties only.}

The alignment distance used in~\cite{DBLP:journals/cgf/LohfinkWLWG20} uses arc properties for the base metric. While  independent of a specific branch decomposition, this approach has a substantial downside, which we will discuss in the following example. If one scalar field has two nested maxima of similar scalar value with the saddle connecting them close to the extreme points, and the other one only has one of the maxima, then the distance between the two trees should be small, as there is only a small split of the maximum from the second tree, which can also be seen as an insertion of one very small maximum branch. However, the persistence of the two arcs corresponding to the leaf nodes are both very small whereas the persistence of the arc leading to the single maximum in the other tree is large. Therefore the edit costs are high, contradicting the (intuitively) desired behaviour. Figure~\ref{fig:badexamplealignment} shows two example trees where this problem appears. A practical example could be two trees where we apply a simplification by persistence to both of them and in one tree the second maximum stays whereas it is removed in the other tree.

\begin{figure}
\centering
\begin{subfigure}[t]{0.49\linewidth}
    \centering
    \resizebox{0.95\linewidth}{!}{
    \begin{tikzpicture}[xscale=0.55,yscale=0.4]
    \node[draw=none,fill=none,circle] at (-4, -3) (dummy) {};
    \node[draw=none,fill=none,circle] at (11, -3) (dummy) {};
    \node[draw=none,fill=none,circle] at (-4, 16) (dummy) {};
    \node[draw=none,fill=none,circle] at (11, 16) (dummy) {};
    
    \node[draw,circle,fill=gray!100,minimum width=0.7cm] at (0, 0) (root_1) {0};
    \node[draw,circle,fill=gray!100,minimum width=0.7cm] at (0, 3) (s1_1) {3};
    \node[draw,circle,fill=red!80,minimum width=0.7cm] at (-2, 12) (m1_1) {12};
    \node[draw,circle,fill=red!80,minimum width=0.7cm] at (2, 12) (m2_1) {12};
    \node[draw,circle,fill=gray!100,minimum width=0.7cm] at (1.333, 9) (s2_1) {9};
    \node[draw,circle,fill=red!80,minimum width=0.7cm] at (0.3, 11) (m3_1) {11};
    \draw[gray,very thick] (root_1) -- (s1_1);
    \draw[gray,very thick] (s1_1) -- (m1_1);
    \draw[gray,very thick] (s1_1) -- (s2_1);
    \draw[gray,very thick] (s2_1) -- (m2_1);
    \draw[gray,very thick] (s2_1) -- (m3_1);
    
    \node[draw,circle,fill=gray!100,minimum width=0.7cm] at (0+7, 0) (root_2) {0};
    \node[draw,circle,fill=gray!100,minimum width=0.7cm] at (0+7, 3) (s1_2) {3};
    \node[draw,circle,fill=red!80,minimum width=0.7cm] at (-2+7, 12) (m1_2) {12};
    \node[draw,circle,fill=red!80,minimum width=0.7cm] at (2+7, 12) (m2_2) {12};
    \draw[gray,very thick] (root_2) -- (s1_2);
    \draw[gray,very thick] (s1_2) -- (m1_2);
    \draw[gray,very thick] (s1_2) -- (m2_2);
    
    \draw[ultra thick,-,dotted,green] (root_1) -- (root_2);
    \draw[ultra thick,-,dotted,green] (s1_1) -- (s1_2);
    \draw[ultra thick,-,dotted,green] (m2_1) to[bend left] (m1_2);
    \draw[ultra thick,-,dotted,green] (m1_1) to[bend left] (m2_2);
    
    \end{tikzpicture}
    }
    \vspace{-10pt}
    \caption{}
    \label{fig:badexamplealignment}
\end{subfigure}
\begin{subfigure}[t]{0.49\linewidth}
    \centering
    \resizebox{0.95\linewidth}{!}{
    \begin{tikzpicture}[xscale=0.55,yscale=0.4]
    \node[draw=none,fill=none,circle] at (-4, -3) (dummy) {};
    \node[draw=none,fill=none,circle] at (11, -3) (dummy) {};
    \node[draw=none,fill=none,circle] at (-4, 16) (dummy) {};
    \node[draw=none,fill=none,circle] at (11, 16) (dummy) {};
    
    \node[draw,circle,fill=gray!100,minimum width=0.7cm] at (-2, 0) (root_1) {0};
    \node[draw,circle,fill=gray!100,minimum width=0.7cm] at (-2, 4) (s1_1) {4};
    \node[draw,circle,fill=red!80,minimum width=0.7cm] at (-2, 13) (m1_1) {13};
    \node[draw,circle,fill=red!80,minimum width=0.7cm] at (2, 12) (m2_1) {12};
    \draw[gray,very thick] (root_1) -- (s1_1);
    \draw[gray,very thick] (s1_1) -- (m1_1);
    \node[draw,circle,fill=gray!100] at (2, 6) (s2_1) {};
    \node[] at (0.5, 7) (s2'_1) {...};
    \draw[gray,very thick] (s2_1) -- (s2'_1);
    \draw[gray,very thick] (s1_1) .. controls (2,4) .. (s2_1);
    \draw[gray,very thick] (s2_1) -- (m2_1);
    \node[draw,circle,fill=gray!100] at (2, 7.5) (s3_1) {};
    \node[] at (0.5, 8.5) (s3'_1) {...};
    \draw[gray,very thick] (s3_1) -- (s3'_1);
    \node[draw,circle,fill=gray!100] at (2, 9) (s4_1) {};
    \node[] at (0.5, 10) (s4'_1) {...};
    \draw[gray,very thick] (s4_1) -- (s4'_1);
    
    \node[draw,circle,fill=gray!100,minimum width=0.7cm] at (2+7, 0) (root_2) {0};
    \node[draw,circle,fill=gray!100,minimum width=0.7cm] at (2+7, 4) (s1_2) {4};
    \node[draw,circle,fill=red!80,minimum width=0.7cm] at (-2+7, 12) (m1_2) {12};
    \node[draw,circle,fill=red!80,minimum width=0.7cm] at (2+7, 13) (m2_2) {13};
    \draw[gray,very thick] (root_2) -- (s1_2);
    \draw[gray,very thick] (s1_2) .. controls (-2+7,4) .. (m1_2);
    \draw[gray,very thick] (s1_2) -- (m2_2);
    \node[draw,circle,fill=gray!100] at (2+7, 6) (s2_2) {};
    \node[] at (0.5+7, 7) (s2'_2) {...};
    \draw[gray,very thick] (s2_2) -- (s2'_2);
    \node[draw,circle,fill=gray!100] at (2+7, 7.5) (s3_2) {};
    \node[] at (0.5+7, 8.5) (s3'_2) {...};
    \draw[gray,very thick] (s3_2) -- (s3'_2);
    \node[draw,circle,fill=gray!100] at (2+7, 9) (s4_2) {};
    \node[] at (0.5+7, 10) (s4'_2) {...};
    \draw[gray,very thick] (s4_2) -- (s4'_2);
    
    \draw[ultra thick,-,dashed,green] (root_1) to[bend right] (root_2);
    \draw[ultra thick,-,dashed,green] (s1_1) to[bend right] (s1_2);
    \draw[ultra thick,-,dashed,green] (m1_1) to[bend left] (m1_2);
    \draw[ultra thick,-,dashed,green] (m2_1) to[bend left] (m2_2);
    
    \draw[ultra thick,-,dashed,green] (s2_1) to[bend right] (s2_2);
    \draw[ultra thick,-,dashed,green] (s3_1) to[bend right] (s3_2);
    \draw[ultra thick,-,dashed,green] (s4_1) to[bend right] (s4_2);
    
    \draw[ultra thick,-,dotted,violet] (root_1) to[out=-20,in=-160] (root_2);
    \draw[ultra thick,-,dotted,violet] (s1_1) to[out=-20,in=-160] (s1_2);
    \draw[ultra thick,-,dotted,violet] (m2_1) to[bend left] (m1_2);
    \draw[ultra thick,-,dotted,violet] (m1_1) to[bend left] (m2_2);
    
    \end{tikzpicture}
    }
    \vspace{-10pt}
    \caption{}
    \label{fig:badexampleBD}
\end{subfigure}
\vspace{-5pt}
\caption{Examples for matching problems. For local properties, in (a) we see two merge trees with the optimal alignment matching in green. Intuitively, the cost should be~$2$, as only one arc of persistence~$2$ is removed. However, the actual cost is $8$ since mapping the right maximum in the left tree (persistence $3$) to the left maximum in the right tree (persistence $9$) has cost of $6$. For a fixed branch decompostion, in (b) two merge trees with two different branch decompositions are shown. The two best matching options are shown in violet and green. Both yield non-optimal results in comparison to the desired distance.}
\end{figure}
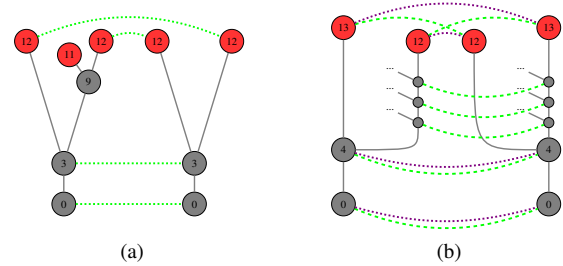

\paragraph*{Using fixed branch decompositions.}

The problem stated in the last section can be fixed by using branch persistences instead of arc persistences. Most methods that we listed above make use of this approach. However, such a metric would depend on a fixed branch decomposition, but these are not unique and two badly chosen branch decompositions of two very similar trees can lead to huge distances. Usually, a very specific branch decomposition is used, namely the one derived by the elder rule, preferring the most persistent branches over small ones. But also this specific decomposition can have crucial instabilities which we will again see in an example. Consider a merge tree with two very persistent branches, one of which is feature-rich and one so simple that it does not have any child branches. Now consider another tree with the same properties. Then, for a low distance, we need to match the feature-rich branches in both trees onto each other as well as the simple branches. If, however, the two branches have similar persistence, the order of the branches can flip between the two trees. Then a good matching between the two branch decompositions has only the following options: either map the child branches of the two feature-rich branches onto each other and take the cost of the persistence difference of the main branches or match the main branches according to their persistence and take the cost of not matching the smaller branches. Figure~\ref{fig:badexampleBD} illustrates this example.

This problem can be illustrated on the following practical example. Consider a
scalar field ensemble consisting of 20 scalar fields as illustrated in Figure~\ref{fig:teaser} (all members shown in Figure~2, supp. material). All fields have four large peaks of similar size with five smaller peaks arising from one of them. The height, size and exact position of all peaks are chosen randomly within small ranges.
This leads to the following scenario: Which of the four main peaks (the one with smaller peaks or one of those without) represents the global maximum differs across the ensemble, and hence also the main branch of the branch decomposition derived by the elder rule differs. All fields are otherwise highly similar, as are their merge trees.
There are no significant differences between the fields apart from small-scale, noise-like instabilities.
However, the merge trees resemble those in Figure~\ref{fig:badexampleBD} and therefore cannot be mapped properly by branch decomposition-dependent distances.

On this dataset, we observe that $d_C$~\cite{DBLP:journals/tvcg/SridharamurthyM20} and $W_2^T$~\cite{DBLP:journals/corr/abs-2107-07789}, using the Wasserstein base metric proposed in~\cite{DBLP:journals/corr/abs-2107-07789}, distinguish four clusters within the members, which does not convey an intuitive understanding of the data.
A clustermap visualization of the edit distances can be seen in Figure~\ref{fig:teaser}.
In contrast to that, our new distance function, which we define in the next section, does not find these false clusters and therefore resembles an intuitive understanding of the dataset much closer.

\section{Branch Mappings}
\label{section:branch_mappings}

All of the methods listed above have one of the mentioned shortcomings. They either use branch properties to overcome the problem of local properties at the expense of the need of a fixed branch decomposition or they are independent of branch decompositions but only use local properties or they combine the positive aspects of both but do not offer the benefits of a true edit distance.

Here, we present a new approach that can be seen as a combination of the ones from~\cite{DBLP:journals/corr/abs-2107-07789}, \cite{DBLP:journals/tvcg/SridharamurthyM20} and~\cite{DBLP:books/daglib/p/BeketayevYMWH14}. We modify the constrained tree edit distance so that it maps branches of two given merge trees onto each other and determine the total cost of such a mapping by taking the sum of distances of all mapped structures while computing the optimal branch decompositions for these costs at the same time. By staying within the notion of edit distances, we keep their global characteristics (in contrast to the more local concept of taking the cost of the worst matched feature pair) and the induced mappings, while still doing the optimization over all possible branch decompositions and mapping the actual features of interest (namely branches, not nodes or edges).

The distance we will introduce is based on the new concept of \emph{branch mappings} that are similar to edit mappings and can be seen as an adaptation of those to branches of trees instead of nodes. Branch mappings can be applied to merge trees with arbitrary branch decompositions as well as merge trees with a fixed, given branch decomposition. They can be computed in $\mathcal{O}(n^2 \cdot m^2)$, which is faster than the only other branch decomposition-independent method (from~\cite{DBLP:books/daglib/p/BeketayevYMWH14}) by more than a linear factor. We achieve this by intertwining the dynamic programming that iterates all branch decompositions with the dynamic programming for the classic edit distance. However, this method of course also has its downsides. We will show that the branch mapping distance, while being a metric on the domain of branch decompositions of merge trees, it is not a metric on the set of merge trees. Furthermore, the computation cost is  higher than for the typical tractable edit distances, which usually can be computed in around $\mathcal{O}(n \cdot m)$.

We start with the core concept underlying our distance, \emph{branch mappings}.
Similar to edit mappings, constrained edit mappings, or alignments, which are mappings between the vertex sets of two given trees with different restrictions to keep certain structures, branch mappings map branches of one tree to branches of another tree while keeping the ancestor relations of the mapped branches.

\begin{definition}
\label{definiton:branch_mappings}
Given two abstract merge trees $T_1,T_2$, a branch mapping between $T_1$ and $T_2$ is a mapping $M \subseteq B_1 \times B_2$ with branch decompositions $B_1 \in B(T_1),B_2 \in B(T_2)$ such that
\begin{enumerate}
    \item If $(a,b),(a',b') \in M$ then $a=a' \Leftrightarrow b=b'$,
    \item $(m,m') \in M$, where $m$ and $m'$ are the main branches of $B_1$ and $B_2$,
    \item If $a$ and $b$ are parent branches of $a'$, $b'$ and $(a',b') \in M$, then $(a,b) \in M$,
    \item If $(a,b),(a',b') \in M$ and $a_1$ is a descendant of $a'_1$ then $b_1$ is a descendant of $b'_1$,
\end{enumerate}
where $a=a_1...a_k,a'=a'_1...a'_{k'} \in B_1$ and $b=b_1...b_\ell,b'=b'_1...b'_{\ell'} \in B_2$ are arbitrary branches of the two trees.
\end{definition}

\noindent
Similar to branch decompositions, we define induced mappings as follows: let $T_1',T_2'$ be subtrees of $T_1,T_2$, $B_1 \in B(T_1),B_2 \in B(T_2)$, $T_1'' = T_1-T_1',T_2'' = T_2-T_2'$ and $M$ be a branch mapping between $B_1$ and $B_2$.
Then we define $M[B_1']$ to be the map
$$\{(a,b) \in M \bigm| a \text{ is a descendant of } m'\}$$
where $B_1'=B_1[T_1']$ and $m'$ is the main branch of $B_1'$.
For $B_1'' = B_1[T_1'']$ we define $M[B_1'']$ to be the map
$$\{(a,b) \in M \bigm| a \notin M[B_1'] \text{ and } a \neq p\} \cup \{(p'',M(p)\}$$
where $B_1'$ and $m'$ are defined as before, $p$ is the parent branch of $m'$ in $B_1$ and $p''$ the corresponding branch in $B_1''$.
We define $M[B_2']$ and $M[B_2'']$ analogously.

\subsection{Branch Mapping Distance}
\label{section:branch_mapping_dist}

In analogy to basing tree edit distances on their corresponding mappings, we now define a distance function on merge trees and their branch decompositions.
For ease of notation, we define for a branch mapping $M$ the set of all edit operations as
$$\overline{M} = M \cup \{(a,\bot) \bigm| \nexists (a,b) \in M\} \cup \{(\bot,b) \bigm| \nexists (a,b) \in M\}.$$
Given a cost function on branches of two merge trees $T_1,T_2$, 
$$ \cost : ((\bigcup_{B_1 \in B(T_1)} B_1 \cup \{\bot\}) \times (\bigcup_{B_2 \in B(T_2)} B_2 \cup \{\bot\})) \rightarrow \mathbb{R},$$
we define the cost of a branch mapping $M$ to be
$$\cost(M) = \sum_{(a,b) \in \overline{M}} \cost(a,b).$$
For a branch decomposition $B \in B(T)$ of some merge tree, we define an empty mapping $M_\bot(B) = \{(b,\bot) \bigm| b \in B\}$ to be the mapping that corresponds to deleting the whole tree $T$. We can now define the distance functions based on the branch mappings.

\begin{definition}
Given two abstract merge trees $T_1,T_2$ with branch decompositions $B_1 \in B(T_1), B_2 \in B(T_2)$, define the distance function
$$ d_B(B_1,B_2) = \min \{\cost(M) \bigm| M \text{ is a branch mapping between } B_1,B_2\}. $$

\noindent
Analogously, for two abstract merge trees $T_1,T_2$, we define
$$ d_B(T_1,T_2) = \min \{d_B(B_1,B_2) \bigm| B_1 \in B(T_1), B_2 \in B(T_2)\}. $$

\noindent
Furthermore, we also include empty trees into the definition of the branch mapping distance:
\begin{eqnarray*}
 d_B(B_1,\bot) &=& \cost(M_\bot(B_1)),d_B(\bot,B_2) = \cost(M_\bot(B_2)),\\
 d_B(T_1,\bot) &=& \min \{d_B(B_1,\bot) \bigm| B_1 \in B(T_1)\}
\end{eqnarray*}
and $d_B(\bot,T_2)$ symmetrically.
\end{definition}

For practical purposes, we want to restrict the base metric for single branches or the branch labels to semantically meaningful ones.
In our case, this means that those properties captured by the labels or costs are actually \emph{branch} properties. Formally, we want a cost function $\cost(a_1...a_k,b_1...b_\ell)$ between two branches to only depend on the start and end points of the branches, i.e.\ $a_1$, $a_k$, $b_1$ and $b_\ell$. The number and position of child branches should be irrelevant. We then call it a \emph{pure} branch distance.

\subsection{Important Properties}
\label{section:branch_mapping_properties}

\paragraph*{Metric properties of the branch mapping distance}

We now show that $d_B$ is a metric on the set of all branch decompositions of abstract merge trees but not on the set of all abstract merge trees.

\begin{theorem}
\label{theorem:metric}
$d_B$ is a metric on the set 
$$\{(T,B) \bigm| T \text{ is an abstract merge tree, } B \in B(T)\},$$
as long as the cost function $c$ on the branch labels is a metric.
\end{theorem}
\begin{proof}
See supplementary material, App.~B.
\end{proof}

\begin{theorem}
$d_B$ is not a metric on the set 
$$\{T \bigm| T \text{ is an abstract merge tree}\}.$$
\end{theorem}
\begin{proof}
If we consider the distance function $d_B$ applied to merge trees with arbitrary branch decompositions, the triangle inequality does (in contrast to the previous result) no longer hold.
The previous argument is not possible in this case since the optimal mapping between $T_1,T_2$ and $T_2,T_3$ can use a different branch decompositions which leads to mapping costs that are not possible with a single branch mapping.
A counterexample can be found in Figure~\ref{fig:notmetric}.
\end{proof}

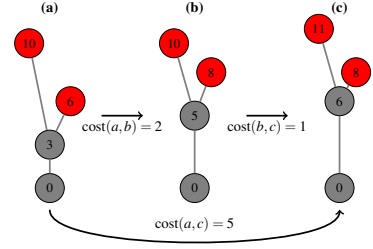
\begin{figure}
    \centering
    \scalebox{0.55}{
    \begin{tikzpicture}[xscale=0.5,yscale=0.35]
    
    \node[draw,circle,fill=gray!100,minimum width=0.7cm] at (0, 0) (root_1) {0};
    \node[draw,circle,fill=gray!100,minimum width=0.7cm] at (0, 3) (s1_1) {3};
    \node[draw,circle,fill=red!100,minimum width=0.7cm] at (-1, 10) (m1_1) {10};
    \node[draw,circle,fill=red!100,minimum width=0.7cm] at (1, 6) (m2_1) {6};
    \draw[gray,very thick] (root_1) -- (s1_1);
    \draw[gray,very thick] (s1_1) -- (m1_1);
    \draw[gray,very thick] (s1_1) -- (m2_1);
    
    \node[draw,circle,fill=gray!100,minimum width=0.7cm] at (0+7, 0) (root_2) {0};
    \node[draw,circle,fill=gray!100,minimum width=0.7cm] at (0+7, 5) (s1_2) {5};
    \node[draw,circle,fill=red!100,minimum width=0.7cm] at (-1+7, 10) (m1_2) {10};
    \node[draw,circle,fill=red!100,minimum width=0.7cm] at (0.8+7, 8) (m2_2) {8};
    \draw[gray,very thick] (root_2) -- (s1_2);
    \draw[gray,very thick] (s1_2) -- (m1_2);
    \draw[gray,very thick] (s1_2) -- (m2_2);
    
    \node[draw,circle,fill=gray!100,minimum width=0.7cm] at (0+14, 0) (root_3) {0};
    \node[draw,circle,fill=gray!100,minimum width=0.7cm] at (0+14, 6) (s1_3) {6};
    \node[draw,circle,fill=red!100,minimum width=0.7cm] at (-1+14, 11) (m1_3) {11};
    \node[draw,circle,fill=red!100,minimum width=0.7cm] at (0.8+14, 8) (m2_3) {8};
    \draw[gray,very thick] (root_3) -- (s1_3);
    \draw[gray,very thick] (s1_3) -- (m1_3);
    \draw[gray,very thick] (s1_3) -- (m2_3);
    
    \draw[very thick,->] (2.5,5) -- (4.5,5);
    \draw[very thick,->] (9.5,5) -- (11.5,5);
    \draw[very thick,->] (0,-1.5) .. controls (-0.1,-4.5) and (14.1,-4.5) .. (14,-1.5);
    \node at (3.5,4.2) {\large $\text{cost}(a,b)=2$};
    \node at (10.5,4.2) {\large $\text{cost}(b,c)=1$};
    \node at (7,-2.5) {\large $\text{cost}(a,c)=5$};
    
    \node at (0,12.4) {\textbf{\large (a)}};
    \node at (7,12.4) {\textbf{\large (b)}};
    \node at (14,12.4) {\textbf{\large (c)}};
    
    \end{tikzpicture}
    }
    \vspace{-5pt}
    \caption{Branch mapping distances for three merge trees using $|b_1-b_2|+|p_1-p_2|$ as base metric for branches with birth values $b_1,b_2$ and persistences $p_1,p_2$. The optimal map between (a) and (b) uses branches $(0,10)$, $(3,6)$ and $(5,8)$, whereas the one for (b) and (c) uses $(0,8)$, $(5,10)$ and $(6,11)$. These maps yield branch mapping distances of $2$ and $1$. The optimal map between (a) and (c) uses branches $(0,10)$, $(3,6)$, $(0,11)$ and $(6,8)$ which yields a total distance of $5$. This violates the metric property as the direct distance between (a) and (c) is greater than the way over (b).}
    \label{fig:notmetric}
\end{figure}

\paragraph*{Recursive Structure}

A core property of the branch mapping distance is the following: as long as the base metric is a pure branch metric, we can decompose a branch mapping into the different subtrees of the given branch decompositions and if we do so, we can do it in the same way for the cost of the mapping.

\begin{lemma} Given two merge trees $T_1,T_2$ with branch decompositions $B_1,B_2$ and a branch mapping $M \subseteq B_1 \times B_2$, for any subtrees $T_1'$ and $T_2'$ of $T_1$ and $T_2$ for which $B_1[T_1']$ and $B_2[T_2']$ exist, and any pure branch distance $c$, it holds that:
\begin{itemize}
    \item $\cost(M) = \cost(M[B[T_1']]) + \cost(M[B[T_1-T_1']])$,
    \item $\cost(M) = \cost(M[B[T_2']]) + \cost(M[B[T_2-T_2']])$.
\end{itemize}
\end{lemma}

\begin{proof} Follows directly from the pureness of $\cost$. \end{proof}

\begin{lemma} Given two merge trees $T_1,T_2$ with roots $v_1,u_1$, let $v_2,u_2$ be the unique children of the two roots and let those have children $v_3,v_4$ and $u_3,u_4$. Let $T_1'$ be the subtree rooted in $(v_2,v_3)$, $T_1''$ rooted in $(v_2,v_4)$, $T_2'$ rooted in $(u_2,u_3)$ and $T_2''$ in $(u_2,u_4)$. Let $M$ be an optimal branch mapping for $T_1$ and $T_2$. Then, for the optimal cost of $M$ it holds that:
\begin{itemize}
    \item $d_B(T_1,T_2) = d_B(T_1',\bot) + d_B(T_1-T_1',T_2)$ or
    \item $d_B(T_1,T_2) = d_B(\bot,T_2') + d_B(T_1,T_2-T_2')$ or
    \item $d_B(T_1,T_2) = d_B(T_1'',\bot) + d_B(T_1-T_1'',T_2)$ or
    \item $d_B(T_1,T_2) = d_B(\bot,T_2'') + d_B(T_1,T_2-T_2'')$ or
    \item $d_B(T_1,T_2) = d_B(T_1',T_2') + d_B(T_1-T_1',T_2-T_2')$ or
    \item $d_B(T_1,T_2) = d_B(T_1'',T_2'') + d_B(T_1-T_1'',T_2-T_2'')$ or
    \item $d_B(T_1,T_2) = d_B(T_1',T_2'') + d_B(T_1-T_1',T_2-T_2'')$ or
    \item $d_B(T_1,T_2) = d_B(T_1'',T_2') + d_B(T_1-T_1'',T_2-T_2')$
\end{itemize}
\label{lemma:rec_inner}
\end{lemma}

\begin{proof} See supplementary material, App.~C.
\end{proof}

\par\medskip
If one of the trees only consists of two nodes, i.e.\ it only has one possible branch decomposition with one branch, we can just omit the corresponding recursions.
Next, we consider the recursion for the case where one tree is empty.

\begin{lemma} Given a merge tree $T_1$ with root $v_1$, let $v_2$ be the unique children of the root and let it have children $v_3,v_4$. Let $T_1'$ be the subtree rooted in $(v_2,v_3)$ and $T_1''$ rooted in $(v_2,v_4)$. Then, for the optimal cost of $M_\bot(T_1)$ it holds that:
\begin{itemize}
    \item $d_B(T_1,\bot) = d_B(T_1',\bot) + d_B(T_1-T_1',\bot)$ or
    \item $d_B(T_1,\bot) = d_B(T_1'',\bot) + d_B(T_1-T_1'',\bot)$,
\end{itemize}
\noindent
and $d_B(\bot,T_2)$ decomposes symmetrically.
\label{lemma:rec_empty}
\end{lemma}

\begin{proof*}
See Appendix~D.
\end{proof*}

\par\medskip
We should note that the previous lemmas only consider binary merge trees.
However, it is easy to adapt for trees of arbitrary degree and we therefore omit this proof for better readability. We continue with the base cases of the recursion.

\begin{lemma} Given two merge trees $T_1 = (\{v_1,v_2\},\{(v_2,v_1)\})$, $T_2 = (\{u_1,u_2\},\{(u_2,u_1)\})$ that only have one branch, the following holds for $d_B$:
\begin{itemize}
    \item $d_B(T_1,\bot) = \cost(v_1v_2,\bot)$,
    \item $d_B(\bot,T_2) = \cost(\bot,u_1u_2)$ and
    \item $d_B(T_1,T_2) = \cost(v_1v_2,u_1u_2)$.
\end{itemize}
\label{lemma:rec_base}
\end{lemma}

\begin{proof}
See supplementary material, App.~E.
\end{proof}

We can use these lemmas to derive a recursive formula for $d_B(T_1,T_2)$ by taking the minimum of the different terms in each case.
This gives a recursion very similar to the ones for the classic edit distances that can be computed with a dynamic programming approach.
In the Section~\ref{section:algo} we will discuss this algorithm.

\subsection{Search Space Comparison to other Methods}
\label{section:searchspaces}

We now want to illustrate more precisely how the branch mapping distance relates to other known distance measures. Since the contour tree alignments from~\cite{DBLP:journals/cgf/LohfinkWLWG20} form some kind of outlier by only considering local properties, we will only look at those methods that use actual branch properties. We therefore try to fit the new method into a hierarchy of branch based methods that we extracted from the papers describing the two closest methods on merge trees~\cite{DBLP:journals/cgf/SaikiaSW14,DBLP:journals/tvcg/SridharamurthyM20,DBLP:journals/corr/abs-2107-07789}. We do this by illustrating the relationship in a schematic drawing of the considered search spaces for mappings between branches/features shown in Figure~\ref{fig:searchspaces}.

\begin{figure}
    \centering
    \begin{tikzpicture}[xscale=0.4, yscale=0.3]
    
    \definecolor{cadmiumgreen}{rgb}{0.0, 0.42, 0.24}
    \draw [draw=cadmiumgreen,ultra thick] (0,-2) rectangle ++(10,7);
    \draw [draw=orange,ultra thick] (2.5,1-1.2) rectangle ++(7.5-0.1,4-0.6);
    \draw [draw=red,ultra thick] (6.5,-1.8) rectangle ++(3.5-0.1,6.6);
    \draw [draw=blue,ultra thick] (6.6,1.1-1.2) rectangle ++(6.5,3.8-0.6);
    \draw [draw=cyan,ultra thick] (6.7,1.2-1.2) rectangle ++(3.1,3.6-0.6);
    
    \node[text=cadmiumgreen] at (1, -1.2) (w2_d) {\tiny$W_2^D$};
    \node[text=orange] at (3.5, 1.5) (d_c) {\tiny$d_C$};
    \node[text=red] at (7.5, -1) (w2_t) {\tiny$W_2^T$};
    \node[text=blue] at (11.5, 1.5) (d_b) {\tiny$d_B$};
    \node[text=cyan] at (8.25, 1.5) (d_b) {\tiny$d_S$};
    
    \end{tikzpicture}
    \caption{A schematic illustration of the search spaces of branch based distance measures for merge trees or scalar fields. $W_2^D$ (Wasserstein Distance between Persistence Diagrams, compare~\cite{DBLP:journals/corr/abs-2107-07789}) maps the branches of a fixed branch decomposition without any structural constraints. $d_C$ and $W_2^T$ restrict the search space by introducing structural constraints based on the merge trees. Since $d_C$ allows for gaps in the mapping and $W_2^T$ is not order preserving, the two models can be seen as orthogonal to each other. $d_S$ also works on a fixed branch decomposition and is strictly more restrictive than the other two. For details, see App.~A (supp. material). $d_B$ then takes the structural constraints from $d_S$ but lifts them to the space of arbitrary branch decompositions, thereby defining a model that  is orthogonal to both $d_C$ and $W_2^T$.}
    \label{fig:searchspaces}
\end{figure}
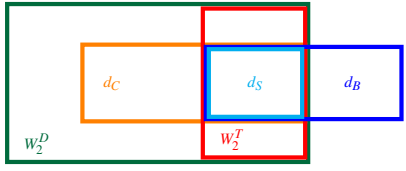

\subsection{Algorithm}
\label{section:algo}

In this section, we will use the recursive structure of optimal branch mappings shown in Lemmas~\ref{lemma:rec_inner}-\ref{lemma:rec_base} to derive a recursive algorithm that computes the branch mapping distance.

First, consider all subtrees $T_1'$, $T_1''$, $T_2'$, $T_2''$, $T_1-T_1'$, $T_1-T_1''$, $T_2-T_2'$ and $T_2-T_2''$ from Lemma~\ref{lemma:rec_inner}.
All of them (in all recursive steps) are abstract merge trees with a root of degree one. We can identify the trees appearing in this recursion by two nodes: the root node of the subtree and its unique child (which does not have to be a child in the original tree $T_1$/$T_2$).
Therefore, we generally switch to identifying trees in this manner in contrast to identifying them just by one node like in the recursions for classic tree edit distances and define an algorithm computing $d_B(n_1,p_1,n_2,p_2)$ where $n_1,n_2$ stand for the current nodes and $p_1,p_2$ stand for the parent nodes of the considered subtrees.
For example, $d_B(T_1,T_2)$ is represented by $d_B(v_1,v_2,u_1,u_2)$, whereas $d_B(T_1-T_1',T_2-T_2'')$ is represented by $d_B(v_4,v_1,u_3,u1)$.
Intuitively, $p_1,p_2$ can be considered as the last matched nodes in the currently tracked branch, which is illustrated in Figure~\ref{fig:recursion_db}.
In App.~F (supp. material), we present pseudo code showing the algorithm for binary trees and discuss how to adapt the algorithm for non-binary trees.

\begin{figure*}
\begin{subfigure}[t]{0.24\linewidth}
    \centering
    \resizebox{0.99\linewidth}{!}{
    \begin{tikzpicture}[xscale=0.6,yscale=0.8]
    
    \node[text=ForestGreen] at (0, 0) (p1) {\tiny$p_1$};
    \node[text=RedOrange] at (0, -1) (n1) {\tiny$n_1$};
    \node[text=ForestGreen] at (-1,-2) (c11) {\tiny$c_{1,1}$};
    \node[text=RedOrange] at (1, -2) (c12) {\tiny$c_{1,2}$};
    \draw [ultra thick,ForestGreen] (p1) -- (n1);
    \draw [ultra thick,ForestGreen] (n1) -- (c11);
    \draw [ultra thick,RedOrange] (n1) -- (c12);
    \draw [very thick,black] (c11) -- (-1.5,-2.4) -- (-1.75,-2.75) -- (-1.25,-2.75) -- (-1.5,-2.4);
    \draw [very thick,black] (c11) -- (-0.5,-2.4) -- (-0.75,-2.75) -- (-0.25,-2.75) -- (-0.5,-2.4);
    \draw [very thick,black] (c12) -- (1.5,-2.4) -- (1.75,-2.75) -- (1.25,-2.75) -- (1.5,-2.4);
    \draw [very thick,black] (c12) -- (0.5,-2.4) -- (0.75,-2.75) -- (0.25,-2.75) -- (0.5,-2.4);
    
    \node[text=ForestGreen] at (0+4, 0) (p2) {\tiny$p_2$};
    \node[text=RedOrange] at (0+4, -1) (n2) {\tiny$n_2$};
    \node[text=ForestGreen] at (-1+4, -2) (c21) {\tiny$c_{2,1}$};
    \node[text=RedOrange] at (1+4, -2) (c22) {\tiny$c_{2,2}$};
    \draw [ultra thick,ForestGreen] (p2) -- (n2);
    \draw [ultra thick,ForestGreen] (n2) -- (c21);
    \draw [ultra thick,RedOrange] (n2) -- (c22);
    \draw [very thick,black] (c21) -- (-1.5+4,-2.4) -- (-1.75+4,-2.75) -- (-1.25+4,-2.75) -- (-1.5+4,-2.4);
    \draw [very thick,black] (c21) -- (-0.5+4,-2.4) -- (-0.75+4,-2.75) -- (-0.25+4,-2.75) -- (-0.5+4,-2.4);
    \draw [very thick,black] (c22) -- (1.5+4,-2.4) -- (1.75+4,-2.75) -- (1.25+4,-2.75) -- (1.5+4,-2.4);
    \draw [very thick,black] (c22) -- (0.5+4,-2.4) -- (0.75+4,-2.75) -- (0.25+4,-2.75) -- (0.5+4,-2.4);
    
    \node[] at (2,-0.3) (label) {\textbf{(a)}};
    
    \end{tikzpicture}
    }
    \caption*{\mbox{\tiny$d_B(c_{1,1},p_1,c_{2,1},p_2)+d_B(c_{1,2},n_1,c_{2,2},n_2)$}}
\end{subfigure}
\begin{subfigure}[t]{0.24\linewidth}
    \centering
    \resizebox{0.99\linewidth}{!}{
    \begin{tikzpicture}[xscale=0.6,yscale=0.8]
    
    \node[text=ForestGreen] at (0, 0) (p1) {\tiny$p_1$};
    \node[text=RedOrange] at (0, -1) (n1) {\tiny$n_1$};
    \node[text=ForestGreen] at (-1,-2) (c11) {\tiny$c_{1,1}$};
    \node[text=RedOrange] at (1, -2) (c12) {\tiny$c_{1,2}$};
    \draw [ultra thick,ForestGreen] (p1) -- (n1);
    \draw [ultra thick,ForestGreen] (n1) -- (c11);
    \draw [ultra thick,RedOrange] (n1) -- (c12);
    \draw [very thick,black] (c11) -- (-1.5,-2.4) -- (-1.75,-2.75) -- (-1.25,-2.75) -- (-1.5,-2.4);
    \draw [very thick,black] (c11) -- (-0.5,-2.4) -- (-0.75,-2.75) -- (-0.25,-2.75) -- (-0.5,-2.4);
    \draw [very thick,black] (c12) -- (1.5,-2.4) -- (1.75,-2.75) -- (1.25,-2.75) -- (1.5,-2.4);
    \draw [very thick,black] (c12) -- (0.5,-2.4) -- (0.75,-2.75) -- (0.25,-2.75) -- (0.5,-2.4);
    
    \node[text=ForestGreen] at (0+4, 0) (p2) {\tiny$p_2$};
    \node[text=RedOrange] at (0+4, -1) (n2) {\tiny$n_2$};
    \node[text=RedOrange] at (-1+4, -2) (c21) {\tiny$c_{2,1}$};
    \node[text=ForestGreen] at (1+4, -2) (c22) {\tiny$c_{2,2}$};
    \draw [ultra thick,ForestGreen] (p2) -- (n2);
    \draw [ultra thick,RedOrange] (n2) -- (c21);
    \draw [ultra thick,ForestGreen] (n2) -- (c22);
    \draw [very thick,black] (c21) -- (-1.5+4,-2.4) -- (-1.75+4,-2.75) -- (-1.25+4,-2.75) -- (-1.5+4,-2.4);
    \draw [very thick,black] (c21) -- (-0.5+4,-2.4) -- (-0.75+4,-2.75) -- (-0.25+4,-2.75) -- (-0.5+4,-2.4);
    \draw [very thick,black] (c22) -- (1.5+4,-2.4) -- (1.75+4,-2.75) -- (1.25+4,-2.75) -- (1.5+4,-2.4);
    \draw [very thick,black] (c22) -- (0.5+4,-2.4) -- (0.75+4,-2.75) -- (0.25+4,-2.75) -- (0.5+4,-2.4);
    
    \node[] at (2,-0.3) (label) {\textbf{(b)}};
    
    \end{tikzpicture}
    }
    \caption*{\mbox{\tiny$d_B(c_{1,1},p_1,c_{2,2},p_2)+d_B(c_{1,2},n_1,c_{2,1},n_2)$}}
\end{subfigure}
\begin{subfigure}[t]{0.24\linewidth}
    \centering
    \resizebox{0.99\linewidth}{!}{
    \begin{tikzpicture}[xscale=0.6,yscale=0.8]
    
    \node[text=ForestGreen] at (0, 0) (p1) {\tiny$p_1$};
    \node[text=gray!80] at (0, -1) (n1) {\tiny$n_1$};
    \node[text=ForestGreen] at (-1,-2) (c11) {\tiny$c_{1,1}$};
    \node[text=gray!80] at (1, -2) (c12) {\tiny$c_{1,2}$};
    \draw [ultra thick,ForestGreen] (p1) -- (n1);
    \draw [ultra thick,ForestGreen] (n1) -- (c11);
    \draw [ultra thick,gray!40] (n1) -- (c12);
    \draw [very thick,black] (c11) -- (-1.5,-2.4) -- (-1.75,-2.75) -- (-1.25,-2.75) -- (-1.5,-2.4);
    \draw [very thick,black] (c11) -- (-0.5,-2.4) -- (-0.75,-2.75) -- (-0.25,-2.75) -- (-0.5,-2.4);
    \draw [very thick,gray!40] (c12) -- (1.5,-2.4) -- (1.75,-2.75) -- (1.25,-2.75) -- (1.5,-2.4);
    \draw [very thick,gray!40] (c12) -- (0.5,-2.4) -- (0.75,-2.75) -- (0.25,-2.75) -- (0.5,-2.4);
    
    \node[text=ForestGreen] at (0+4, 0) (p2) {\tiny$p_2$};
    \node[text=ForestGreen] at (0+4, -1) (n2) {\tiny$n_2$};
    \node[text=black] at (-1+4, -2) (c21) {\tiny$c_{2,1}$};
    \node[text=black] at (1+4, -2) (c22) {\tiny$c_{2,2}$};
    \draw [ultra thick,ForestGreen] (p2) -- (n2);
    \draw [ultra thick,black] (n2) -- (c21);
    \draw [ultra thick,black] (n2) -- (c22);
    \draw [very thick,black] (c21) -- (-1.5+4,-2.4) -- (-1.75+4,-2.75) -- (-1.25+4,-2.75) -- (-1.5+4,-2.4);
    \draw [very thick,black] (c21) -- (-0.5+4,-2.4) -- (-0.75+4,-2.75) -- (-0.25+4,-2.75) -- (-0.5+4,-2.4);
    \draw [very thick,black] (c22) -- (1.5+4,-2.4) -- (1.75+4,-2.75) -- (1.25+4,-2.75) -- (1.5+4,-2.4);
    \draw [very thick,black] (c22) -- (0.5+4,-2.4) -- (0.75+4,-2.75) -- (0.25+4,-2.75) -- (0.5+4,-2.4);
    
    \node[] at (2,-0.3) (label) {\textbf{(c)}};
    
    \end{tikzpicture}
    }
    \caption*{\mbox{\tiny$d_B(c_{1,1},p_1,n_1,p_2)+d_B(c_{1,2},n_1,\bot,\bot)$}}
\end{subfigure}
\begin{subfigure}[t]{0.24\linewidth}
    \centering
    \resizebox{0.99\linewidth}{!}{
    \begin{tikzpicture}[xscale=0.6,yscale=0.8]
    
    \node[text=ForestGreen] at (0, 0) (p1) {\tiny$p_1$};
    \node[text=gray!80] at (0, -1) (n1) {\tiny$n_1$};
    \node[text=gray!80] at (-1,-2) (c11) {\tiny$c_{1,1}$};
    \node[text=ForestGreen] at (1, -2) (c12) {\tiny$c_{1,2}$};
    \draw [ultra thick,ForestGreen] (p1) -- (n1);
    \draw [ultra thick,gray!40] (n1) -- (c11);
    \draw [ultra thick,ForestGreen] (n1) -- (c12);
    \draw [very thick,gray!40] (c11) -- (-1.5,-2.4) -- (-1.75,-2.75) -- (-1.25,-2.75) -- (-1.5,-2.4);
    \draw [very thick,gray!40] (c11) -- (-0.5,-2.4) -- (-0.75,-2.75) -- (-0.25,-2.75) -- (-0.5,-2.4);
    \draw [very thick,black] (c12) -- (1.5,-2.4) -- (1.75,-2.75) -- (1.25,-2.75) -- (1.5,-2.4);
    \draw [very thick,black] (c12) -- (0.5,-2.4) -- (0.75,-2.75) -- (0.25,-2.75) -- (0.5,-2.4);
    
    \node[text=ForestGreen] at (0+4, 0) (p2) {\tiny$p_2$};
    \node[text=ForestGreen] at (0+4, -1) (n2) {\tiny$n_2$};
    \node[text=black] at (-1+4, -2) (c21) {\tiny$c_{2,1}$};
    \node[text=black] at (1+4, -2) (c22) {\tiny$c_{2,2}$};
    \draw [ultra thick,ForestGreen] (p2) -- (n2);
    \draw [ultra thick,black] (n2) -- (c21);
    \draw [ultra thick,black] (n2) -- (c22);
    \draw [very thick,black] (c21) -- (-1.5+4,-2.4) -- (-1.75+4,-2.75) -- (-1.25+4,-2.75) -- (-1.5+4,-2.4);
    \draw [very thick,black] (c21) -- (-0.5+4,-2.4) -- (-0.75+4,-2.75) -- (-0.25+4,-2.75) -- (-0.5+4,-2.4);
    \draw [very thick,black] (c22) -- (1.5+4,-2.4) -- (1.75+4,-2.75) -- (1.25+4,-2.75) -- (1.5+4,-2.4);
    \draw [very thick,black] (c22) -- (0.5+4,-2.4) -- (0.75+4,-2.75) -- (0.25+4,-2.75) -- (0.5+4,-2.4);
    
    \node[] at (2,-0.3) (label) {\textbf{(d)}};
    
    \end{tikzpicture}
    }
    \caption*{\mbox{\tiny$d_B(c_{1,2},p_1,n_1,p_2)+d_B(c_{1,1},n_1,\bot,\bot)$}}
\end{subfigure}
\caption{Exemplary illustration of recursive structure of $d_B(n_1,p_1,n_2,p_2)$ in Algorithm~1 in App.~F (supp. material). In (a) we continue main branch tracking with the left subtree in both trees and match the right subtrees onto each other, whereas in (b) we continue main branch tracking with the left subtree in $T_1$ and the right subtree in $T_2$ and match the other two subtrees onto each other. In (c), we continue main branch tracking with the left subtree in $T_1$ and delete the right subtree. In $T_2$, we do nothing. In (d), we continue main branch tracking with the right subtree in $T_1$, delete the left subtree and do nothing in $T_2$.}
\label{fig:recursion_db}
\end{figure*}
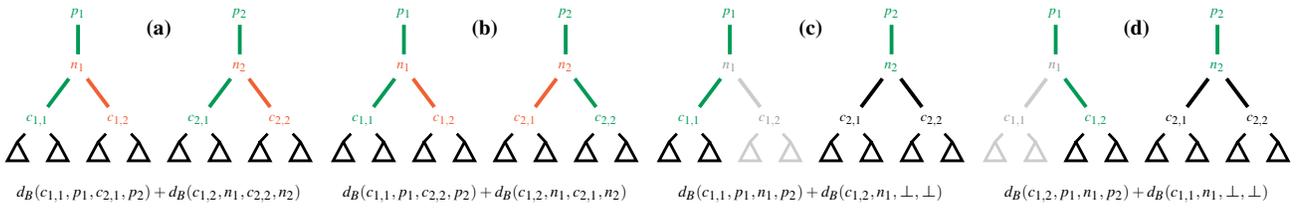

\paragraph*{Running time.} An obvious bound on the running time of Algorithm~1 in App.~F (supp. material) using memoization is $\mathcal{O}(|T_1|^2 \cdot |T_2|^2)$, since the number of subproblems is constant, they are always smaller than the current one and there are only $|T_1|^2 \cdot |T_2|^2$ many 4-tuples of nodes. However, this can be tightened to $\mathcal{O}(|T_1|\cdot\depth(T_1) \cdot |T_2|\cdot\depth(T_2))$, since we know that $p_1,p_2$ are always ancestors of $n_1,n_2$. A short discussion on the runtime for non-bounded degree can be found in App.~F (supp. material).

These bounds of course only hold if the base metric is computable in constant time. Typical metrics like difference in persistence or bottleneck, Wasserstein and $L_\infty$-distances are all constant time operations. For branch mappings, it is also important that these are pure branch metrics, since then they only depend on the start and end point of a branch, which are always available when evaluating the base metric (through $n_1,p_1$ or $n_2,p_2$).

\paragraph*{Variants.} There are a few adaptations that can be applied and should be noted here. First, saddle swap instabilities are a key problem of all tractable edit distances on merge trees and are of course still present in the branch mapping distance. To handle such instabilities, \cite{DBLP:journals/tvcg/SridharamurthyM20}~and~\cite{DBLP:journals/corr/abs-2107-07789} use a preprocessing step to collapse small edges (through an $\epsilon$-parameter) such that branches in a parent-child relation that are probable to be influenced by such instabilities become siblings instead. Of course, this preprocessing is also applicable to our method and should yield the same improvements, however, we did not implement or test it. Furthermore, it is possible to adapt our algorithm to not try all branch decompositions and use a fixed one instead. The computed distance and mapping then becomes equivalent to the one from~\cite{DBLP:journals/cgf/SaikiaSW14}.

It is also possible to not just sum up the costs of all pairs in the optimal matchings but to sum up the squared pair costs and then take the square root of the sum. This approach was used in~\cite{DBLP:journals/corr/abs-2107-07789} and makes the resulting distance directly comparable to the classic Wasserstein distance for persistence diagrams. This is of course also possible with our method or $d_C$ and we used this variant for the motivational examples in Section~\ref{fig:teaser}.

Another property to note is that branch mappings induce vertex mappings in the sense of constrained (or even 1-degree) edit mappings. Therefore, it is possible to construct an alignment of two trees by computing the optimal branch mapping. This alignment is a supertree of the two input trees and can therefore again be aligned. This gives us the possibility to use an iterated heuristic for an alignment of a set of merge trees (based on the branch mapping distance) as it has been done for contour trees in~\cite{DBLP:journals/cgf/LohfinkWLWG20}. We plan to implement such a method in future work.

\section{Experiments}
\label{section:applications}

In the following, we demonstrate the utility of our technique as a basis for typical
tasks
in visualization, and compare it against other, similar techniques. The basis for these experiments is a straightforward Python implementation of the algorithm given in Section~\ref{section:algo} and the constrained edit distance. Merge trees are computed using TTK~\cite{DBLP:journals/tvcg/TiernyFLGM18}, and ParaView~\cite{squillacote2007paraview} is used for result visualization. The three experiments described in the following correspond to three different uses of distance functions between scalar fields as found in many visualization applications: outlier identification, periodicity detection, and feature tracking. App.~H (supp. material) discusses the results on noisy versions of the same data.

For benchmarks, we use a C++ implementation which is publicly available on Github as well as the easier to read Python
version~\cite{repository}.
Computation of individual distances for the here used
trees with up to 80 vertices took only a few milliseconds. On noisier trees with sizes between 100 and 400, times went up to the range of seconds. For more details, see App.~G (supp. material).

\subsection{Outlier Detection}
\label{section:task_outlier}

Outlier detection is a major concern in ensemble analysis, and the distance between merge trees of two scalar ensemble members can be used towards this. We consider a synthetic example dataset shown in Figure~\ref{fig:teaser}. This dataset is based on the example ensemble from Section~\ref{section:advantages_disadvantes}, with one modification: we add a fifth peak in the middle of the four large peaks, but create an outlier by omitting this peak in a single ensemble member. Outlier detection then proceeds by computing the pairwise merge tree distances for all pairs of trees, using the squared sum of the mapping costs.

Using clustermaps to visualize the distance matrices, it is easy to see that the branch mapping distance ($d_B$) distinguishes the outlier clearly, whereas the constrained edit distance $d_C$ and the Wasserstein distance $W_2^T$ do not. For the latter two, the cost of the missing branch in the outlier ensemble member is concealed by the cluster effects when working with fixed branch decompositions. A visualization of the entire  ensemble and further results obtained using additional parameter choice can be found in the supplementary material (App.~H). Note that this dataset was constructed specifically to illustrate the problem of fixed branch decompositions, as discussed in Section~\ref{section:advantages_disadvantes}. We make both datasets available publicly with our implementation~\cite{repository}.

\subsection{Periodicity Detection}
\label{section:task_periodic}

Periodicity detection is an often encountered problem in time-varying data; we here consider a time-varying dataset consisting of 1001 time steps of a $400\times 50$ grid representing the scalar velocity magnitude of the flow around a cylinder that forms a periodic Kármán vortex street. It was simulated by Weinkauf~\cite{weinkauf10c} using the \emph{GerrisFlowSolver}~\cite{gerrisflowsolver}.

Reproducing the experiment of Sridharamurthy et al.~\cite{DBLP:journals/tvcg/SridharamurthyM20} with our distance measure, we compute the merge tree distance matrix for all pairs of time steps and visualize this as a heatmap (cf.~\autoref{fig:heatmap_weinkauf}). Here, the branch mapping distance matrix exhibits exactly the same periodicity pattern as the constrained edit distance matrix  (Fig.~13 in \cite{DBLP:journals/tvcg/SridharamurthyM20}), giving a period/half-period of 75/37 time steps, respectively (see App.~H (supp. material) for a direct comparison).

\begin{figure*}
    \centering
    \scalebox{1}[0.95]{
    \includegraphics[width=.85\linewidth]{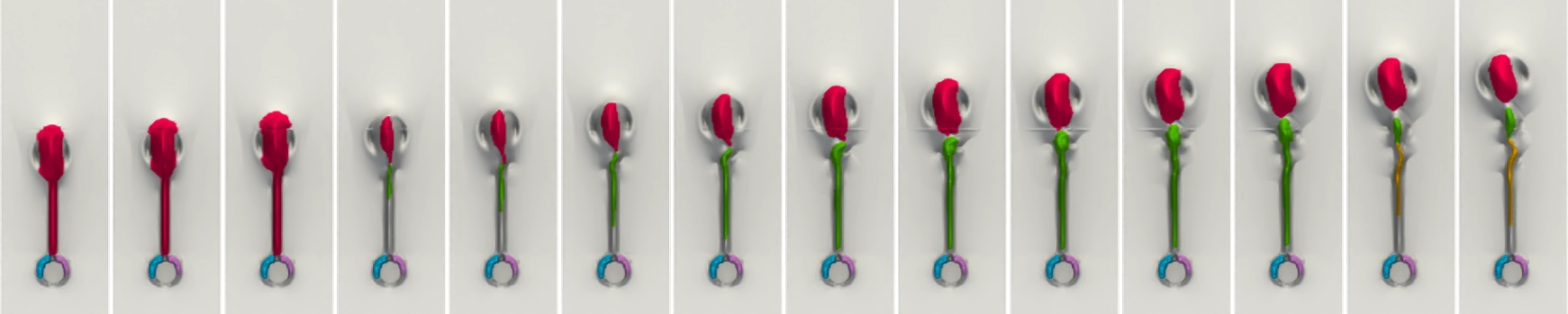}
    }
    \includegraphics[width=.85\linewidth]{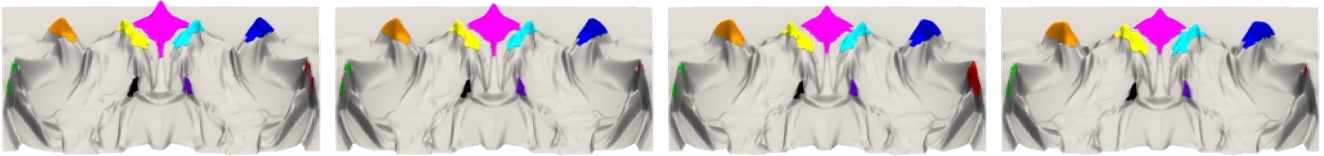}
    \caption{Feature Tracking for the Heated Cylinder (top) and ion density (bottom) datasets. We rendered all maxima matched by the optimal branch mapping in consistent colors over all time points. The matching for the heated cylinder shows how the rising plume splits up over time into two and then three distinct maxima while the two peaks around the pole stay constant. The rendering for the ion density dataset shows that the branch mapping represents the intuitively correct mapping of the corresponding peaks, which is for example not true for a mapping obtained by Wasserstein distances on persistence diagrams according to~\cite{DBLP:journals/corr/abs-2107-07789}.
    \label{fig:tracking}}
\end{figure*}

\subsection{Feature Tracking over Time}
\label{section:task_tracking}

As a last experiment, we utilize our distance metric for feature tracking within a time series of scalar fields. Here, we consider two different time series. The \textbf{SciVis contest 2008 dataset}~\cite{scivisContest2008} describes development of ion density during universe formation, and was also used as in a tracking case study in prior work~\cite{DBLP:journals/corr/abs-2107-07789}. Second, the \textbf{heated cylinder dataset} describes an ensemble of flows around a heated pole in a fluid. This dataset was also used in prior work on ensemble analysis~\cite{DBLP:journals/cgf/LohfinkWLWG20}, and we make it available publicly~\cite{repository}. We again consider velocity magnitude as the variable of interest for tracking.

For both datasets, we provide a proof-of-concept for the proposed branch mapping distance by replicating previous results. We are able to extract semantically meaningful matchings of features, comparable to those found earlier~\cite{DBLP:journals/cgf/LohfinkWLWG20,DBLP:journals/corr/abs-2107-07789} from the branch mappings for both datasets (see App.~H (supp. material) for a direct comparison). Visualizations of the matched features are shown in Figure~\ref{fig:tracking}, indicating that our distance measure is useful for tracking applications in topology-based visualization.

\section{Conclusion and Outlook}
\label{section:outlook}

In this paper, we presented a novel variant of edit distances that is tailored specifically to matching branches of merge trees.
We provided a formal definition and analysis of metric properties as well as an algorithm and implementation.
We showed that it is as expressive as previous approaches on practical datasets and that it can improve the quality of the matching and distance measure significantly on datasets containing specific structures.
We also discussed limitations of the new method, them being its higher complexity and the fact that it is not a metric on merge trees, which does, however, not influence its utility for the provided tasks.
In the context of our classification, we have described an edit distance-based method that works on merge trees and uses global properties.
However, this method does so in a very specific and restricted way, which we now discuss together with options for less restrictive methods and future work.

\begin{figure}[b!]
    \centering
    \includegraphics[height=0.4\linewidth]{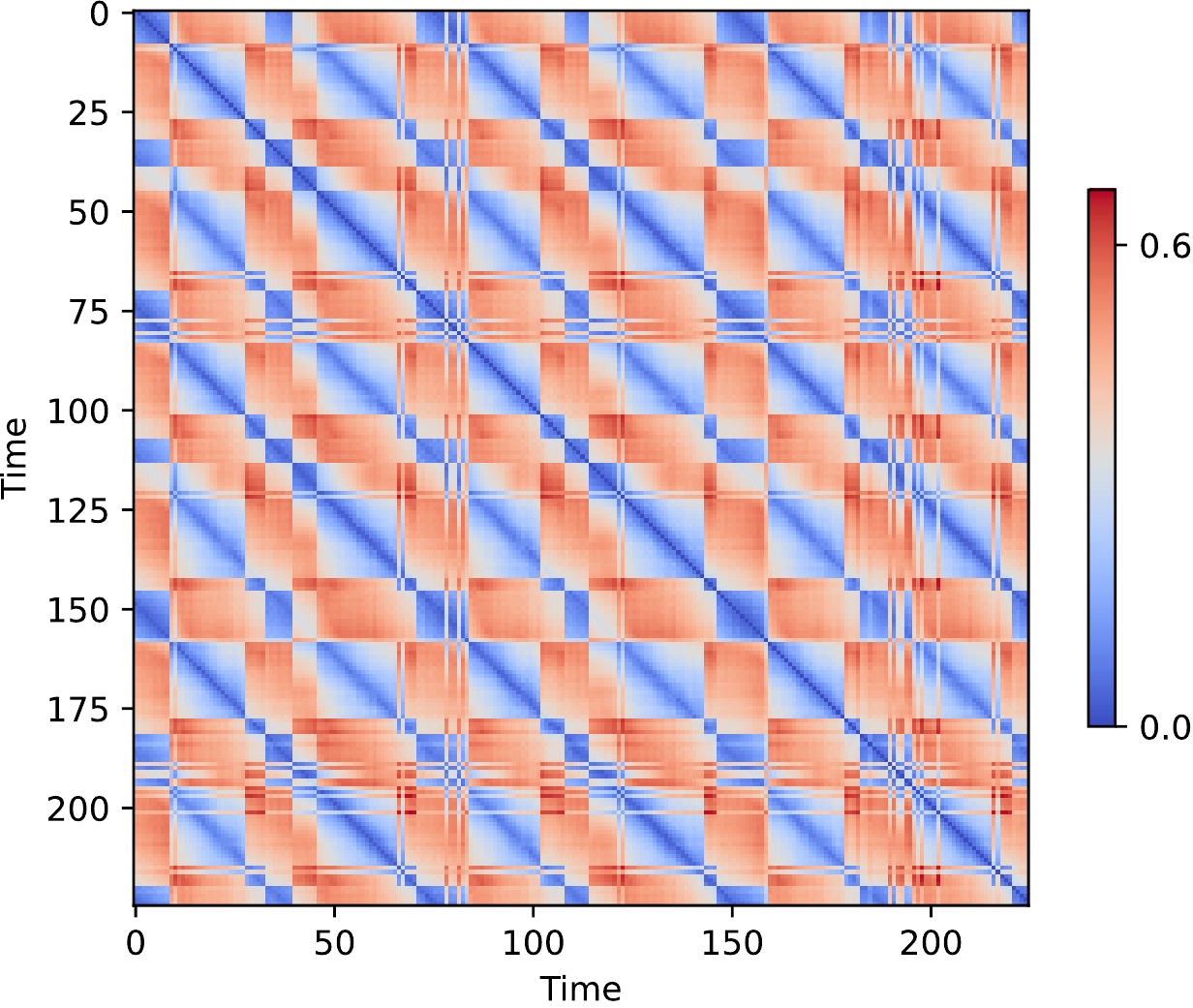}
    \caption{A heatmap of the distance matrix for all pairs of the first 225 time steps of the vortex street dataset (full matrix cf. supp. material, App.~H), clearly indicating a period/half-period 75/37.}
    \label{fig:heatmap_weinkauf}
\end{figure}

\paragraph*{Other Branch Decomposition-Independent Methods.}
A restriction we used was to only map \emph{branches}.
We could generalize this to arbitrary \emph{paths} of a merge tree.
As branches are paths, this would just extend the concept of branch mappings by again increasing the search space.
Such \emph{path mappings} would actually be computable in the same time and space bounds.
In future work, we want to investigate the formal foundations of path mappings.

Another restriction is that the new distance corresponds to the 1-degree edit distance in the classic setting.
We also want to consider adaptions of constrained edit mappings, alignments, the general edit distance or even unrooted versions of those.
For these options, we have to find suitable definitions and study their complexity to see which of these options are applicable in practice.

\paragraph*{Other Applications.}
Furthermore, as hinted in Section~\ref{section:algo}, we plan to integrate branch mappings (and maybe path mappings) into the alignment framework introduced in~\cite{DBLP:journals/cgf/LohfinkWLWG20} (which is definitely possible) and investigate the possibility of combining it with the geodesic/barycenter methods introduced in~\cite{DBLP:journals/corr/abs-2107-07789}.
The latter seems more likely to be possible for certain variants of path mappings which we believe to be a metric, since the normalization step from~\cite{DBLP:journals/corr/abs-2107-07789} seems hard to apply to non-fixed branch decompositions and barycenter computation could be hindered by the missing triangle inequality.
A very typical application of distance measures for scalar fields that we have not studied in this paper is symmetry or self-similarity detection.
As the memoization table computed by the algorithm includes distances for all subtrees of both input trees, it is possible to use it to find substructures of high similarity by applying methods similar to those from~\cite{DBLP:journals/cgf/SaikiaSW14}.
We aim to implement them in future work, too.
Furthermore, we plan to investigate if similar distances can be defined for contour trees instead of merge trees.
Hurdles to overcome here would be, for example, that paths in contour trees are not necessarily monotone, and by restricting to monotone branches, the impact of saddle-swap instabilities increase significantly.

\paragraph*{Parallel Implementation.}
For a more efficient implementation, we want to look deeper into optimization options and also develop a parallel implementation of our algorithm.
This could lead to a significantly improved performance, as the branching factor in the $\mathcal{O}(n^4)$-sized search space is high, which (intuitively) should lead to high parallelization potential.

\bibliographystyle{eg-alpha-doi}
\bibliography{ms}

\newcommand{\etalchar}[1]{$^{#1}$}
\begin{thebibliography}{\uppercase{CSEHM10}}

\bibitem[BDFL16]{reebgrapheditdistance}
\textsc{Bauer U., Di~Fabio B., Landi C.}:
\newblock An edit distance for {R}eeb graphs.
\newblock In \emph{Proceedings of the Eurographics 2016 Workshop on 3D Object
  Retrieval} (Goslar Germany, Germany, 2016), 3DOR '16, Eurographics
  Association, pp.~27--34.
\newblock \href {https://doi.org/10.2312/3dor.20161084}
  {\path{doi:10.2312/3dor.20161084}}.

\bibitem[BGW13]{reebgraphdistance}
\textsc{Bauer U., Ge X., Wang Y.}:
\newblock Measuring distance between {R}eeb graphs.
\newblock \emph{Proceedings of the Annual Symposium on Computational Geometry}
  (07 2013).
\newblock \href {https://doi.org/10.1145/2582112.2582169}
  {\path{doi:10.1145/2582112.2582169}}.

\bibitem[Bil05]{treeEditSurvey}
\textsc{Bille P.}:
\newblock A survey on tree edit distance and related problems.
\newblock \emph{Theoretical Computer Science 337}, 1-3 (2005), 217--239.
\newblock \href {https://doi.org/10.1016/j.tcs.2004.12.030}
  {\path{doi:10.1016/j.tcs.2004.12.030}}.

\bibitem[BWP{\etalchar{*}}10]{Bremer2010}
\textsc{Bremer P.-T., Weber G., Pascucci V., Day M., Bell J.}:
\newblock Analyzing and tracking burning structures in lean premixed hydrogen
  flames.
\newblock \emph{{IEEE} Transactions on Visualization and Computer Graphics 16},
  2 (Mar. 2010), 248--260.
\newblock \href {https://doi.org/10.1109/tvcg.2009.69}
  {\path{doi:10.1109/tvcg.2009.69}}.

\bibitem[BYM{\etalchar{*}}14]{DBLP:books/daglib/p/BeketayevYMWH14}
\textsc{Beketayev K., Yeliussizov D., Morozov D., Weber G.~H., Hamann B.}:
\newblock Measuring the distance between merge trees.
\newblock In \emph{Topological Methods in Data Analysis and Visualization III,
  Theory, Algorithms, and Applications}, Bremer P., Hotz I., Pascucci V.,
  Peikert R., (Eds.). Springer, 2014, pp.~151--165.
\newblock \href {https://doi.org/10.1007/978-3-319-04099-8_10}
  {\path{doi:10.1007/978-3-319-04099-8_10}}.

\bibitem[CEH07]{DBLP:journals/dcg/Cohen-SteinerEH07}
\textsc{Cohen{-}Steiner D., Edelsbrunner H., Harer J.}:
\newblock Stability of persistence diagrams.
\newblock \emph{Discret. Comput. Geom. 37}, 1 (2007), 103--120.
\newblock \href {https://doi.org/10.1007/s00454-006-1276-5}
  {\path{doi:10.1007/s00454-006-1276-5}}.

\bibitem[CO17]{localEquivalence}
\textsc{Carri{\`{e}}re M., Oudot S.}:
\newblock Local equivalence and intrinsic metrics between reeb graphs.
\newblock In \emph{33rd International Symposium on Computational Geometry, SoCG
  2017, July 4-7, 2017, Brisbane, Australia} (2017), Aronov B., Katz M.~J.,
  (Eds.), vol.~77 of \emph{LIPIcs}, Schloss Dagstuhl - Leibniz-Zentrum
  f{\"{u}}r Informatik, pp.~25:1--25:15.
\newblock \href {https://doi.org/10.4230/LIPIcs.SoCG.2017.25}
  {\path{doi:10.4230/LIPIcs.SoCG.2017.25}}.

\bibitem[CSEHM10]{Cohen-Steiner:2010us}
\textsc{Cohen-Steiner D., Edelsbrunner H., Harer J., Mileyko Y.}:
\newblock Lipschitz functions have lp-stable persistence.
\newblock \emph{Foundations of Computational Mathematics 10}, 2 (2010),
  127--139.
\newblock \href {https://doi.org/10.1007/s10208-010-9060-6}
  {\path{doi:10.1007/s10208-010-9060-6}}.

\bibitem[DG18]{DBLP:conf/icalp/DudekG18}
\textsc{Dudek B., Gawrychowski P.}:
\newblock Edit distance between unrooted trees in cubic time.
\newblock In \emph{45th International Colloquium on Automata, Languages, and
  Programming, {ICALP} 2018, July 9-13, 2018, Prague, Czech Republic} (2018),
  Chatzigiannakis I., Kaklamanis C., Marx D., Sannella D., (Eds.), vol.~107 of
  \emph{LIPIcs}, Schloss Dagstuhl - Leibniz-Zentrum f{\"{u}}r Informatik,
  pp.~45:1--45:14.
\newblock \href {https://doi.org/10.4230/LIPIcs.ICALP.2018.45}
  {\path{doi:10.4230/LIPIcs.ICALP.2018.45}}.

\bibitem[DMRW07]{DBLP:conf/icalp/DemaineMRW07}
\textsc{Demaine E.~D., Mozes S., Rossman B., Weimann O.}:
\newblock An optimal decomposition algorithm for tree edit distance.
\newblock In \emph{Automata, Languages and Programming, 34th International
  Colloquium, {ICALP} 2007, Wroclaw, Poland, July 9-13, 2007, Proceedings}
  (2007), Arge L., Cachin C., Jurdzinski T., Tarlecki A., (Eds.), vol.~4596 of
  \emph{Lecture Notes in Computer Science}, Springer, pp.~146--157.
\newblock \href {https://doi.org/10.1007/978-3-540-73420-8\_15}
  {\path{doi:10.1007/978-3-540-73420-8\_15}}.

\bibitem[EH10]{DBLP:books/daglib/0025666}
\textsc{Edelsbrunner H., Harer J.}:
\newblock \emph{Computational Topology - an Introduction}.
\newblock American Mathematical Society, 2010.
\newblock URL: \url{http://www.ams.org/bookstore-getitem/item=MBK-69}.

\bibitem[EHMP04]{10.1145/997817.997872}
\textsc{Edelsbrunner H., Harer J., Mascarenhas A., Pascucci V.}:
\newblock Time-varying reeb graphs for continuous space-time data.
\newblock In \emph{Proceedings of the Twentieth Annual Symposium on
  Computational Geometry} (New York, NY, USA, 2004), SCG '04, Association for
  Computing Machinery, p.~366–372.
\newblock \href {https://doi.org/10.1145/997817.997872}
  {\path{doi:10.1145/997817.997872}}.

\bibitem[ELZ00]{DBLP:conf/focs/EdelsbrunnerLZ00}
\textsc{Edelsbrunner H., Letscher D., Zomorodian A.}:
\newblock Topological persistence and simplification.
\newblock In \emph{41st Annual Symposium on Foundations of Computer Science,
  {FOCS} 2000, 12-14 November 2000, Redondo Beach, California, {USA}} (2000),
  {IEEE} Computer Society, pp.~454--463.
\newblock \href {https://doi.org/10.1109/SFCS.2000.892133}
  {\path{doi:10.1109/SFCS.2000.892133}}.

\bibitem[GST14]{DBLP:journals/cgf/GuntherST14}
\textsc{G{\"{u}}nther D., Salmon J., Tierny J.}:
\newblock Mandatory critical points of 2d uncertain scalar fields.
\newblock \emph{Comput. Graph. Forum 33}, 3 (2014), 31--40.
\newblock \href {https://doi.org/10.1111/cgf.12359}
  {\path{doi:10.1111/cgf.12359}}.

\bibitem[HLH{\etalchar{*}}16]{ChristiansSurvey}
\textsc{Heine C., Leitte H., Hlawitschka M., Iuricich F., De~Floriani L.,
  Scheuermann G., Garth C.}:
\newblock A survey of topology-based methods in visualization.
\newblock \emph{Computer Graphics Forum 35} (06 2016), 643--667.
\newblock \href {https://doi.org/10.1111/cgf.12933}
  {\path{doi:10.1111/cgf.12933}}.

\bibitem[JWZ94]{DBLP:conf/cpm/JiangWZ94}
\textsc{Jiang T., Wang L., Zhang K.}:
\newblock Alignment of trees - an alternative to tree edit.
\newblock In \emph{Combinatorial Pattern Matching, 5th Annual Symposium, {CPM}
  94, Asilomar, California, USA, June 5-8, 1994, Proceedings} (1994),
  Crochemore M., Gusfield D., (Eds.), vol.~807 of \emph{Lecture Notes in
  Computer Science}, Springer, pp.~75--86.
\newblock \href {https://doi.org/10.1007/3-540-58094-8\_7}
  {\path{doi:10.1007/3-540-58094-8\_7}}.

\bibitem[Kra10]{DBLP:conf/imagapp/Kraus10}
\textsc{Kraus M.}:
\newblock Visualization of uncertain contour trees.
\newblock In \emph{{IMAGAPP} 2010 - Proceedings of the International Conference
  on Imaging Theory and Applications and {IVAPP} 2010 - Proceedings of the
  International Conference on Information Visualization Theory and
  Applications, Angers, France, May 17 - 21, 2010} (2010), Richard P., Braz J.,
  (Eds.), {INSTICC} Press, pp.~132--139.
\newblock \href {https://doi.org/https://doi.org/10.5220/0002817201320139}
  {\path{doi:https://doi.org/10.5220/0002817201320139}}.

\bibitem[LGW{\etalchar{*}}21]{DBLP:journals/corr/abs-2107-12682}
\textsc{Lohfink A.~P., Gartzky F., Wetzels F., Vollmer L., Garth C.}:
\newblock Time-varying fuzzy contour trees.
\newblock In \emph{2021 {IEEE} Visualization Conference, {IEEE} {VIS} 2021 -
  Short Papers, New Orleans, LA, USA, October 24-29, 2021} (2021), {IEEE},
  pp.~86--90.
\newblock \href {https://doi.org/10.1109/VIS49827.2021.9623286}
  {\path{doi:10.1109/VIS49827.2021.9623286}}.

\bibitem[LWL{\etalchar{*}}20]{DBLP:journals/cgf/LohfinkWLWG20}
\textsc{Lohfink A.~P., Wetzels F., Lukasczyk J., Weber G.~H., Garth C.}:
\newblock Fuzzy contour trees: Alignment and joint layout of multiple contour
  trees.
\newblock \emph{Comput. Graph. Forum 39}, 3 (2020), 343--355.
\newblock \href {https://doi.org/10.1111/cgf.13985}
  {\path{doi:10.1111/cgf.13985}}.

\bibitem[LWM{\etalchar{*}}17]{nestedTrackingGraphs}
\textsc{Lukasczyk J., Weber G., Maciejewski R., Garth C., Leitte H.}:
\newblock Dynamic nested tracking graphs.
\newblock \emph{Computer Graphics Forum 36}, 3 (2017), 12--22.
\newblock \href {https://doi.org/10.1111/cgf.13164}
  {\path{doi:10.1111/cgf.13164}}.

\bibitem[MBW13]{morozov2013interleaving}
\textsc{Morozov D., Beketayev K., Weber G.}:
\newblock Interleaving distance between merge trees.
\newblock \emph{Discrete and Computational Geometry 49}, 22-45 (2013), 52.

\bibitem[MW13]{DBLP:conf/ppopp/MorozovW13}
\textsc{Morozov D., Weber G.~H.}:
\newblock Distributed merge trees.
\newblock In \emph{{ACM} {SIGPLAN} Symposium on Principles and Practice of
  Parallel Programming, PPoPP '13, Shenzhen, China, February 23-27, 2013}
  (2013), Nicolau A., Shen X., Amarasinghe S.~P., Vuduc R.~W., (Eds.), {ACM},
  pp.~93--102.
\newblock \href {https://doi.org/10.1145/2442516.2442526}
  {\path{doi:10.1145/2442516.2442526}}.

\bibitem[NTN15]{DBLP:conf/apvis/NarayananTN15}
\textsc{Narayanan V., Thomas D.~M., Natarajan V.}:
\newblock Distance between extremum graphs.
\newblock In \emph{2015 {IEEE} Pacific Visualization Symposium, PacificVis
  2015, Hangzhou, China, April 14-17, 2015} (2015), Liu S., Scheuermann G.,
  Takahashi S., (Eds.), {IEEE} Computer Society, pp.~263--270.
\newblock \href {https://doi.org/10.1109/PACIFICVIS.2015.7156386}
  {\path{doi:10.1109/PACIFICVIS.2015.7156386}}.

\bibitem[OHW{\etalchar{*}}17]{Oesterling2017}
\textsc{Oesterling P., Heine C., Weber G.~H., Morozov D., Scheuermann G.}:
\newblock Computing and visualizing time-varying merge trees for
  high-dimensional data.
\newblock In \emph{Mathematics and Visualization}. Springer International
  Publishing, 2017, pp.~87--101.
\newblock \href {https://doi.org/10.1007/978-3-319-44684-4_5}
  {\path{doi:10.1007/978-3-319-44684-4_5}}.

\bibitem[Pop04]{gerrisflowsolver}
\textsc{Popinet S.}:
\newblock Free computational fluid dynamics.
\newblock \emph{ClusterWorld 2}, 6 (2004).
\newblock URL: \url{http://gfs.sf.net/}.

\bibitem[PVDT21]{DBLP:journals/corr/abs-2107-07789}
\textsc{Pont M., Vidal J., Delon J., Tierny J.}:
\newblock Wasserstein distances, geodesics and barycenters of merge trees.
\newblock \emph{IEEE Transactions on Visualization and Computer Graphics}
  (2021), 1--1.
\newblock \href {https://doi.org/10.1109/TVCG.2021.3114839}
  {\path{doi:10.1109/TVCG.2021.3114839}}.

\bibitem[RSL20]{orderedEditPersistenceHierarchies}
\textsc{Rieck B., Sadlo F., Leitte H.}:
\newblock Hierarchies and ranks for persistence pairs.
\newblock In \emph{Topological Methods in Data Analysis and Visualization V:
  Theory, Algorithms, and Applications} (02 2020).
\newblock \href {https://doi.org/https://doi.org/10.1007/978-3-030-43036-8_1}
  {\path{doi:https://doi.org/10.1007/978-3-030-43036-8_1}}.

\bibitem[SAL{\etalchar{*}}07]{squillacote2007paraview}
\textsc{Squillacote A.~H., Ahrens J., Law C., Geveci B., Moreland K., King B.}:
\newblock \emph{The paraview guide}, vol.~366.
\newblock Kitware Clifton Park, NY, 2007.

\bibitem[Sel77]{DBLP:journals/ipl/Selkow77}
\textsc{Selkow S.~M.}:
\newblock The tree-to-tree editing problem.
\newblock \emph{Inf. Process. Lett. 6}, 6 (1977), 184--186.
\newblock \href {https://doi.org/10.1016/0020-0190(77)90064-3}
  {\path{doi:10.1016/0020-0190(77)90064-3}}.

\bibitem[SGL{\etalchar{*}}16]{DBLP:conf/apvis/ShuGLCLY16}
\textsc{Shu Q., Guo H., Liang J., Che L., Liu J., Yuan X.}:
\newblock Ensemblegraph: Interactive visual analysis of spatiotemporal
  behaviors in ensemble simulation data.
\newblock In \emph{2016 {IEEE} Pacific Visualization Symposium, PacificVis
  2016, Taipei, Taiwan, April 19-22, 2016} (2016), Hansen C., Viola I., Yuan
  X., (Eds.), {IEEE} Computer Society, pp.~56--63.
\newblock \href {https://doi.org/10.1109/PACIFICVIS.2016.7465251}
  {\path{doi:10.1109/PACIFICVIS.2016.7465251}}.

\bibitem[SHD{\etalchar{*}}20]{DBLP:journals/tvcg/SchnorrHDKH20}
\textsc{Schnorr A., Helmrich D.~N., Denker D., Kuhlen T.~W., Hentschel B.}:
\newblock Feature tracking by two-step optimization.
\newblock \emph{{IEEE} Trans. Vis. Comput. Graph. 26}, 6 (2020), 2219--2233.
\newblock \href {https://doi.org/10.1109/TVCG.2018.2883630}
  {\path{doi:10.1109/TVCG.2018.2883630}}.

\bibitem[SMKN20]{DBLP:journals/tvcg/SridharamurthyM20}
\textsc{Sridharamurthy R., Masood T.~B., Kamakshidasan A., Natarajan V.}:
\newblock Edit distance between merge trees.
\newblock \emph{{IEEE} Trans. Vis. Comput. Graph. 26}, 3 (2020), 1518--1531.
\newblock \href {https://doi.org/10.1109/TVCG.2018.2873612}
  {\path{doi:10.1109/TVCG.2018.2873612}}.

\bibitem[SMP15]{categorifiedreebgraphs}
\textsc{Silva V., Munch E., Patel A.}:
\newblock Categorified {R}eeb graphs.
\newblock \emph{Discrete and Computational Geometry 55} (01 2015).
\newblock \href {https://doi.org/10.1007/s00454-016-9763-9}
  {\path{doi:10.1007/s00454-016-9763-9}}.

\bibitem[SN21]{DBLP:journals/corr/abs-2111-04382}
\textsc{Sridharamurthy R., Natarajan V.}:
\newblock Comparative analysis of merge trees using local tree edit distance.
\newblock \emph{IEEE Transactions on Visualization and Computer Graphics}
  (2021), 1--1.
\newblock \href {https://doi.org/10.1109/TVCG.2021.3122176}
  {\path{doi:10.1109/TVCG.2021.3122176}}.

\bibitem[SSW14]{DBLP:journals/cgf/SaikiaSW14}
\textsc{Saikia H., Seidel H., Weinkauf T.}:
\newblock Extended branch decomposition graphs: Structural comparison of scalar
  data.
\newblock \emph{Comput. Graph. Forum 33}, 3 (2014), 41--50.
\newblock \href {https://doi.org/10.1111/cgf.12360}
  {\path{doi:10.1111/cgf.12360}}.

\bibitem[SW17]{DBLP:journals/cgf/SaikiaW17}
\textsc{Saikia H., Weinkauf T.}:
\newblock Global feature tracking and similarity estimation in time-dependent
  scalar fields.
\newblock \emph{Comput. Graph. Forum 36}, 3 (2017), 1--11.
\newblock \href {https://doi.org/10.1111/cgf.13163}
  {\path{doi:10.1111/cgf.13163}}.

\bibitem[Tai79]{DBLP:journals/jacm/Tai79}
\textsc{Tai K.}:
\newblock The tree-to-tree correction problem.
\newblock \emph{J. {ACM} 26}, 3 (1979), 422--433.
\newblock \href {https://doi.org/10.1145/322139.322143}
  {\path{doi:10.1145/322139.322143}}.

\bibitem[TFL{\etalchar{*}}18]{DBLP:journals/tvcg/TiernyFLGM18}
\textsc{Tierny J., Favelier G., Levine J.~A., Gueunet C., Michaux M.}:
\newblock The topology toolkit.
\newblock \emph{{IEEE} Trans. Vis. Comput. Graph. 24}, 1 (2018), 832--842.
\newblock \href {https://doi.org/10.1109/TVCG.2017.2743938}
  {\path{doi:10.1109/TVCG.2017.2743938}}.

\bibitem[TMMH14]{DBLP:journals/dcg/TurnerMMH14}
\textsc{Turner K., Mileyko Y., Mukherjee S., Harer J.}:
\newblock Fr{\'{e}}chet means for distributions of persistence diagrams.
\newblock \emph{Discret. Comput. Geom. 52}, 1 (2014), 44--70.
\newblock \href {https://doi.org/10.1007/s00454-014-9604-7}
  {\path{doi:10.1007/s00454-014-9604-7}}.

\bibitem[TN11]{DBLP:journals/tvcg/ThomasN11}
\textsc{Thomas D.~M., Natarajan V.}:
\newblock Symmetry in scalar field topology.
\newblock \emph{{IEEE} Trans. Vis. Comput. Graph. 17}, 12 (2011), 2035--2044.
\newblock \href {https://doi.org/10.1109/TVCG.2011.236}
  {\path{doi:10.1109/TVCG.2011.236}}.

\bibitem[WLG21]{repository}
\textsc{Wetzels F., Leitte H., Garth C.}:
\newblock Branch decomposition-independent edit distances (supplementary source
  code).
\newblock \url{https://github.com/scivislab/bdi-ed}, 2021.

\bibitem[WN08]{scivisContest2008}
\textsc{Whalen D., Norman M.~L.}:
\newblock Competition data set and description, in 2008 ieee visualization
  design contest, 2008.
\newblock URL: \url{http://vis.computer.org/VisWeek2008/vis/contests.html}.

\bibitem[WT10]{weinkauf10c}
\textsc{Weinkauf T., Theisel H.}:
\newblock Streak lines as tangent curves of a derived vector field.
\newblock \emph{IEEE Transactions on Visualization and Computer Graphics
  (Proceedings Visualization 2010) 16}, 6 (November - December 2010),
  1225--1234.
\newblock URL: \url{http://tinoweinkauf.net/}.

\bibitem[WZ13]{wu2013contour}
\textsc{Wu K., Zhang S.}:
\newblock A contour tree based visualization for exploring data with
  uncertainty.
\newblock \emph{International Journal for Uncertainty Quantification 3}, 3
  (2013).
\newblock \href
  {https://doi.org/https://doi.org/10.1615/INT.J.UNCERTAINTYQUANTIFICATION.2012003956}
  {\path{doi:https://doi.org/10.1615/INT.J.UNCERTAINTYQUANTIFICATION.2012003956}}.

\bibitem[YMS{\etalchar{*}}21]{DBLP:journals/cgf/YanMSRNHW21}
\textsc{Yan L., Masood T.~B., Sridharamurthy R., Rasheed F., Natarajan V., Hotz
  I., Wang B.}:
\newblock Scalar field comparison with topological descriptors: Properties and
  applications for scientific visualization.
\newblock \emph{Comput. Graph. Forum 40}, 3 (2021), 599--633.
\newblock \href {https://doi.org/10.1111/cgf.14331}
  {\path{doi:10.1111/cgf.14331}}.

\bibitem[YWM{\etalchar{*}}20]{DBLP:journals/tvcg/YanWMGW20}
\textsc{Yan L., Wang Y., Munch E., Gasparovic E., Wang B.}:
\newblock A structural average of labeled merge trees for uncertainty
  visualization.
\newblock \emph{{IEEE} Trans. Vis. Comput. Graph. 26}, 1 (2020), 832--842.
\newblock \href {https://doi.org/10.1109/TVCG.2019.2934242}
  {\path{doi:10.1109/TVCG.2019.2934242}}.

\bibitem[Zha96]{DBLP:journals/algorithmica/Zhang96}
\textsc{Zhang K.}:
\newblock A constrained edit distance between unordered labeled trees.
\newblock \emph{Algorithmica 15}, 3 (1996), 205--222.
\newblock \href {https://doi.org/10.1007/BF01975866}
  {\path{doi:10.1007/BF01975866}}.

\bibitem[ZJ94]{DBLP:journals/ipl/ZhangJ94}
\textsc{Zhang K., Jiang T.}:
\newblock Some {MAX} snp-hard results concerning unordered labeled trees.
\newblock \emph{Inf. Process. Lett. 49}, 5 (1994), 249--254.
\newblock \href {https://doi.org/10.1016/0020-0190(94)90062-0}
  {\path{doi:10.1016/0020-0190(94)90062-0}}.

\bibitem[ZSS92]{DBLP:journals/ipl/ZhangSS92}
\textsc{Zhang K., Statman R., Shasha D.~E.}:
\newblock On the editing distance between unordered labeled trees.
\newblock \emph{Inf. Process. Lett. 42}, 3 (1992), 133--139.
\newblock \href {https://doi.org/10.1016/0020-0190(92)90136-J}
  {\path{doi:10.1016/0020-0190(92)90136-J}}.

\end{thebibliography}


\begin{thebibliography}{\uppercase{SMKN20}}

\bibitem[PVDT21]{DBLP:journals/corr/abs-2107-07789}
\textsc{Pont M., Vidal J., Delon J., Tierny J.}:
\newblock Wasserstein distances, geodesics and barycenters of merge trees.
\newblock \emph{IEEE Transactions on Visualization and Computer Graphics}
  (2021), 1--1.
\newblock \href {https://doi.org/10.1109/TVCG.2021.3114839}
  {\path{doi:10.1109/TVCG.2021.3114839}}.

\bibitem[SMKN20]{DBLP:journals/tvcg/SridharamurthyM20}
\textsc{Sridharamurthy R., Masood T.~B., Kamakshidasan A., Natarajan V.}:
\newblock Edit distance between merge trees.
\newblock \emph{{IEEE} Trans. Vis. Comput. Graph. 26}, 3 (2020), 1518--1531.
\newblock \href {https://doi.org/10.1109/TVCG.2018.2873612}
  {\path{doi:10.1109/TVCG.2018.2873612}}.

\bibitem[SSW14]{DBLP:journals/cgf/SaikiaSW14}
\textsc{Saikia H., Seidel H., Weinkauf T.}:
\newblock Extended branch decomposition graphs: Structural comparison of scalar
  data.
\newblock \emph{Comput. Graph. Forum 33}, 3 (2014), 41--50.
\newblock \href {https://doi.org/10.1111/cgf.12360}
  {\path{doi:10.1111/cgf.12360}}.

\bibitem[Zha96]{DBLP:journals/algorithmica/Zhang96}
\textsc{Zhang K.}:
\newblock A constrained edit distance between unordered labeled trees.
\newblock \emph{Algorithmica 15}, 3 (1996), 205--222.
\newblock \href {https://doi.org/10.1007/BF01975866}
  {\path{doi:10.1007/BF01975866}}.

\end{thebibliography}


\end{document}


\maketitle

\appendix

\section{Unordered vs Ordered BDTs and other Relations}
\label{appendix:unorderedVsordered}

We illustrate the difference between the ordered distance from~\cite{DBLP:journals/cgf/SaikiaSW14} and the unordered distance from~\cite{DBLP:journals/corr/abs-2107-07789} on an example merge tree which can be seen in Figure~\ref{fig:unorderedVSordered}. The three side branches b-f, c-g and d-h have a natural ordering defined by the scalar value of their saddles. This ordering is shown in the sibling orderings in the BDTs. Pont et al.\ allow arbitrary mappings, i.e.\ all three mappings are valid for their distance, whereas Saikia et al.\ only allow order-preserving mappings. This means that crossings in the mappings such as in the first one are not taken into account. This can also be phrased differently: The allowed mappings between child branches of a node correspond to string-edit mappings between the two child-sequences in this method by Saikia et al.\ whereas the method of Pont et al.\ checks for the optimal maximum matching between the two sets of children.

Furthermore, the unordered approach on BDTs used by Pont et al.\ gives them advantage over all methods that work directly on merge trees, as it allows for mappings that are ancestor-preserving in the BDT but not in the origial merge tree, as can be seen by the implicit mapping of the node with label $b$ and $c$ in Figure~\ref{fig:unorderedVSordered}.

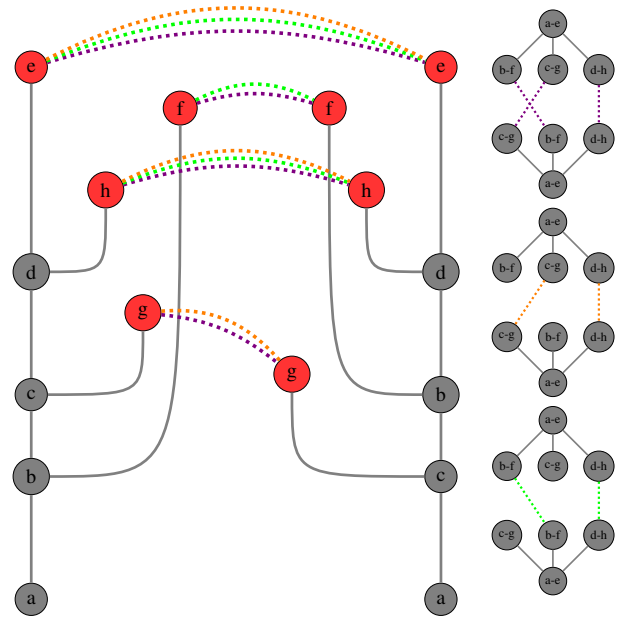
\begin{figure}[h]
    \centering
    \begin{subfigure}[c]{0.72\linewidth}
        \centering
        \resizebox{\linewidth}{!}{
        \begin{tikzpicture}[scale=0.58,yscale=1.1]
            
            \node[draw,circle,fill=gray!100] at (-2, 0) (root_1) {a};
            \node[draw,circle,fill=gray!100] at (-2, 3) (s1_1) {b};
            \node[draw,circle,fill=gray!100] at (-2, 5) (s2_1) {c};
            \node[draw,circle,fill=gray!100] at (-2, 8) (s3_1) {d};
            \node[draw,circle,fill=red!80] at (-2, 13) (m1_1) {e};
            \node[draw,circle,fill=red!80] at (2, 12) (m2_1) {f};
            \node[draw,circle,fill=red!80] at (1, 7) (m3_1) {g};
            \node[draw,circle,fill=red!80] at (0, 10) (m4_1) {h};
            \draw[gray,very thick] (root_1) -- (s1_1);
            \draw[gray,very thick] (s1_1) -- (s2_1);
            \draw[gray,very thick] (s2_1) -- (s3_1);
            \draw[gray,very thick] (s3_1) -- (m1_1);
            \draw[gray,very thick] (s1_1) .. controls (2,3) .. (m2_1);
            \draw[gray,very thick] (s2_1) .. controls (1,5) .. (m3_1);
            \draw[gray,very thick] (s3_1) .. controls (0,8) .. (m4_1);
            
            \node[draw,circle,fill=gray!100] at (2+7, 0) (root_2) {a};
            \node[draw,circle,fill=gray!100] at (2+7, 5) (s1_2) {b};
            \node[draw,circle,fill=gray!100] at (2+7, 3) (s2_2) {c};
            \node[draw,circle,fill=gray!100] at (2+7, 8) (s3_2) {d};
            \node[draw,circle,fill=red!80] at (-1+7, 12) (m1_2) {f};
            \node[draw,circle,fill=red!80] at (2+7, 13) (m2_2) {e};
            \node[draw,circle,fill=red!80] at (-2+7, 5.5) (m3_2) {g};
            \node[draw,circle,fill=red!80] at (0+7, 10) (m4_2) {h};
            \draw[gray,very thick] (root_2) -- (s2_2);
            \draw[gray,very thick] (s1_2) .. controls (-1+7,5) .. (m1_2);
            \draw[gray,very thick] (s2_2) .. controls (-2+7,3) .. (m3_2);
            \draw[gray,very thick] (s3_2) .. controls (0+7,8) .. (m4_2);
            \draw[gray,very thick] (s2_2) -- (s1_2);
            \draw[gray,very thick] (s1_2) -- (s3_2);
            \draw[gray,very thick] (s3_2) -- (m2_2);
            
            \draw[ultra thick,-,dotted,violet] (m2_1) to[bend left=15] (m1_2);
            \draw[ultra thick,-,dotted,violet] (m1_1) to[bend left=15] (m2_2);
            \draw[ultra thick,-,dotted,violet] (m4_1) to[bend left=15] (m4_2);
            \draw[ultra thick,-,dotted,violet] (m3_1) to[bend left=15] (m3_2);
            
            \draw[ultra thick,-,dotted,green] (m2_1) to[bend left=25] (m1_2);
            \draw[ultra thick,-,dotted,green] (m1_1) to[bend left=20] (m2_2);
            \draw[ultra thick,-,dotted,green] (m4_1) to[bend left=20] (m4_2);
            
            \draw[ultra thick,-,dotted,orange] (m3_1) to[bend left=25] (m3_2);
            \draw[ultra thick,-,dotted,orange] (m1_1) to[bend left=25] (m2_2);
            \draw[ultra thick,-,dotted,orange] (m4_1) to[bend left=25] (m4_2);
            
        \end{tikzpicture}
        }
    \end{subfigure}
    \begin{subfigure}[c]{0.26\linewidth}
        \centering
        \resizebox{0.78\linewidth}{!}{
        \begin{tikzpicture}[scale=0.58]
            \node[draw=none,fill=none,circle] at (0, 4.5) (dummy) {};
            \node[draw=none,fill=none,circle] at (0, -3.5) (dummy) {};
            
            \node[draw,circle,fill=gray!100] at (0, 4) (root_1) {a-e};
            \node[draw,circle,fill=gray!100] at (-2, 2) (b1_1) {b-f};
            \node[draw,circle,fill=gray!100] at (0, 2) (b2_1) {c-g};
            \node[draw,circle,fill=gray!100] at (2, 2) (b3_1) {d-h};
            \draw[gray,very thick] (root_1) -- (b1_1);
            \draw[gray,very thick] (root_1) -- (b2_1);
            \draw[gray,very thick] (root_1) -- (b3_1);
            
            \node[draw,circle,fill=gray!100] at (0, -3) (root_2) {a-e};
            \node[draw,circle,fill=gray!100] at (-2, -1) (b1_2) {c-g};
            \node[draw,circle,fill=gray!100] at (0, -1) (b2_2) {b-f};
            \node[draw,circle,fill=gray!100] at (2, -1) (b3_2) {d-h};
            \draw[gray,very thick] (root_2) -- (b1_2);
            \draw[gray,very thick] (root_2) -- (b2_2);
            \draw[gray,very thick] (root_2) -- (b3_2);
            
            \draw[ultra thick,-,dotted,violet] (b1_1) to[] (b2_2);
            \draw[ultra thick,-,dotted,violet] (b2_1) to[] (b1_2);
            \draw[ultra thick,-,dotted,violet] (b3_1) to[] (b3_2);
            
        \end{tikzpicture}
        }
        
        \resizebox{0.78\linewidth}{!}{
        \begin{tikzpicture}[scale=0.58]
            \node[draw=none,fill=none,circle] at (0, 4.5) (dummy) {};
            \node[draw=none,fill=none,circle] at (0, -3.5) (dummy) {};
            
            \node[draw,circle,fill=gray!100] at (0, 4) (root_1) {a-e};
            \node[draw,circle,fill=gray!100] at (-2, 2) (b1_1) {b-f};
            \node[draw,circle,fill=gray!100] at (0, 2) (b2_1) {c-g};
            \node[draw,circle,fill=gray!100] at (2, 2) (b3_1) {d-h};
            \draw[gray,very thick] (root_1) -- (b1_1);
            \draw[gray,very thick] (root_1) -- (b2_1);
            \draw[gray,very thick] (root_1) -- (b3_1);
            
            \node[draw,circle,fill=gray!100] at (0, -3) (root_2) {a-e};
            \node[draw,circle,fill=gray!100] at (-2, -1) (b1_2) {c-g};
            \node[draw,circle,fill=gray!100] at (0, -1) (b2_2) {b-f};
            \node[draw,circle,fill=gray!100] at (2, -1) (b3_2) {d-h};
            \draw[gray,very thick] (root_2) -- (b1_2);
            \draw[gray,very thick] (root_2) -- (b2_2);
            \draw[gray,very thick] (root_2) -- (b3_2);
            
            \draw[ultra thick,-,dotted,orange] (b2_1) to[] (b1_2);
            \draw[ultra thick,-,dotted,orange] (b3_1) to[] (b3_2);
            
        \end{tikzpicture}
        }
        
        \resizebox{0.78\linewidth}{!}{
        \begin{tikzpicture}[scale=0.58]
            \node[draw=none,fill=none,circle] at (0, 4.5) (dummy) {};
            \node[draw=none,fill=none,circle] at (0, -3.5) (dummy) {};
            
            \node[draw,circle,fill=gray!100] at (0, 4) (root_1) {a-e};
            \node[draw,circle,fill=gray!100] at (-2, 2) (b1_1) {b-f};
            \node[draw,circle,fill=gray!100] at (0, 2) (b2_1) {c-g};
            \node[draw,circle,fill=gray!100] at (2, 2) (b3_1) {d-h};
            \draw[gray,very thick] (root_1) -- (b1_1);
            \draw[gray,very thick] (root_1) -- (b2_1);
            \draw[gray,very thick] (root_1) -- (b3_1);
            
            \node[draw,circle,fill=gray!100] at (0, -3) (root_2) {a-e};
            \node[draw,circle,fill=gray!100] at (-2, -1) (b1_2) {c-g};
            \node[draw,circle,fill=gray!100] at (0, -1) (b2_2) {b-f};
            \node[draw,circle,fill=gray!100] at (2, -1) (b3_2) {d-h};
            \draw[gray,very thick] (root_2) -- (b1_2);
            \draw[gray,very thick] (root_2) -- (b2_2);
            \draw[gray,very thick] (root_2) -- (b3_2);
            
            \draw[ultra thick,-,dotted,green] (b1_1) to[] (b2_2);
            \draw[ultra thick,-,dotted,green] (b3_1) to[] (b3_2);
            
        \end{tikzpicture}
        }
    \end{subfigure}
    \caption{On the left, we see a merge tree with three different mappings between the leaves in purple, green and orange. On the right, we see the three corresponding mappings between the branch decomposition graphs.}
    \label{fig:unorderedVSordered}
\end{figure}

\paragraph*{Relation of $d_C$, $W_2^T$ and $d_S$}

We now discuss the relation of the three branch decomposition-dependent methods in Figure~4 in more detail, however, still in an intuitive manner without providing formal proofs for the claims. The constrained edit distance $d_C$ is, as stated in the categorization, working on BDTs. In the original paper, it is defined to work on merge trees that are labeled with branch properties. This makes it equivalent (in terms of the mapping search space) to a constrained edit distance on \emph{ordered} BDTs, since the ancestor preservation condition in merge trees translates to the ordering in BDTs. Therefore, the search space of $d_C$ is a strict superset of the search space of $d_S$ (as long as the same base metric is used).

As discussed before, the Wasserstein Distance on merge trees $W_2^T$ uses $d_1$ as the underlying edit distance, as does $d_S$, however, it works on \emph{unordered} BDTs. The lifting to unordered BDTs is the only difference (in terms of search space) between $d_S$ and $W_2^T$. The only difference between $d_S$ and $d_C$ is the step from $d_1$ to $d_c$. As both variants require order preserving mappings, both consider strictly ordered BDTs. From these two observations, we can follow that the search space extensions of $d_C$ and $W_2^T$ in comparison to $d_S$ are mutually exclusive and therefore $d_S$ has to be the intersection of the two, as shown in Figure~4.

\section{Proof of Theorem~1}
\label{appendix:metric_proof}

For completeness, we first state the theorem again:

$d_B$ is a metric on the set 
$$\{(T,B) \bigm| T \text{ is an abstract merge tree, } B \in B(T)\},$$
as long as the cost function $c$ on the branch labels is a metric.

\begin{proof}
Given two trees $T_1,T_2$ with branch decompositions $B_1 \in B(T_1),B_2 \in B(T_2)$, the first two metric properties, identity and symmetry are trivial, so only the triangle inequality remains to be proven.
Therefore, consider a third tree $T_3$ with branch decomposition $B_3 \in B(T_3)$ and optimal mappings $M_{1,2} \subseteq B_1 \times B_2,M_{2,3} \subseteq B_2 \times B_3$ with $d_B(B_1,B_2) = \cost(M_{1,2})$ and $d_B(B_2,B_3) = \cost(M_{2,3})$.
Now suppose that the triangle inequality is violated for the triple $B_1,B_2,B_3$, i.e.\ $d_B(B_1,B_3)>d_B(B_1,B_2)+d_B(B_2,B_3)$.

To show a contradiction, we construct another mapping $M_{1,3}$ between $B_1$ and $B_3$ in the following way.
We define $(a,c) \in M_{1,3}$ for all branches $a \in B_1, c \in B_3$ where there is a branch $b \in B_2$ such that $(a,b) \in M_{1,2}$ and $(b,c) \in M_{2,3}$.
Now consider the costs of the new mapping
$$\cost(M_{1,3}) = \sum_{(a,c) \in \overline{M_{1,3}}} \cost(a,c),$$
and specifically the single terms in the sum, pairs $p$ with costs $\cost(p) = \cost(a,c)$.
If $a \in B_1$ and $c \in B_3$, then there is a branch $b \in B_2$ such that $(a,b) \in M_{1,2}$ and $(b,c) \in M_{2,3}$.
Then we know that $\cost(a,b)$ and $\cost(b,c)$ are terms in the sums of $\cost(M_{1,2})$ and $\cost(M_{2,3})$ and also that $\cost(a,b)+\cost(b,c)>\cost(a,c)=\cost(p)$ due to the metric property of the basic distance function.
If $c \notin B_3$ (i.e.\ $c=\bot$), then either $(a,\bot) \in \overline{M_{1,2}}$ holds or $(a,b) \in \overline{M_{1,2}}$ and $(b,\bot) \in \overline{M_{2,3}}$.
In both cases we know that $\cost(a,\bot) \geq \cost(p)$ and $\cost(a,b)+\cost(b,\bot) \geq \cost(p)$.
If $a \notin B_1$ (i.e.\ $a=\bot$), then either $(\bot,c) \in \overline{M_{2,3}}$ holds or $(\bot,b) \in \overline{M_{1,2}}$ and $(b,c) \in \overline{M_{2,3}}$.
Again, in both cases we know that $\cost(\bot,c) \geq \cost(p)$ and $\cost(a,b)+\cost(b,\bot) \geq \cost(p)$.
In total, we can conclude that $\cost(M_{1,3}) \leq \cost(M_{1,2})+\cost(M_{2,3})$ and therefore $\cost(M_{1,3}) \leq d_B(M_1,M_3)$, which leads to a contradiction.
\end{proof}

\section{Proof of Lemma~2}
\label{appendix:rec_inner_proof}

For completeness, we state the lemma again:

Given two merge trees $T_1,T_2$ with roots $v_1,u_1$, let $v_2,u_2$ be the unique children of the two roots and let those have children $v_3,v_4$ and $u_3,u_4$. Let $T_1'$ be the subtree rooted in $(v_2,v_3)$, $T_1''$ rooted in $(v_2,v_4)$, $T_2'$ rooted in $(u_2,u_3)$ and $T_2''$ in $(u_2,u_4)$. Let $M$ be an optimal branch mapping for $T_1$ and $T_2$. Then, for the optimal cost of $M$ it holds that:
\begin{itemize}
    \item $d_B(T_1,T_2) = d_B(T_1',\bot) + d_B(T_1-T_1',T_2)$ or
    \item $d_B(T_1,T_2) = d_B(\bot,T_2') + d_B(T_1,T_2-T_2')$ or
    \item $d_B(T_1,T_2) = d_B(T_1'',\bot) + d_B(T_1-T_1'',T_2)$ or
    \item $d_B(T_1,T_2) = d_B(\bot,T_2'') + d_B(T_1,T_2-T_2'')$ or
    \item $d_B(T_1,T_2) = d_B(T_1',T_2') + d_B(T_1-T_1',T_2-T_2')$ or
    \item $d_B(T_1,T_2) = d_B(T_1'',T_2'') + d_B(T_1-T_1'',T_2-T_2'')$ or
    \item $d_B(T_1,T_2) = d_B(T_1',T_2'') + d_B(T_1-T_1',T_2-T_2'')$ or
    \item $d_B(T_1,T_2) = d_B(T_1'',T_2') + d_B(T_1-T_1'',T_2-T_2')$
\end{itemize}

\begin{proof}
To prove this lemma, we will do two nested case distinctions. First, consider the optimal mapping $M$. It is build upon two branch decompositions $B_1 \in B(T_1)$ and $B_2 \in B(T_2)$. Since the roots have degree one, in both branch decompositions the main branch has to go through $v_2$ and $u_2$. Therefore, in $B_1$ either the edge $(v_2,v_3)$ or the edge $(v_2,v_4)$ is contained in the main branch, and in $B_2$ either $(u_2,u_3)$ or $(u_2,u_4)$.

Let us first consider the case that $(v_2,v_4)$ and $(u_2,u_4)$ are part of the main branches of $B_1$ and $B_2$. Then there is a branch $a=v_2v_3...v_{\text{l}} \in B_1$ and a branch $b=u_2u_3...u_{\text{l}} \in B_2$ and we know that $B_1[T_1']$ and $B_2[T_2']$ exist. For the mapping $M$ we have the following options:
\begin{itemize}
    \item[(A)] $(a,b) \in M$: In this case the subtrees $T_1'$ and $T_2'$ are mapped to each other, i.e.\ $M[B_1[T_1']] = M[B_2[T_2']]$ and it holds that $\cost(M) = \cost(M[B_1[T_1']]) + \cost(M[B_1[T_1-T_1']])$ which then means that $M[B_1[T_1']]$ and $M[B_1[T_1-T_1']]$ are optimal mappings between $T_1'$ and $T_2'$ and between $T_1-T_1'$ and $T_2-T_2'$, as otherwise they could be replaced in $M$ by better mappings contradicting the optimality of $M$.
    \item[(B)] $(a,b') \in M$, but $b' \neq b$ is not a branch of $T_2'$: If $a$ is mapped to a branch other than $b$ this means that $b$ and all its descendant branches are not mapped in $M$ as this would either contradict the order condition or the parent condition of branch mappings (conditions~3 and~4, Definition~2).
    \item[(C)] $(a',b) \in M$, but $a' \neq a$ is not a branch of $T_1'$: If $b$ is mapped to a branch other than $a$ this means that $a$ and all its descendant branches are not mapped in $M$ as this would either contradict the order condition  or the parent condition of branch mappings  (conditions~3 and~4, Definition~2).
    \item[(D)] There is no branch $a'$ of $T_1'$ and no branch $b'$ of $T_2$ with ${(a,b') \in M}$ or ${(a',b') \in M}$.
\end{itemize}
\noindent
In each case, we can conclude the following for $M$ and $\cost(M)=d_B(T_1,T_2)$:
\begin{itemize}
    \item[(A)] Since $M[B_1[T_1']]$ and $M[B_1[T_1-T_1']]$ are optimal, we can conclude that $$d_B(T_1,T_2) = d_B(T_1',T_2') + d_B(T_1-T_1',T_2-T_2').$$
    \item[(B)] Since $T_2'$ is not covered by $M$, we can conclude that it is mapped to $\bot$ in $\overline{M}$ and therefore $$d_B(T_1,T_2) = d_B(\bot,T_2') + d_B(T_1,T_2-T_2').$$
    \item[(C)] Since $T_1'$ is not covered by $M$, we can conclude that it is mapped to $\bot$ in $\overline{M}$ and therefore $$d_B(T_1,T_2) = d_B(T_1',\bot) + d_B(T_1-T_1',T_2).$$
    \item[(D)] Since $T_1'$ and $T_2'$ are not covered by $M$, we can conclude that they are mapped to $\bot$ in $\overline{M}$ and therefore $$d_B(T_1,T_2) = d_B(T_1',\bot) + d_B(T_1-T_1',T_2)$$ and $$d_B(T_1,T_2) = d_B(\bot,T_2') + d_B(T_1,T_2-T_2')$$ both hold.
\end{itemize}
Together with the other cases for $B_1$ and $B_2$, which yield analogous equations, we get all the terms from above.
\end{proof}

\section{Proof of Lemmma~3}
\label{appendix:rec_delete_proof}

For completeness, we state the lemma again:

Given a merge tree $T_1$ with root $v_1$, let $v_2$ be the unique children of the root and let it have children $v_3,v_4$. Let $T_1'$ be the subtree rooted in $(v_2,v_3)$ and $T_1''$ rooted in $(v_2,v_4)$. Then, for the optimal cost of $M_\bot(T_1)$ it holds that:
\begin{itemize}
    \item $d_B(T_1,\bot) = d_B(T_1',\bot) + d_B(T_1-T_1',\bot)$ or
    \item $d_B(T_1,\bot) = d_B(T_1'',\bot) + d_B(T_1-T_1'',\bot)$,
\end{itemize}
\noindent
and $d_B(\bot,T_2)$ decomposes symmetrically.

\begin{proof*}
If one of the trees is empty, then we only have to find the branch decomposition for the other one that has the lowest cost. Locally, we only have two cases. Either $v_3$ lies on the main branch or $v_4$ does. In either case, the optimal branch decomposition $B_1 \in B(T_1)$ decomposes into those of the trees listed in the lemma, as seen in Section~3.1. If $v_3$ is on the main branch of $B_1$, then
$$\cost(B_1,\bot) = \cost(B_1[T_1''],\bot) + \cost(B_1[T_1-T_1''],\bot)$$
and if $v_4$ is on the main branch of $B_1$, then
$$\cost(B_1,\bot) = \cost(B_1[T_1'],\bot) + \cost(B_1[T_1-T_1'],\bot).\hspace{10pt}\proofbox$$
\end{proof*}

\section{Proof of Lemmma~4}
\label{appendix:rec_leaf_proof}

For completeness, we state the lemma again:

Given two merge trees $T_1 = (\{v_1,v_2\},\{(v_2,v_1)\})$, $T_2 = (\{u_1,u_2\},\{(u_2,u_1)\})$ that only have one branch, the following holds for $d_B$:
\begin{itemize}
    \item $d_B(T_1,\bot) = \cost(v_1v_2,\bot)$,
    \item $d_B(\bot,T_2) = \cost(\bot,u_1u_2)$ and
    \item $d_B(T_1,T_2) = \cost(v_1v_2,u_1u_2)$.
\end{itemize}

\begin{proof}
For $d_B(T_1,\bot)$ and $d_B(\bot,T_2)$ this follows directly from the definition of $M_\bot(T_1),M_\bot(T_2)$ and the uniqueness of the branch decompositions.

For $d_B(T_1,T_2)$ it follows directly from the requirement to map the main branches (condition~2 in Definiton~2) and the uniqueness of the branch decompositions.
\end{proof}

\section{Algorithm}
\label{appendix:algo}

We now give the pseudo code for an algorithm computing the distance $d_B$.
Again, we only show the recursion for binary trees, but it is easy to adapt for trees of arbitrary degree (see below).
Algorithm~\ref{alg:branchDist} shows the recursive procedure without memoization.

\SetKwComment{Comment}{/* }{ */}
\setlength{\algomargin}{5pt}
\SetAlgoVlined
\begin{algorithm}
\label{alg:branch_dist}
\caption{Computing the branch mapping distance}\label{alg:branchDist}
\SetKwFunction{branchDist}{$d_B$}
\SetKwProg{Fn}{Function}{:}{}
\DontPrintSemicolon
\Fn{\branchDist{$n_1,p_1,n_2,p_2$}}{
\If{$n_1=\bot$ and $n_2$ is a leaf}{
  \Return $\cost(\bot,p_2...n_2)$\;
}
\If{$n_2=\bot$ and $n_1$ is a leaf}{
  \Return $\cost(p_1...n_1,\bot)$\;
}
\If{$n_1$ is a leaf and $n_2$ is a leaf}{
  \Return $\cost(p_1...n_1,p_2...n_2)$\;
}
\If{$n_1=\bot$ and $n_2$ is an inner node}{
  Let $c_{2,1},c_{2,2}$ be the children of $n_2$\;
  \Return~$\min \begin{cases}
           \branchDist(\bot,\bot,c_{2,1},p_2) + \branchDist(\bot,\bot,c_{2,2},n_2)\\
           \branchDist(\bot,\bot,c_{2,2},p_2) + \branchDist(\bot,\bot,c_{2,1},n_2)
           \end{cases}$\;
}
\If{$n_2=\bot$ and $n_1$ is an inner node}{
  Let $c_{1,1},c_{1,2}$ be the children of $n_1$\;
  \Return~$\min \begin{cases}
           \branchDist(c_{1,1},p_1,\bot,\bot) + \branchDist(c_{1,2},n_1,\bot,\bot)\\
           \branchDist(c_{1,2},p_1,\bot,\bot) + \branchDist(c_{1,1},n_1,\bot,\bot)
           \end{cases}$\;
}
\If{$n_1$ is a leaf and $n_2$ is an inner node}{
  Let $c_{2,1},c_{2,2}$ be the children of $n_2$\;
  \Return~$\min \begin{cases}
           \branchDist(n_1,p_1,c_{2,1},p_2) + \branchDist(\bot,\bot,c_{2,2},n_2)\\
           \branchDist(n_1,p_1,c_{2,2},p_2) + \branchDist(\bot,\bot,c_{2,1},n_2)
           \end{cases}$\;
}
\If{$n_2$ is a leaf and $n_1$ is an inner node}{
  Let $c_{1,1},c_{1,2}$ be the children of $n_1$\;
  \Return~$\min \begin{cases}
           \branchDist(c_{1,1},p_1,n_2,p_2) + \branchDist(c_{1,2},n_1,\bot,\bot)\\
           \branchDist(c_{1,2},p_1,n_2,p_2) + \branchDist(c_{1,1},n_1,\bot,\bot)
           \end{cases}$\;
}
\If{$n_1$ is an inner node and $n_2$ is an inner node}{
  Let $c_{1,1},c_{1,2}$ be the children of $n_1$\;
  Let $c_{2,1},c_{2,2}$ be the children of $n_2$\;
  \Return~$\min \begin{cases}
           \branchDist(c_{1,1},p_1,n_2,p_2) + \branchDist(c_{1,2},n_1,\bot,\bot)\\
           \branchDist(c_{1,2},p_1,n_2,p_2) + \branchDist(c_{1,1},n_1,\bot,\bot)\\
           \branchDist(n_1,p_1,c_{2,1},p_2) + \branchDist(\bot,\bot,c_{2,2},n_2)\\
           \branchDist(n_1,p_1,c_{2,2},p_2) + \branchDist(\bot,\bot,c_{2,1},n_2)\\
           \branchDist(c_{1,1},p_1,c_{2,1},p_2) + \branchDist(c_{1,2},n_1,c_{2,2},n_2)\\
           \branchDist(c_{1,1},n_1,c_{2,1},n_2) + \branchDist(c_{1,2},p_1,c_{2,2},p_2)\\
           \branchDist(c_{1,2},p_1,c_{2,1},p_2) + \branchDist(c_{1,1},n_1,c_{2,2},n_2)\\
           \branchDist(c_{1,2},n_1,c_{2,1},n_2) + \branchDist(c_{1,1},p_1,c_{2,2},p_2)\\
           \end{cases}$\;
}
}\end{algorithm}

\begin{theorem} Let $T_1,T_2$ be abstract merge trees with roots $r_1,r_2$ and let $r_1',r_2'$ be their unique child. Then $d_B(r_1,r_1',r_2,r_2')$ in Algorithm~\ref{alg:branchDist} computes the branch mapping distance between $T_1$ and $T_2$.
\end{theorem}

\begin{proof}
Algorithm~\ref{alg:branchDist} strictly resembles the recursion from Lemmas~2-4, where $T_1$ is represented by $(n_1,p_1)$, $T_1'$ by $(c_{1,1},n_1)$, $T_1-T_1'$ by $(c_{1,1},p_1)$, $T_1''$ by $(c_{1,2},n_1)$, $T_1-T_1'$ by $(c_{1,2},p_1)$ and $T_2,T_2',T_2'',T_2-T_2',T_2-T_2''$ symmetrically. This correspondence then also yields the correctness of the initial call $d_B(r_1,r_1',r_2,r_2')$ for $T_1,T_2$.
\end{proof}

\paragraph*{Adaption for Arbitrary Degree.}
For the purpose of readability, we omitted a detailed description of the algorithm for unbounded degree, but we now give a short discussion on which adaptations have to be made.
To adapt the algorithm for trees of arbitrary degree, we have to replace the last four options in line 23 through an optimal matching of the subtrees, similar to the algorithm for the constrained edit distance by Zhang~\cite{DBLP:journals/algorithmica/Zhang96}. For each pair of possible main branches in the two trees, we have to find the optimal matching between the remaining children. For the constrained edit distance, this adds a factor of $(\deg_1+\deg_2) \cdot \log(\deg_1+\deg_2)$, but since we still have to find the optimal pair of main branches, we get another factor of $\deg_1 \cdot \deg_2$. To obtain a more precise upper bound, one would have to study how these costs amortize over the whole tree, since there are obvious limits to the sum of degrees (e.g.\ there can only be a constant number of degree $\Omega(n)$ nodes), but we omit this analysis as merge trees of high degree do usually not appear in realistic scenarios.

\section{Experimental Runtimes}
\label{appendix:applications}

Using a C++ implementation of Algorithm~\ref{alg:branchDist}, we perform benchmarks on the datasets from Section~5. We use different simplification thresholds for the heated cylinder dataset, the synthetic outlier ensemble, and the vortex street dataset. For the latter two, we first add artificial noise. The observed running times can be seen in Table~\ref{tab:runtimes}. For each dataset, we compute the branch mapping distance for 100 random pairs. The table shows the average sizes of the merge trees as well as the average running times over all 100 pairs. We expect that further optimization of our code (e.g. parallelization) could yield practical running times for even larger trees. We leave this for future work.

\begin{table*}[]
\centering
\begin{tabular}{l|ccccccccc}
     & HC (10) & O (20) & VS (68) & VS (135) & O (196) & VS (213) & HC (233) & VS (259) & O (356)  \\ \hline
  Runtime & $1.5\cdot10^{-5}$s & $9.5\cdot10^{-5}$s & $3\cdot10^{-3}$s & $7\cdot10^{-2}$s & $3.7\cdot10^{0}$s & $1.1\cdot10^{0}$s & $6.5\cdot10^{0}$s & $1.5\cdot10^{0}$s & $4.2\cdot10^{1}$s  \\
\end{tabular}
\caption{Running times of the branch mapping distance on the datasets from Section~5: The synthetic outlier dataset (O), the heated cylinder (HC) and the vortex strees (VS). We used different versions of the datasets with different levels of noise and simplification. The sizes of the merge trees are shown in brackets.}
\label{tab:runtimes}
\end{table*}

\section{Applications}
\label{appendix:applications}

We now show the more detailed results on some of the datasets from Section~5.
We start with renderings of the full synthetic ensembles in Figures~\ref{fig:ensemble_fourpeaks} and~\ref{fig:results_outlier_withDendogram}.
Furthermore, in these figures the clustermaps are shown together with dendograms of the underlying hierarchical clusterings.

To show that this example is independent of the chosen parameters (Wasserstein distance and squared costs), we provide the distance matrices for the branch mappings and constrained edit mappings also for the $L_\infty$- and overhang cost from~\cite{DBLP:journals/tvcg/SridharamurthyM20} in Figure~\ref{fig:results_outlier_otherbase}.
Furthermore, Figure~\ref{fig:results_outlier2} shows another ensemble and the distance matrices of the branch mapping and constrained edit distance, where we did not use the squared costs.
We were unable to do these extra comparisons for the Wasserstein distance, as its TTK-implementation does not allow for a base metric change.

Figure~\ref{fig:heatmap_weinkauf_full} shows the complete distance matrix for the dataset from Section~5.2 and Figure~\ref{fig:comparison_vortexStreet} a direct comparison with our implementation of the distance by Sridharamurthy et al.~\cite{DBLP:journals/tvcg/SridharamurthyM20} on the same dataset.

In Figure~\ref{fig:comparison_iondensity}, we provide a comparison of our tracking results on the ion density dataset with the results of the Wasserstein distance on the same dataset provided by Pont et al.\ in~\cite{DBLP:journals/corr/abs-2107-07789}.

We also provide the distance matrices for noisier versions of the outlier ensemble and the vortex street dataset. For the outlier example, we did this on a noisy version of the ensemble from Figure~\ref{fig:results_outlier2} and a new outlier ensemble that was directly created with noise. Figures~\ref{fig:dm_outlier2_noise} and~\ref{fig:dm_outlierwithnoise} show that the outlier is still visible if we add artificial noise to the synthetic ensemble, however, not as clearly. The same holds for the periodic pattern in the vortex street dataset with artificial noise, which can be seen in Figures~\ref{fig:dm_vortexStreet_noise} and~\ref{fig:dm_vortexStreet_noise_wassersteinbase}. Here, the difference to the normal version is more significant. The periodicity detection is possible, but clearly impeded. We should note that an improved robustness against this kind of noise is not the goal of our approach, but rather robustness against \emph{structural} perturbations like in the original synthetic ensemble, which was its motivation. We included these examples with noise to show that our method has \emph{some} robustness agaist this type of noise, comparable to other merge tree-based methods, but it is not significantly improved. In fact, it might even perform worse in some cases due to the generally smaller distances.

\begin{figure}
    \centering
    \includegraphics[width=0.49\linewidth]{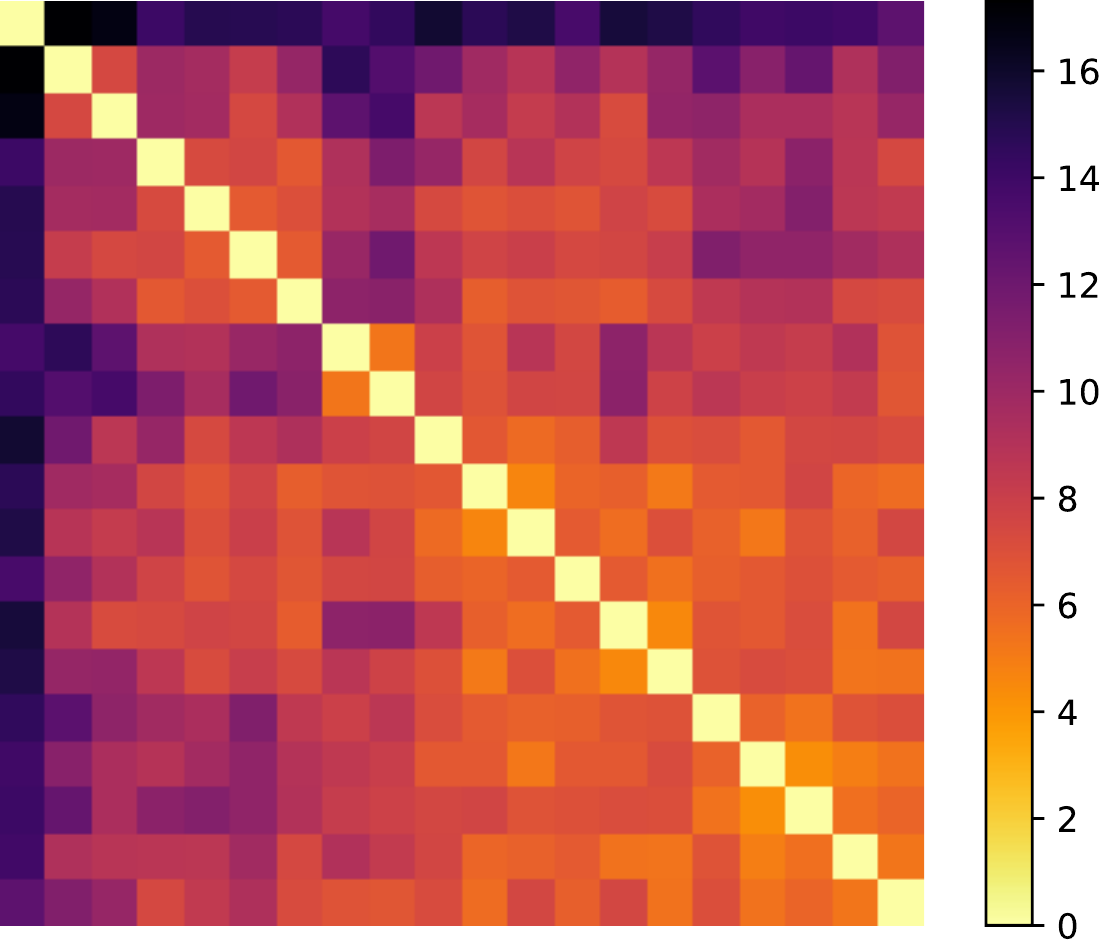}
    \includegraphics[width=0.49\linewidth]{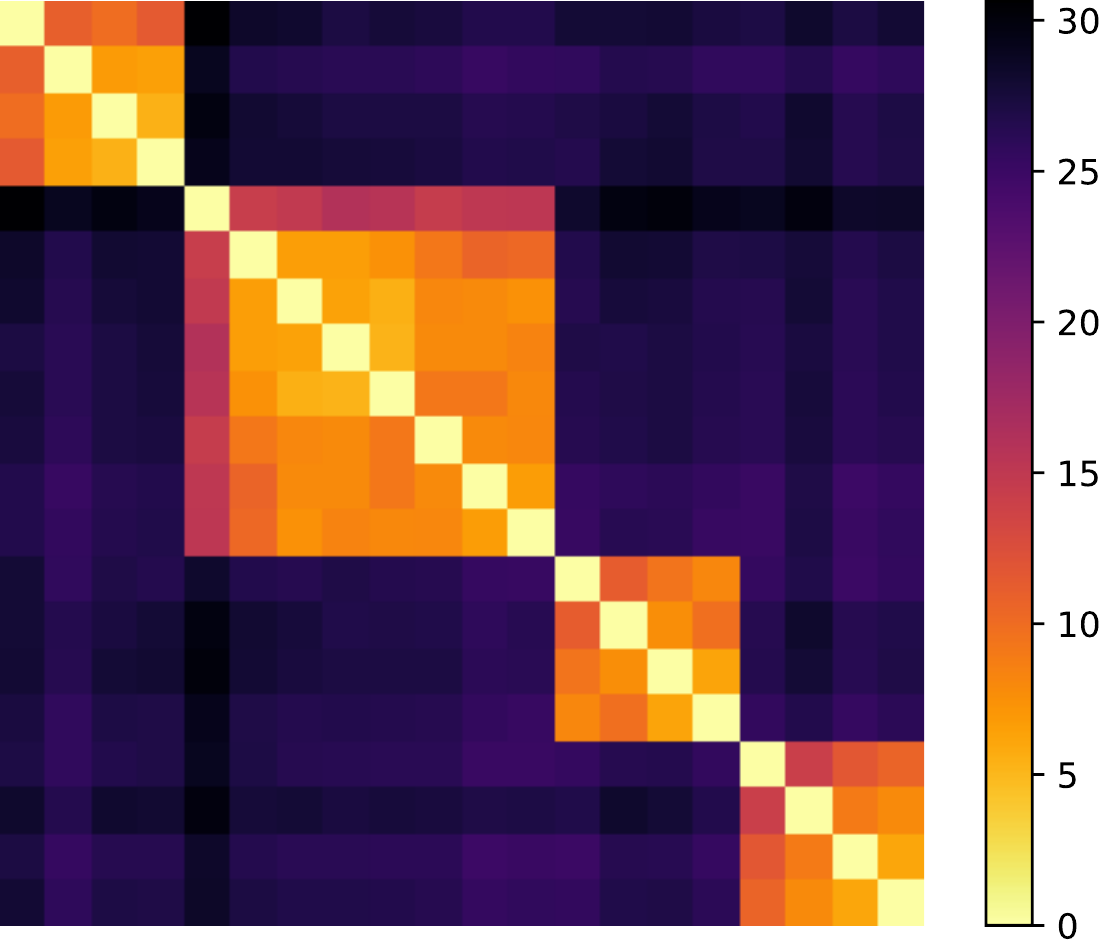}
    \includegraphics[width=0.49\linewidth]{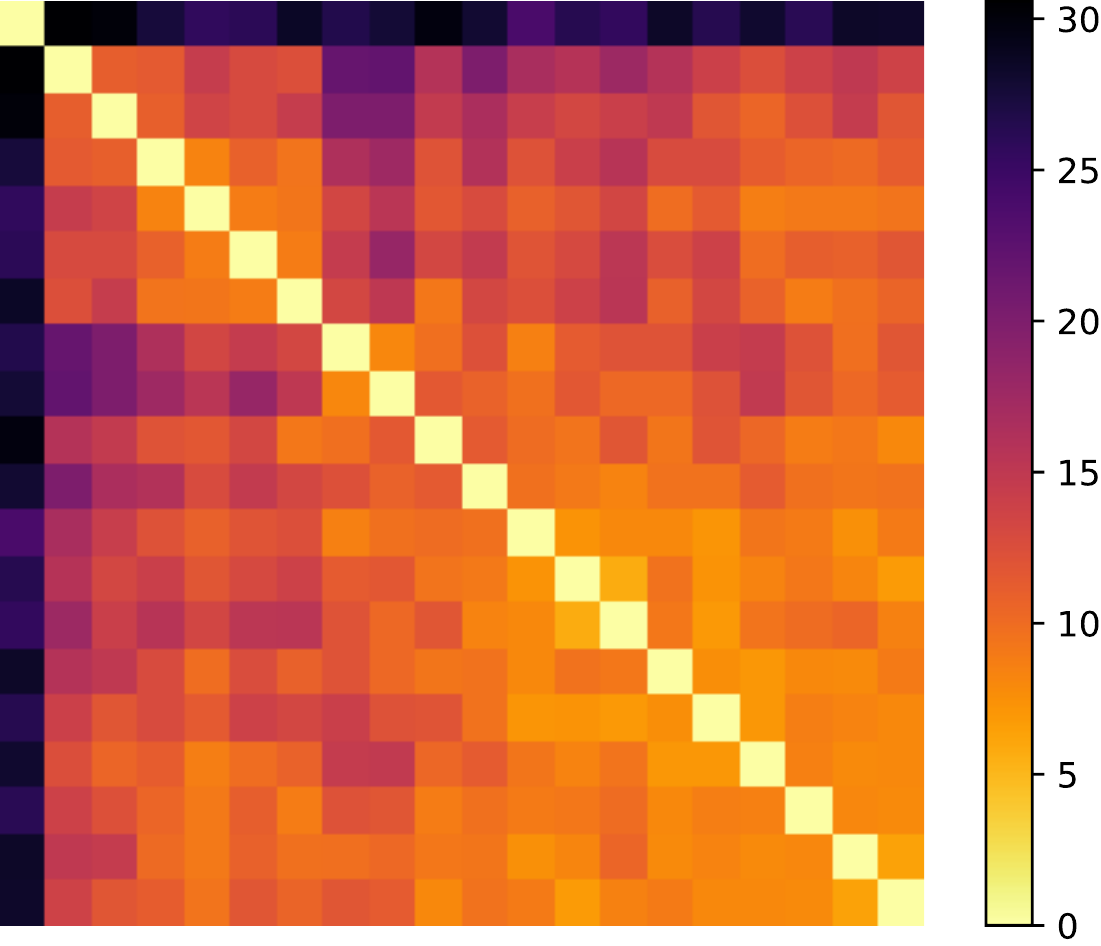}
    \includegraphics[width=0.49\linewidth]{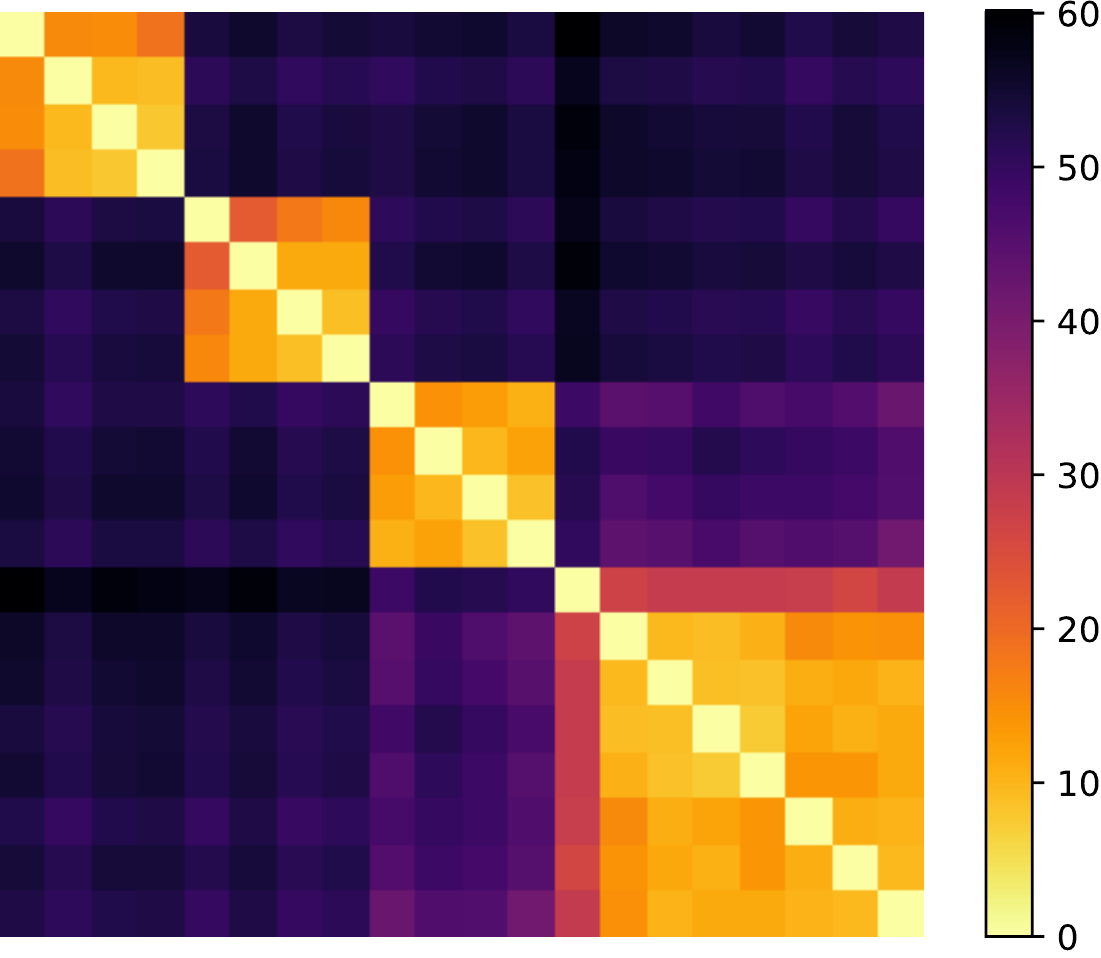}
    \caption{Comparison of branch mapping distance and constrained edit distance on the outlier ensemble using different base metrics. The distance matrices for the branch mapping distance can be seen on the left and for the constrained edit distance on the right. The top matrices represent the distances using the $L_\infty$ base metric and the bottom matrices the distances using the Overhang base metric.}
    \label{fig:results_outlier_otherbase}
\end{figure}

\begin{figure*}
    \centering
    \begin{subfigure}[c]{0.75\linewidth}
    \includegraphics[width=\linewidth]{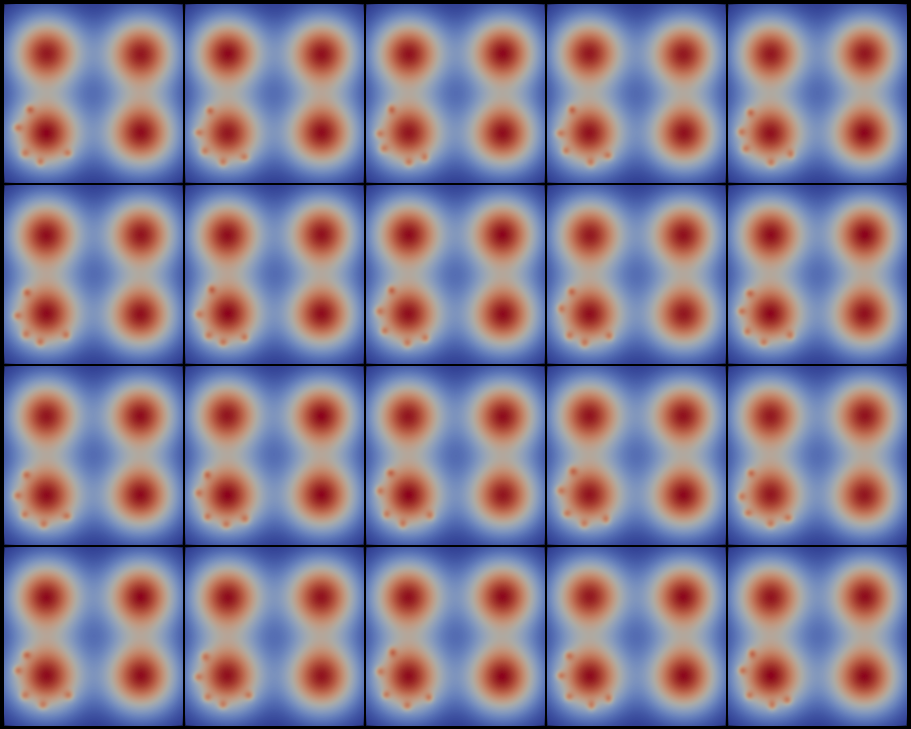}
    \end{subfigure}
    \begin{subfigure}[c]{0.24\linewidth}
        \centering
        \includegraphics[width=\textwidth]{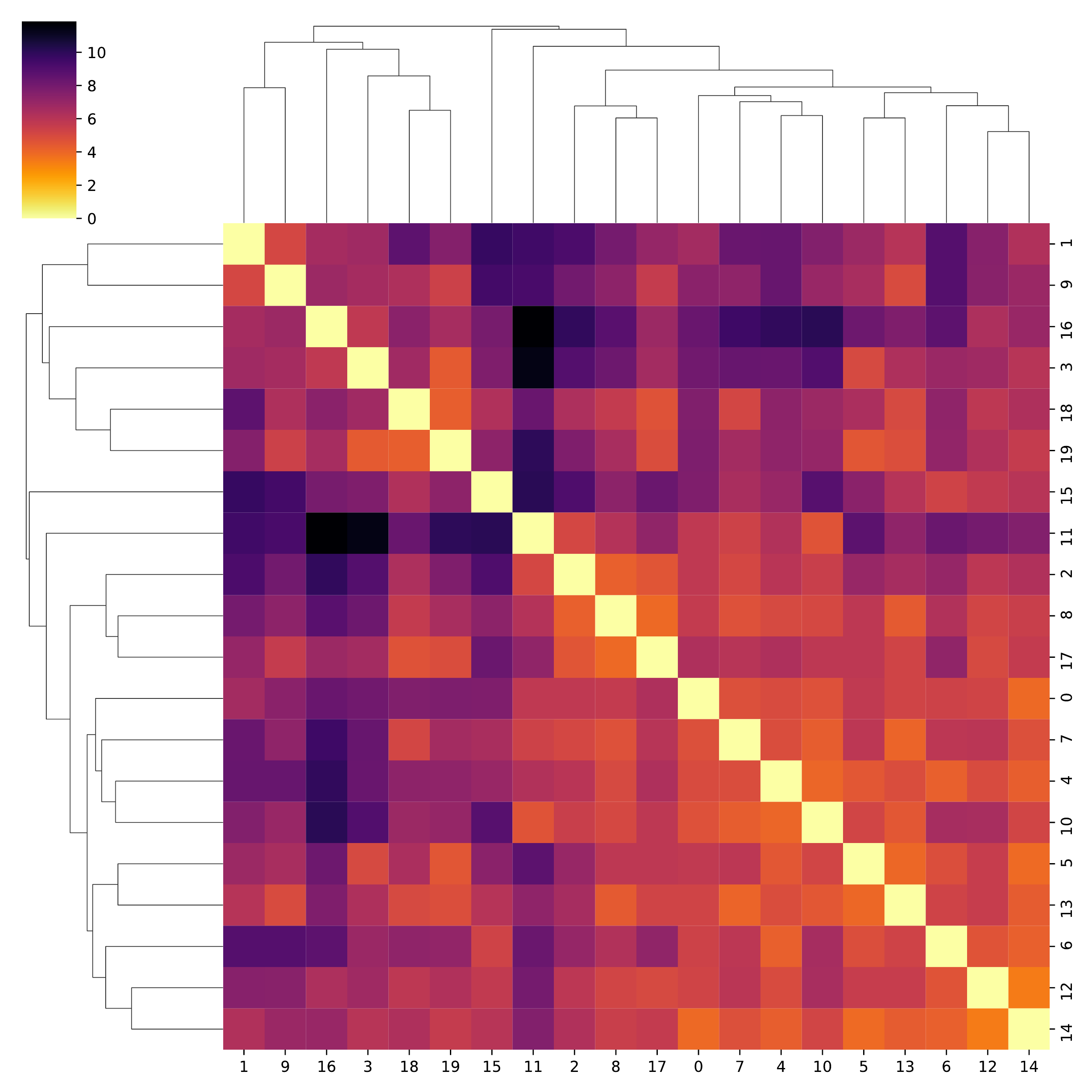}
        
        \includegraphics[width=\textwidth]{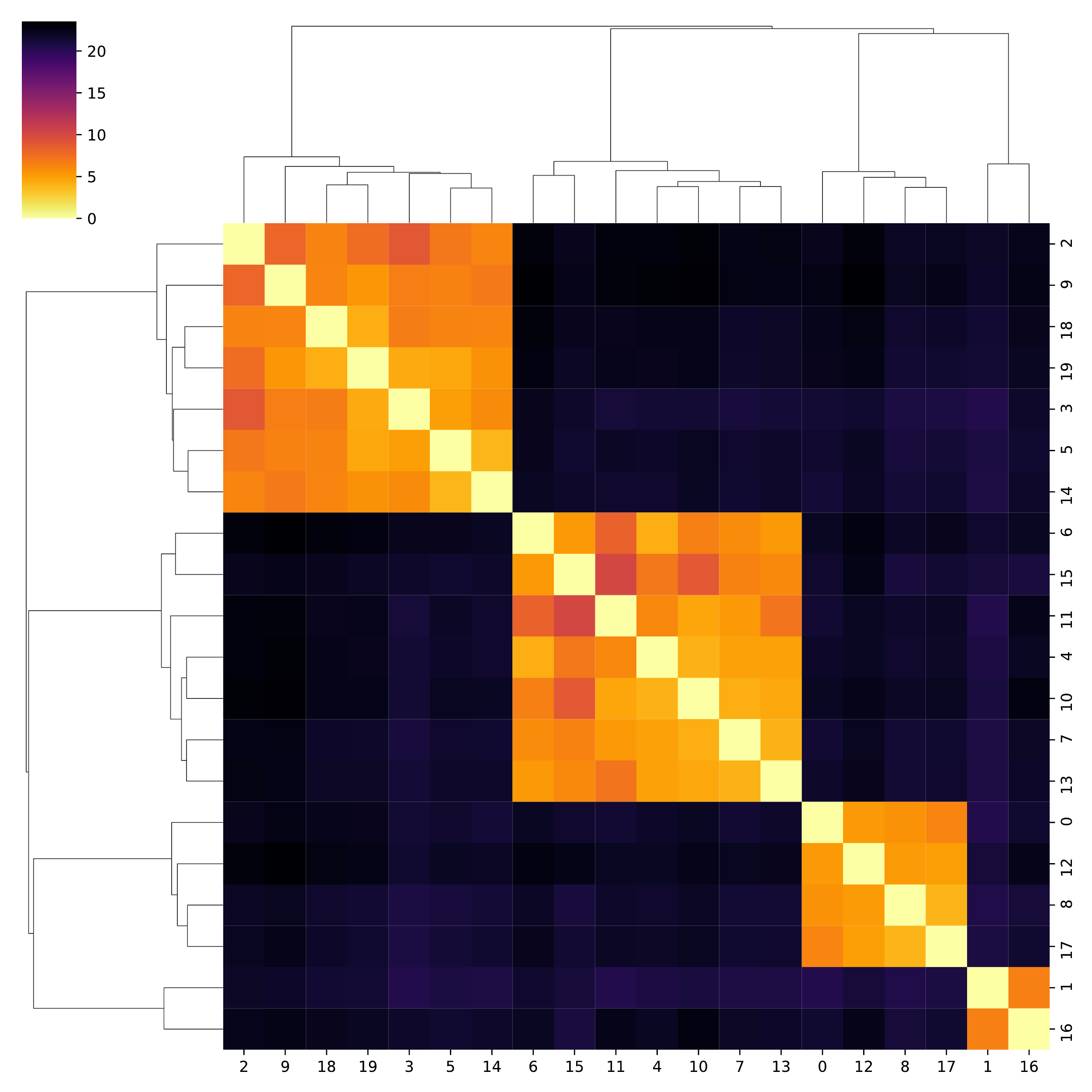}
        
        \includegraphics[width=\textwidth]{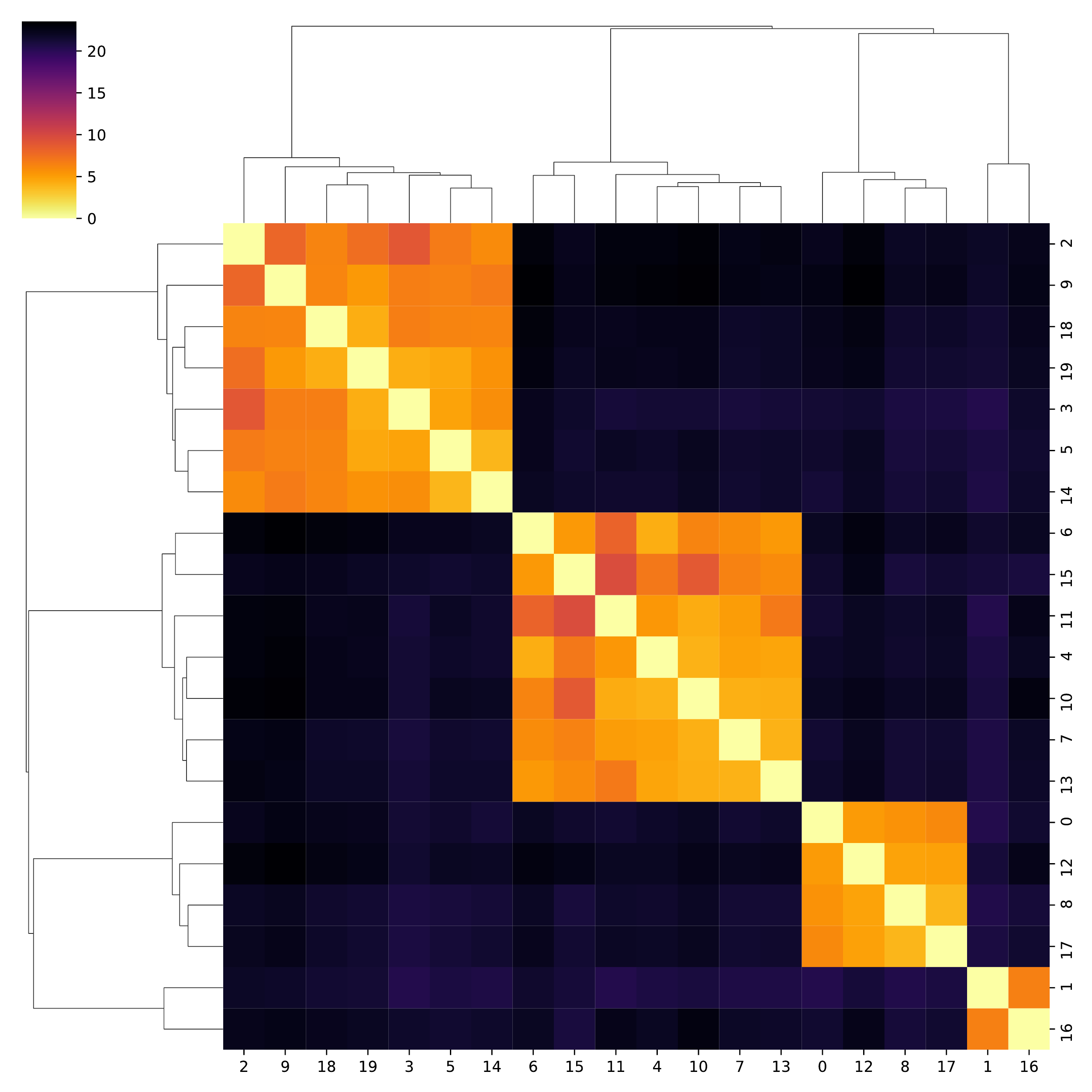}
    \end{subfigure}
    \caption{The example ensemble (left) with a comparison of branch mapping distance (top right), constrained edit distance (center right) and Wasserstein distance (bottom right). Heatsmaps for distance matrices of the three distance measures are shown where the rows and columns are ordered by the clustermap function of the seaborn library to clearly identify clusters.}
    \label{fig:ensemble_fourpeaks}
\end{figure*}

\begin{figure*}
    \centering
    \begin{subfigure}[c]{0.75\linewidth}
    \includegraphics[width=\linewidth]{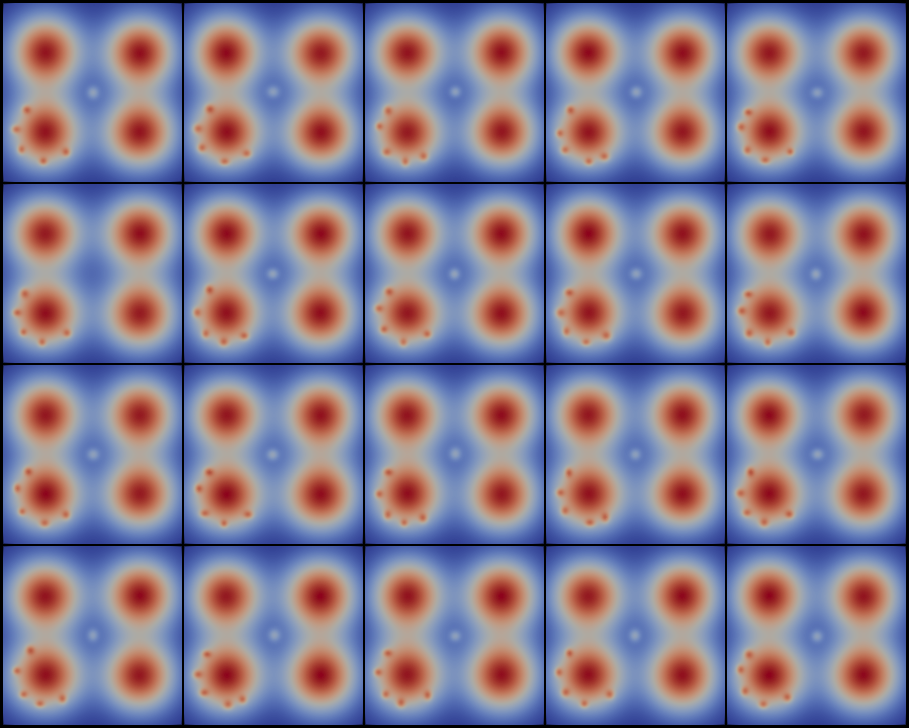}
    \end{subfigure}
    \begin{subfigure}[c]{0.24\linewidth}
        \centering
        \includegraphics[width=\textwidth]{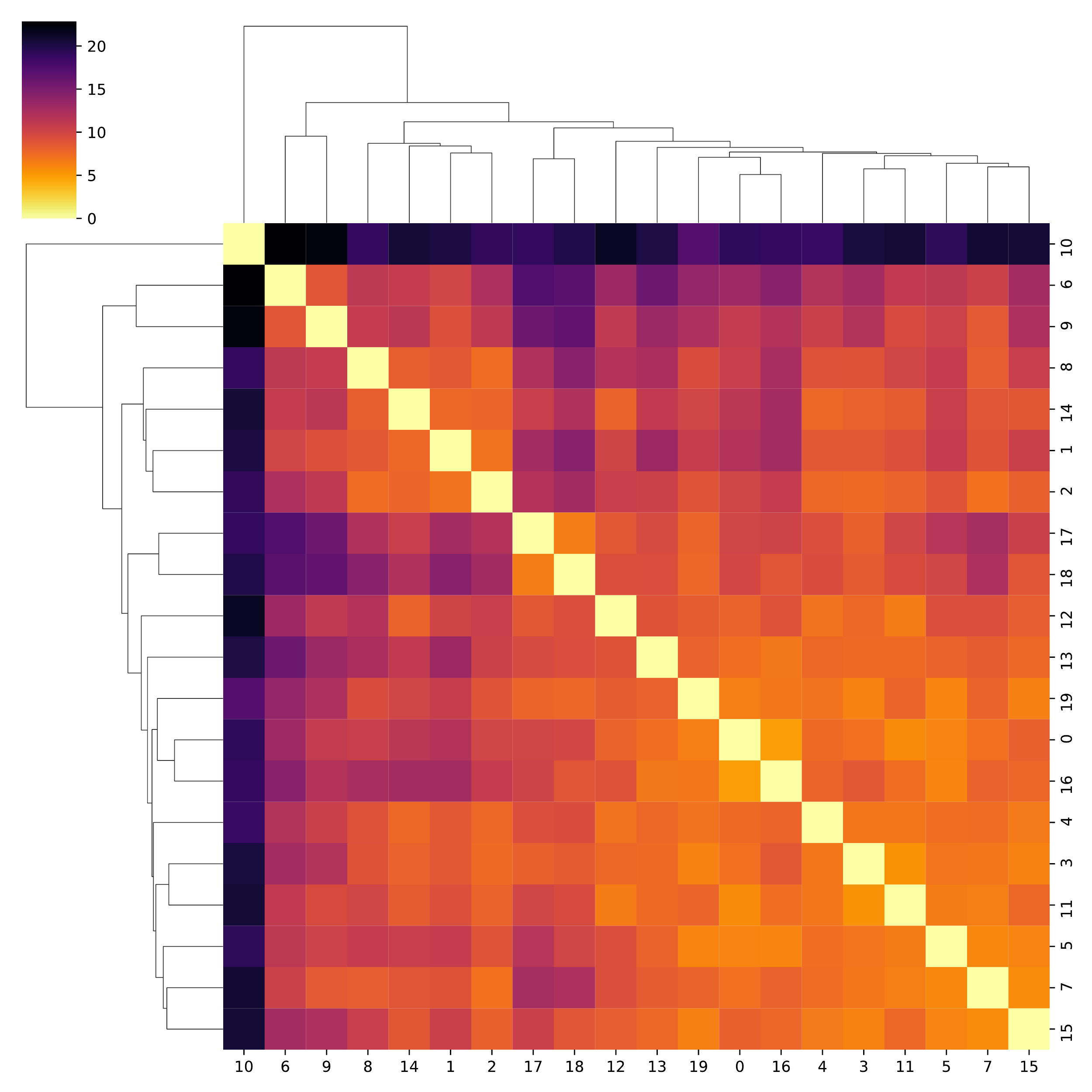}
        
        \includegraphics[width=\textwidth]{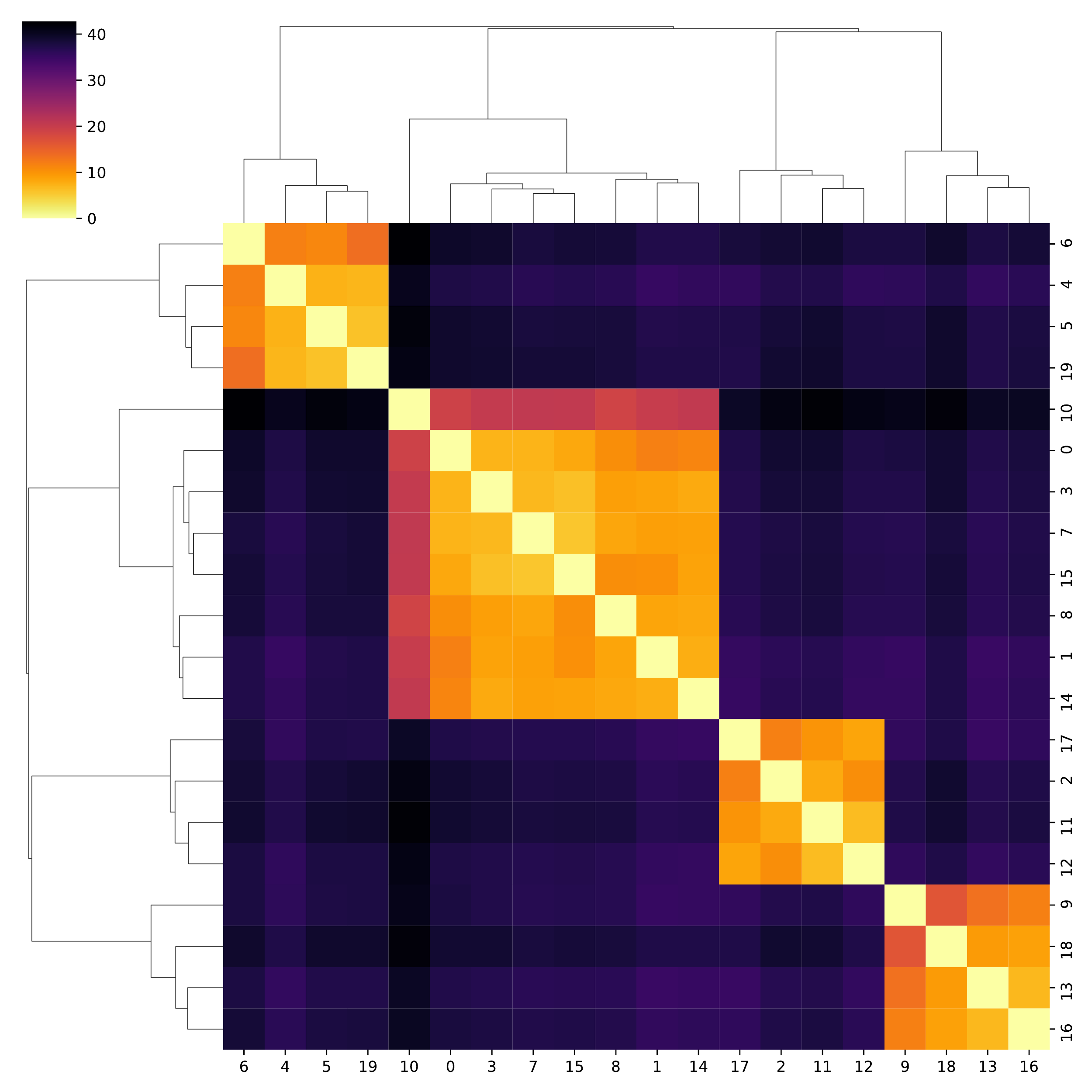}
        
        \includegraphics[width=\textwidth]{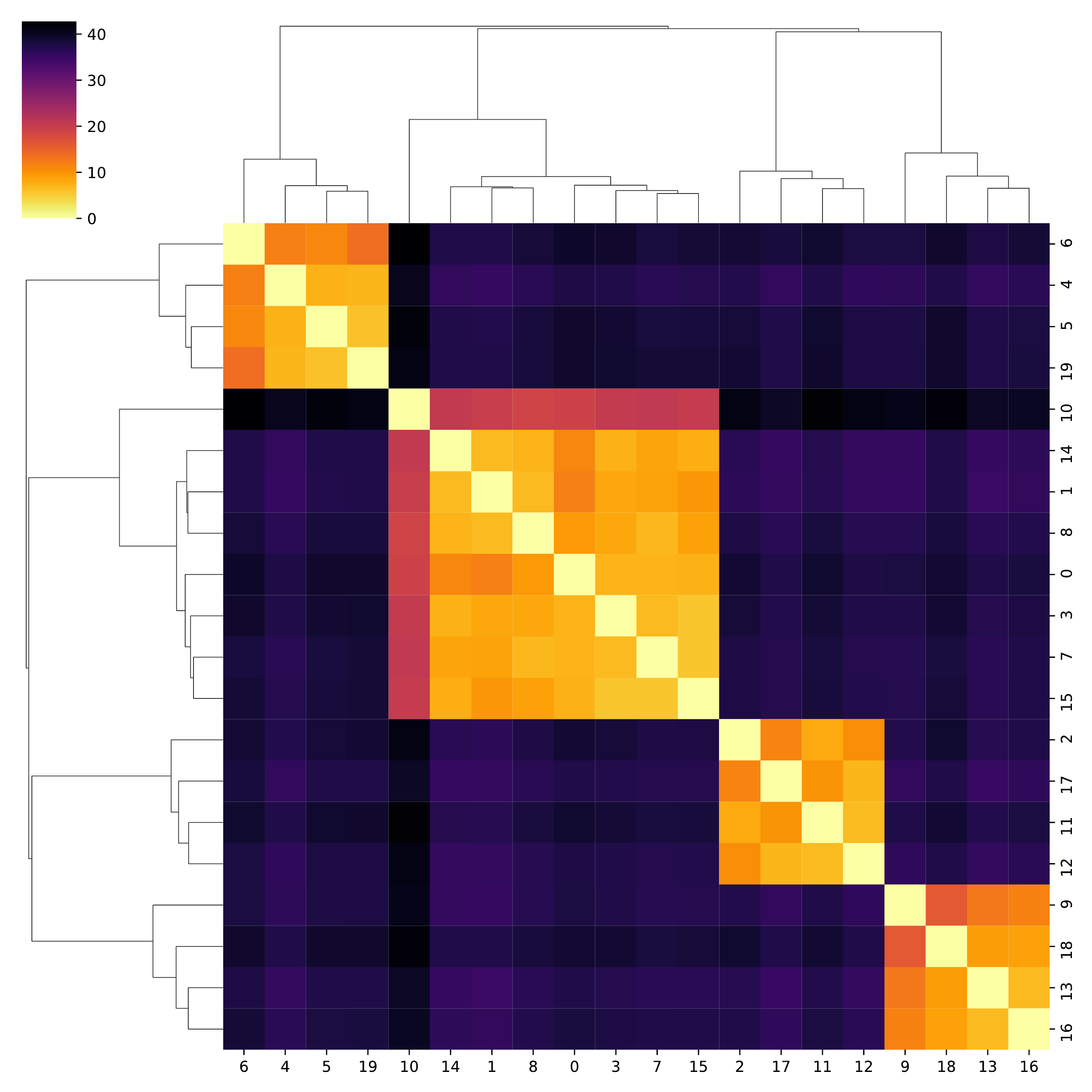}
    \end{subfigure}
    \caption{The outlier ensemble (left) with a comparison of branch mapping distance (top right), constrained edit distance (center right) and Wasserstein distance (bottom right). Heatsmaps for distance matrices of the three distance measures are shown where the rows and columns are ordered by the clustermap function of the seaborn library to clearly identify clusters. The plots also show the dendograms of the underlying clusterings.}
    \label{fig:results_outlier_withDendogram}
\end{figure*}

\begin{figure*}
    \centering
    \begin{subfigure}[c]{0.75\linewidth}
    \includegraphics[width=\linewidth]{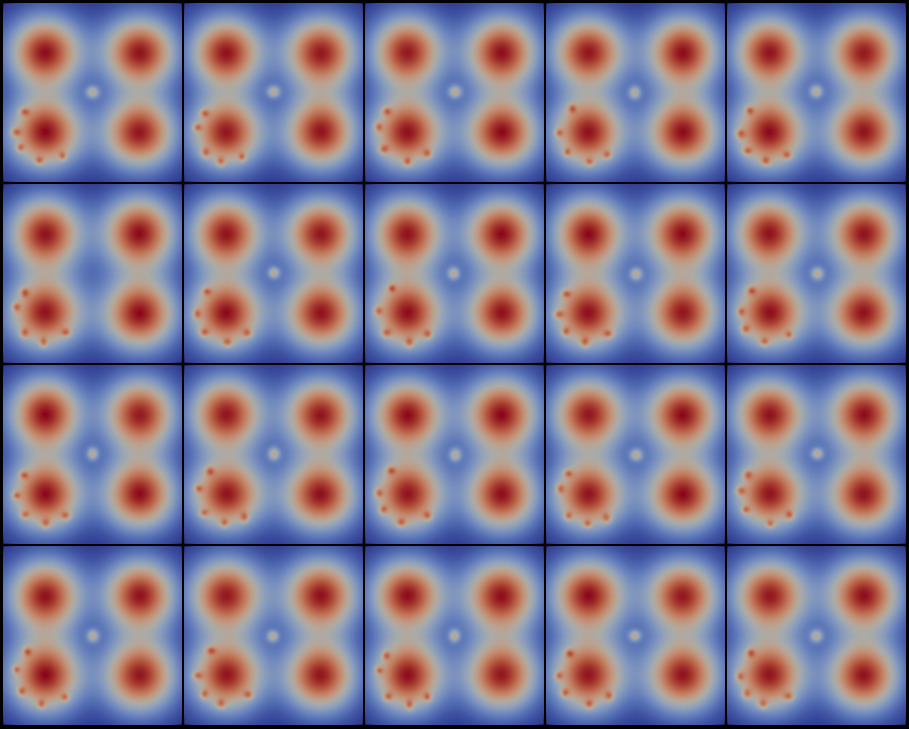}
    \end{subfigure}
    \begin{subfigure}[c]{0.24\linewidth}
        \centering
        \includegraphics[width=\textwidth]{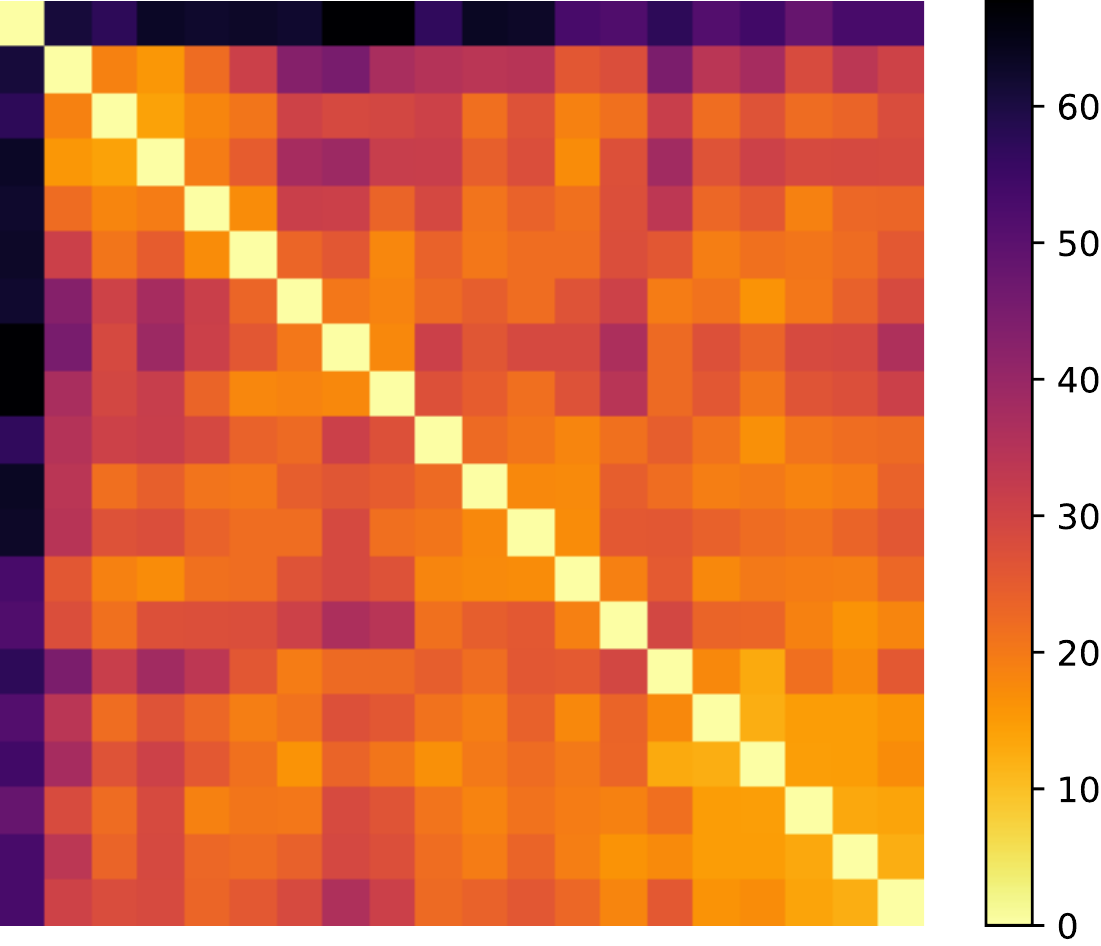}
        
        \includegraphics[width=\textwidth]{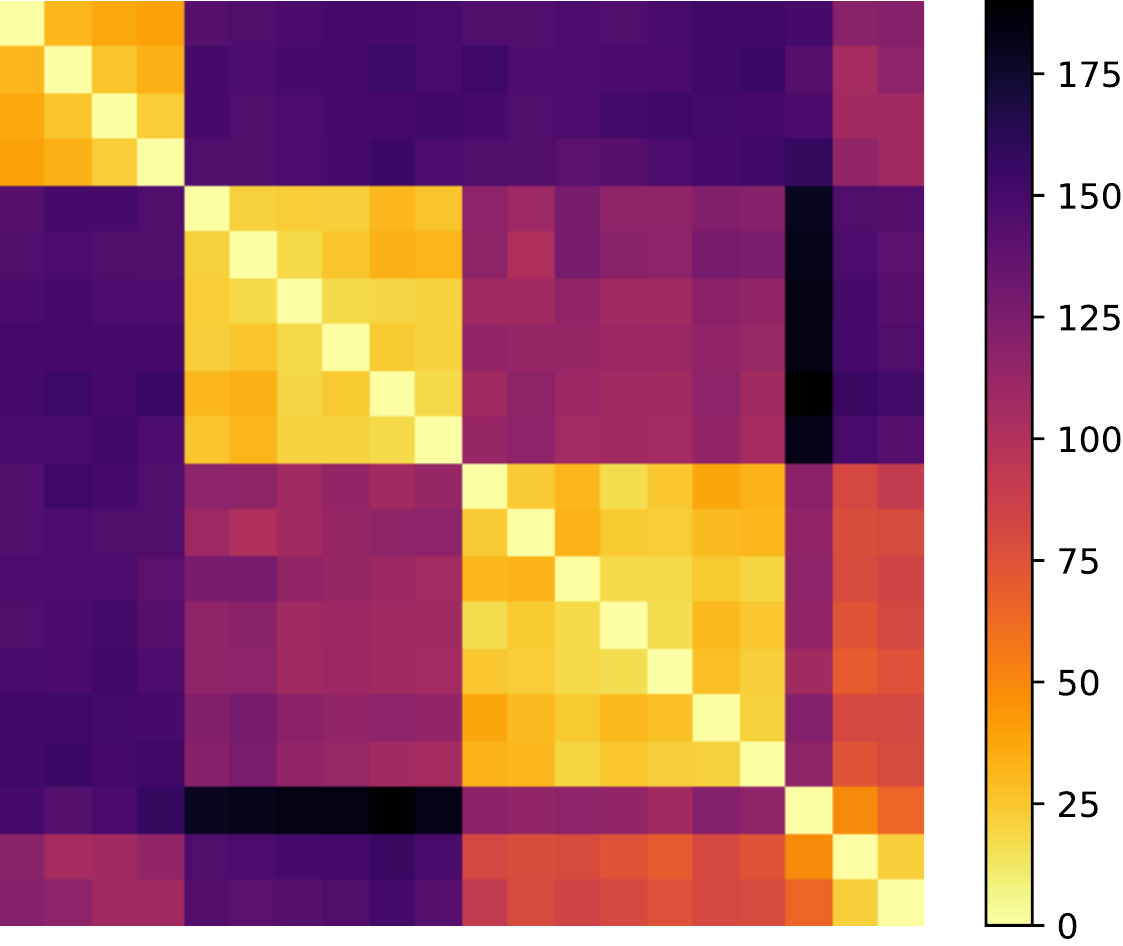}
    \end{subfigure}
    \caption{Another outlier ensemble (left) with a comparison of branch mapping distance (top right) and constrained edit distance (bottom right). Heatmaps for distance matrices of the two distance measures are shown where the rows and columns are ordered by the clustermap function of the Seaborn library to clearly identify clusters.}
    \label{fig:results_outlier2}
\end{figure*}

\begin{figure*}
    \centering
    \includegraphics[width=0.7\linewidth]{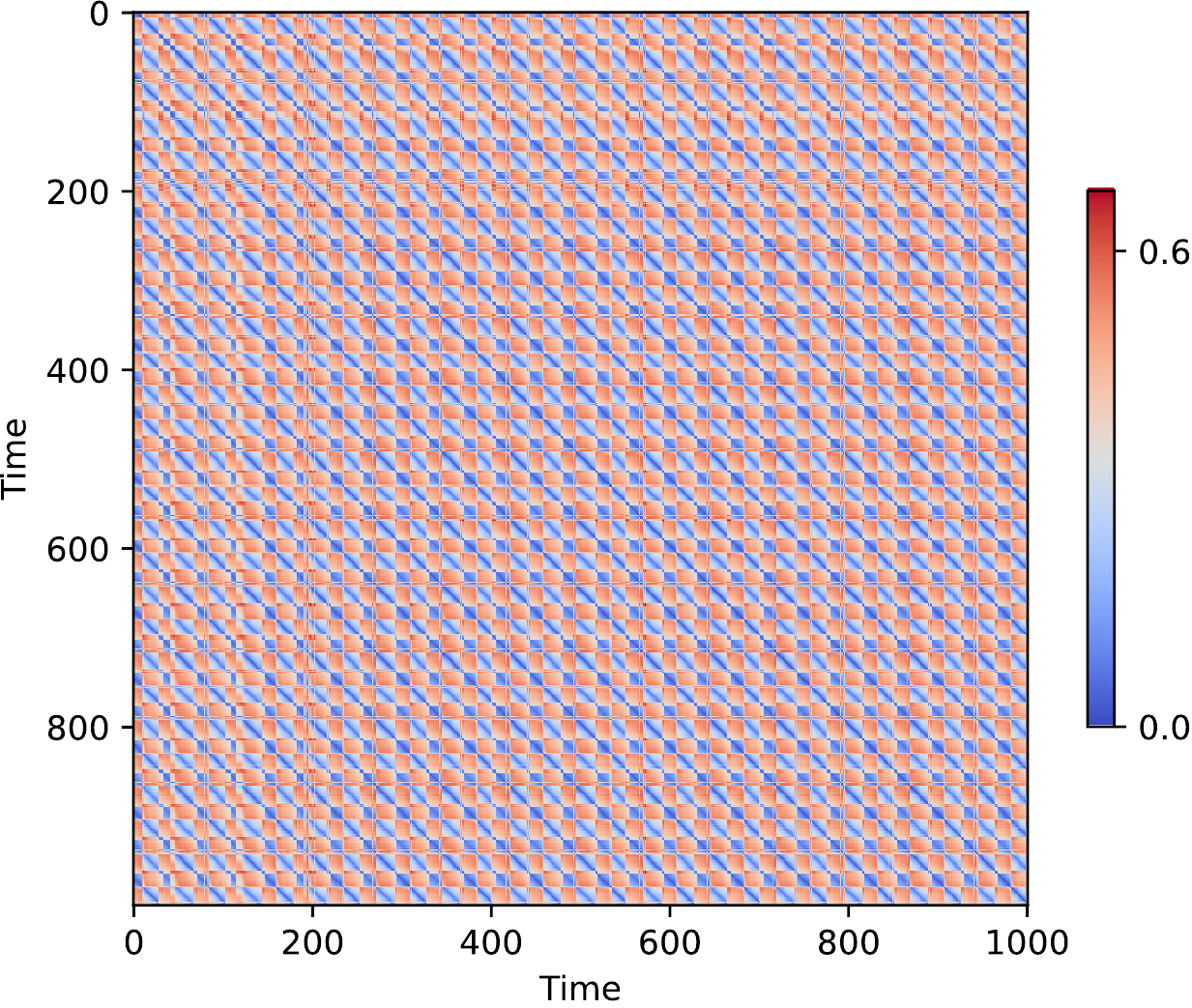}
    \caption{The complete distance matrix for all 1001 time steps of the vortex street dataset.}
    \label{fig:heatmap_weinkauf_full}
\end{figure*}

\begin{figure*}
    \centering
    \begin{subfigure}[c]{0.4\linewidth}
    \includegraphics[width=\linewidth]{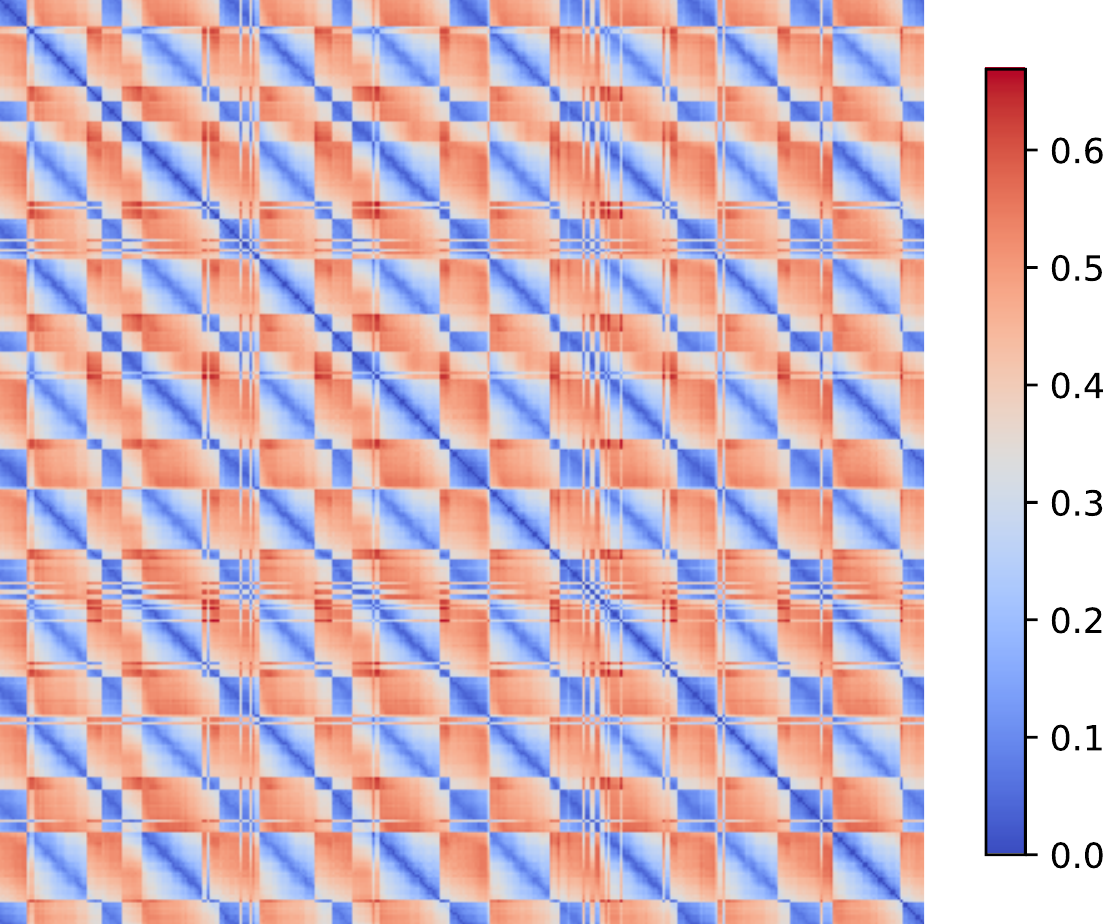}
    \end{subfigure}
    \begin{subfigure}[c]{0.4\linewidth}
        \centering
        \includegraphics[width=\linewidth]{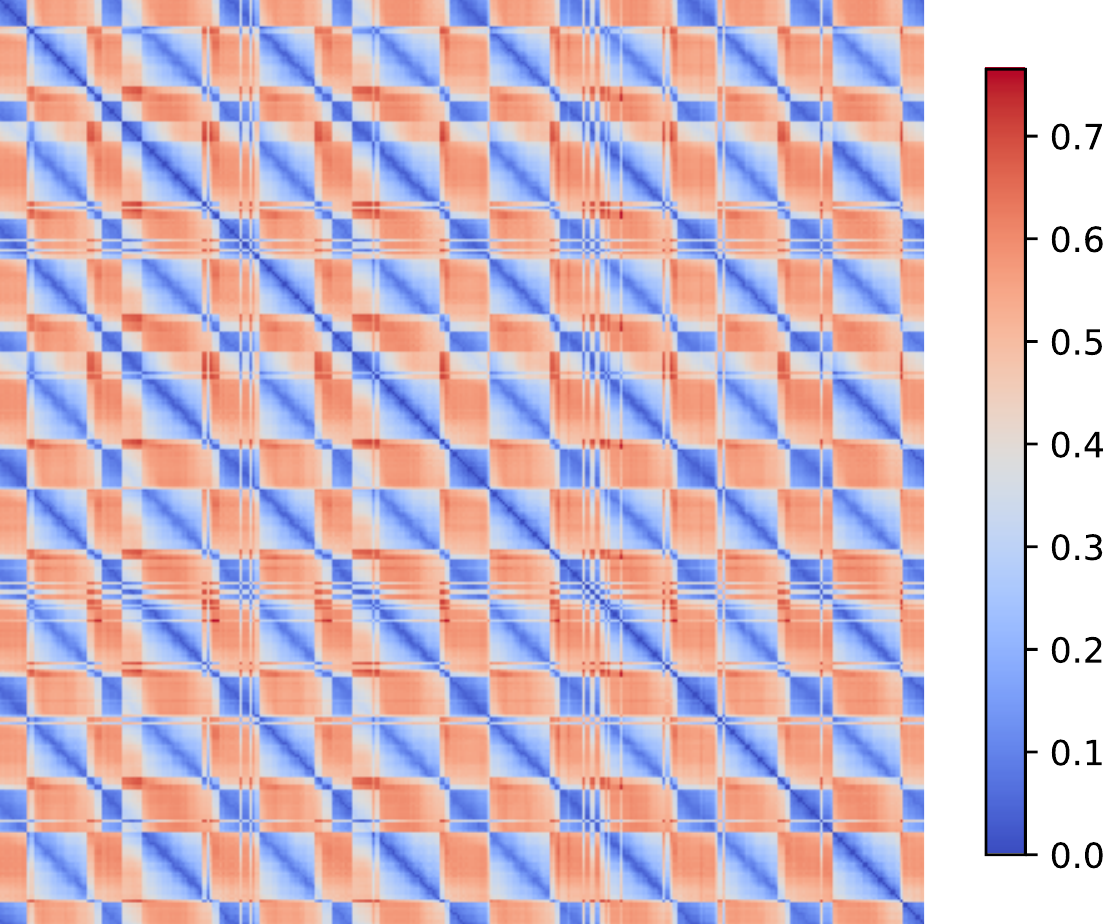}
    \end{subfigure}
    \caption{Comparison of the (partial) distance matrices of the vortex street dataset using the branch mapping distance (left) and the constrained edit distance (right). The right image was computed using our own implementation of the constrained edit distance and resembles Figure~13 in~\cite{DBLP:journals/tvcg/SridharamurthyM20} (small differences can be due to different merge tree computation or different ways of attaching branch properties to vertices). Both, or better to say all three, matrices show the same periodic pattern in the data.}
    \label{fig:comparison_vortexStreet}
\end{figure*}

\begin{figure*}
    \centering
    \begin{subfigure}[c]{0.48\linewidth}
    \includegraphics[width=\linewidth]{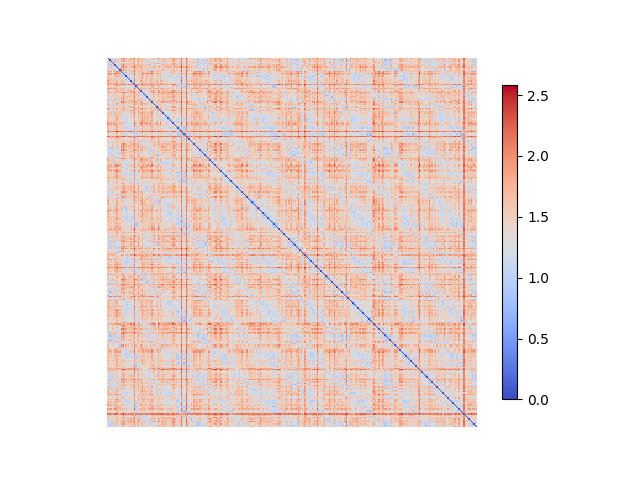}
    \end{subfigure}
    \begin{subfigure}[c]{0.48\linewidth}
        \centering
        \includegraphics[width=\textwidth]{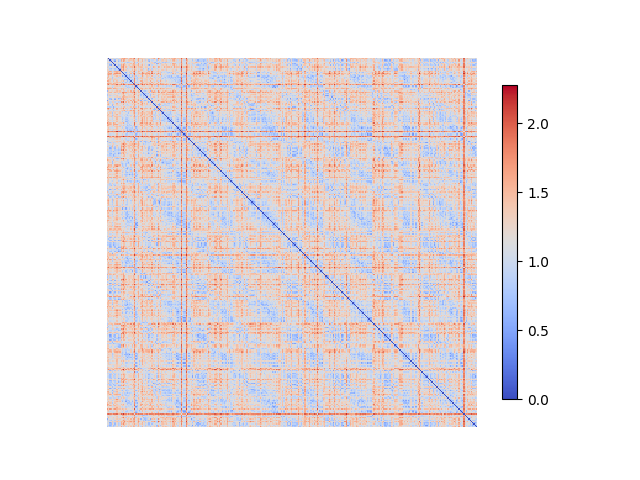}
    \end{subfigure}
    \begin{subfigure}[c]{0.48\linewidth}
    \includegraphics[width=\linewidth]{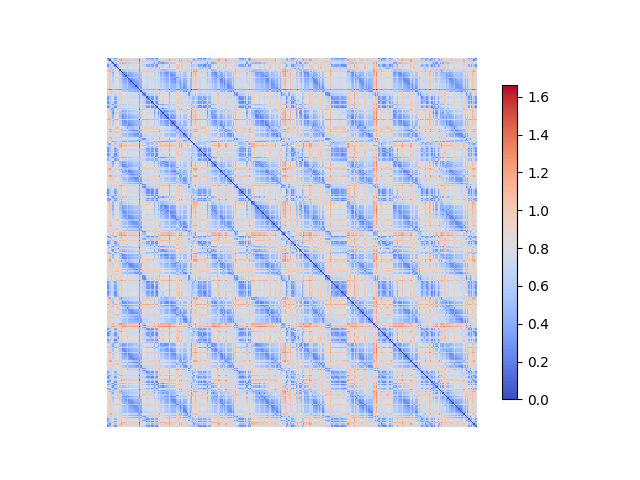}
    \end{subfigure}
    \begin{subfigure}[c]{0.48\linewidth}
        \centering
        \includegraphics[width=\textwidth]{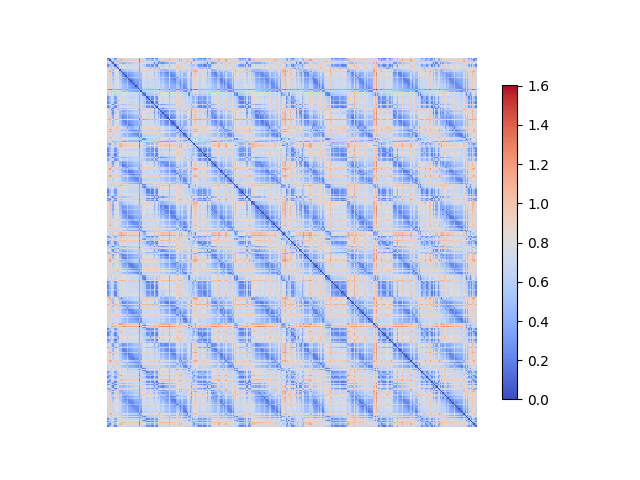}
    \end{subfigure}
    \caption{Distance matrices (using Wasserstein distance as base metric) for noisy versions of the vortex street dataset. Four different variants of noise are shown using noise of different intensity/amplitude (top $1\%$, bottom $0.4\%$ of scalar range) and different simplification (left $\approx200$ vertices, right $\approx150$ vertices, original data $\approx65$ vertices).}
    \label{fig:dm_vortexStreet_noise_wassersteinbase}
\end{figure*}

\begin{figure*}
    \centering
    \begin{subfigure}[c]{0.48\linewidth}
    \includegraphics[width=\linewidth]{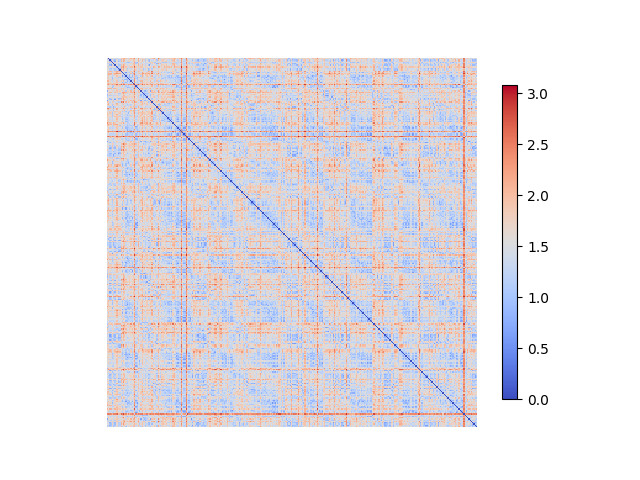}
    \end{subfigure}
    \begin{subfigure}[c]{0.48\linewidth}
        \centering
        \includegraphics[width=\textwidth]{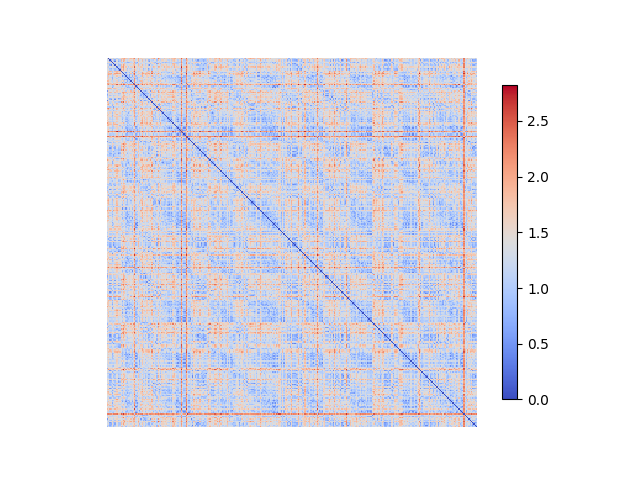}
    \end{subfigure}
    \begin{subfigure}[c]{0.48\linewidth}
    \includegraphics[width=\linewidth]{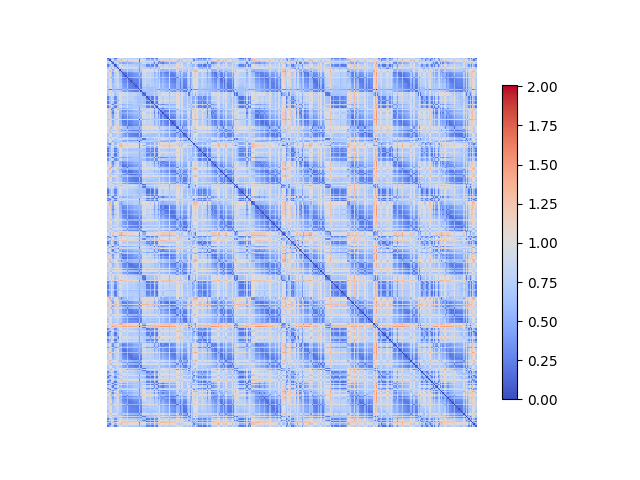}
    \end{subfigure}
    \begin{subfigure}[c]{0.48\linewidth}
        \centering
        \includegraphics[width=\textwidth]{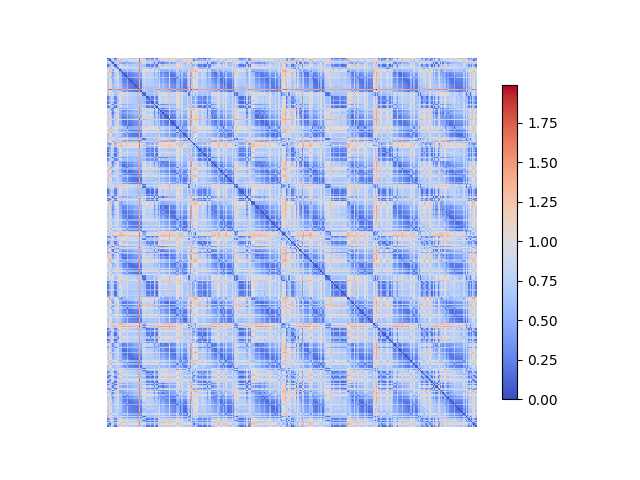}
    \end{subfigure}
    \caption{Distance matrices (using persistence difference as base metric) for noisy versions of the vortex street dataset. Four different variants of noise are shown using noise of different intensity/amplitude (top $1\%$, bottom $0.4\%$ of scalar range) and different simplification (left $\approx200$ vertices, right $\approx150$ vertices, original data $\approx65$ vertices).}
    \label{fig:dm_vortexStreet_noise}
\end{figure*}

\begin{figure*}
    \centering
    \begin{subfigure}[c]{0.48\linewidth}
    \includegraphics[width=\linewidth]{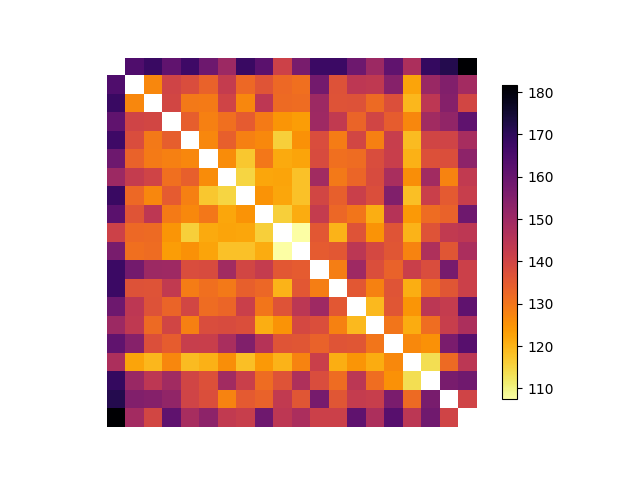}
    \end{subfigure}
    \begin{subfigure}[c]{0.48\linewidth}
        \centering
        \includegraphics[width=\textwidth]{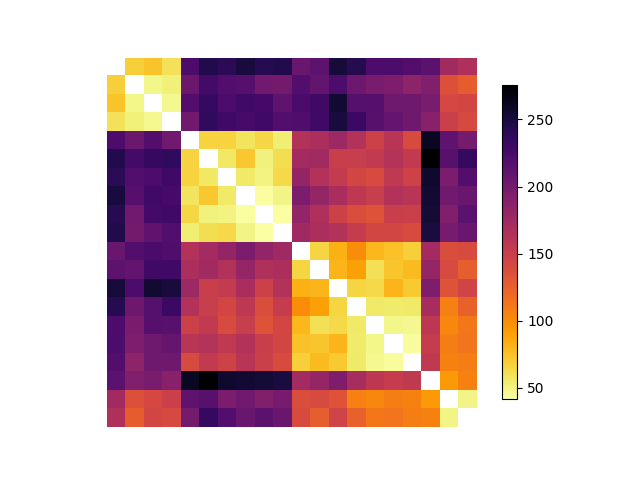}
    \end{subfigure}
    \caption{Distance matrices for noisy versions of the outlier ensemble from Figure~\ref{fig:results_outlier2}. The left matrix was computed using the branch mapping distance, the right one using the constrained edit distance. The merge trees were simplified to $\approx200$ vertices.}
    \label{fig:dm_outlier2_noise}
\end{figure*}

\begin{figure*}
    \centering
    \begin{subfigure}[c]{0.48\linewidth}
    \includegraphics[width=\linewidth]{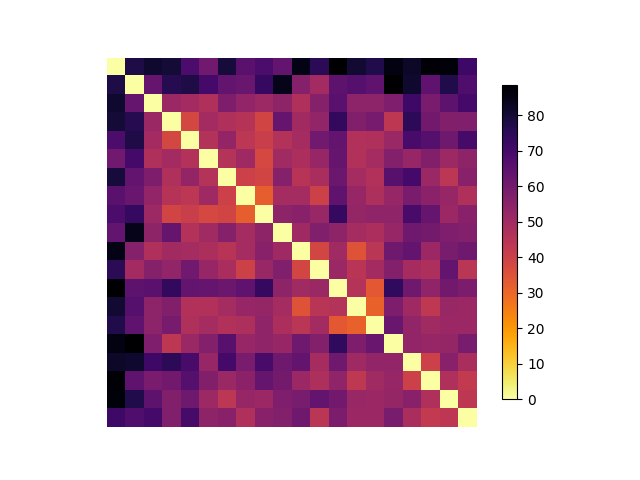}
    \end{subfigure}
    \begin{subfigure}[c]{0.48\linewidth}
        \centering
        \includegraphics[width=\textwidth]{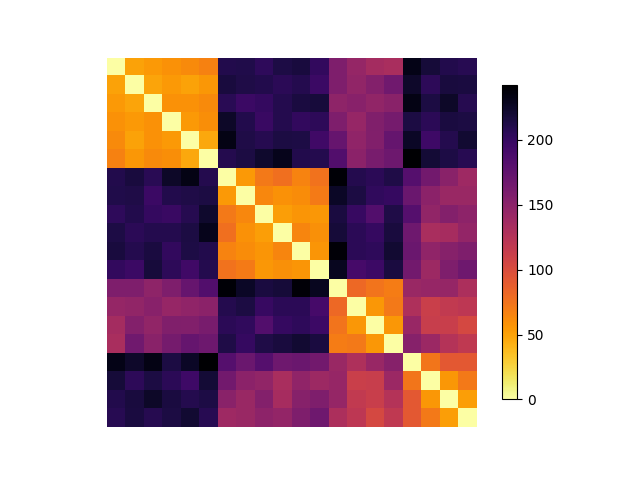}
    \end{subfigure}
    \caption{Distance matrices for a new outlier ensemble with noise. The left matrix was computed using the branch mapping distance, the right one using the constrained edit distance. The merge trees were simplified to $\approx200$ vertices.}
    \label{fig:dm_outlierwithnoise}
\end{figure*}

\begin{figure*}
    \centering
    \begin{subfigure}[c]{0.8\linewidth}
    \includegraphics[width=\linewidth]{tracking_sciviscontest2008-2.jpg}
    \end{subfigure}
    \vspace{3pt}
    \begin{subfigure}[c]{0.8\linewidth}
        \centering
        \includegraphics[width=\textwidth]{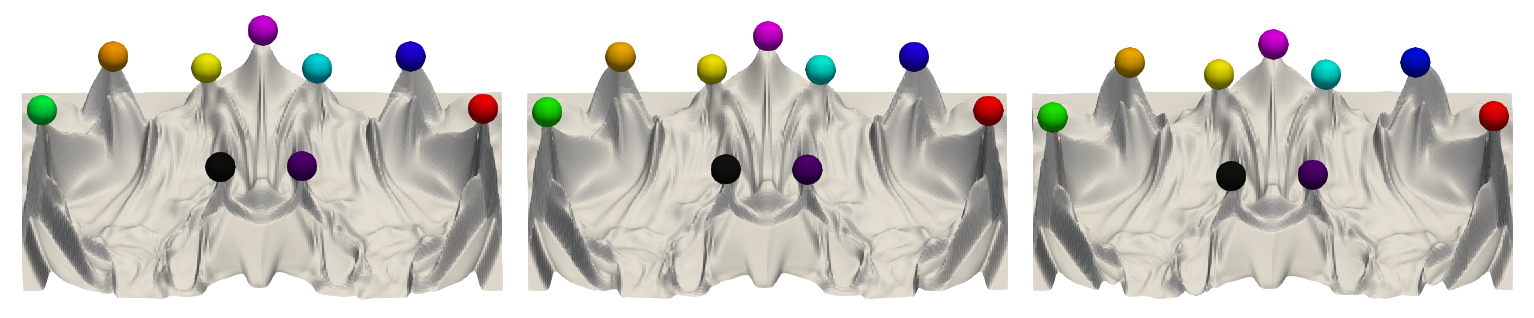}
    \end{subfigure}
    \caption{Comparison of the tracking on the ion density dataset using the branch mapping distance (top) and the Wasserstein distance (bottom). For the bottom image, we replicated the mapping from Figure~11 in~\cite{DBLP:journals/corr/abs-2107-07789} using the TTK implementation of the Wasserstein distance. It is easy to see that both methods yield the same semantically meaningful matching.}
    \label{fig:comparison_iondensity}
\end{figure*}

\bibliographystyle{eg-alpha-doi}
\bibliography{supplement}